\documentclass[10pt,journal,compsoc]{IEEEtran}
\newcounter{mytempeqncnt}
\newtheorem{thm}{Theorem}
\newtheorem{lem}{Lemma}

\newtheorem{crlry}{Corollary}
\newtheorem{defntn}{Definition}

\newenvironment{proof}[1][Proof]{\begin{trivlist}
\item[\hskip \labelsep {\bfseries #1}]}{\end{trivlist}}

\newcommand{\qed}{\nobreak \ifvmode \relax \else
      \ifdim\lastskip<1.5em \hskip-\lastskip
      \hskip1.5em plus0em minus0.5em \fi \nobreak
      \vrule height0.75em width0.5em depth0.25em\fi}

\ifCLASSINFOpdf
  \usepackage[pdftex]{graphicx}
  \graphicspath{{./}{Figs/}}
  \DeclareGraphicsExtensions{.pdf,.jpeg,.png}
\else
\fi
\usepackage{amssymb}
\usepackage{amsmath}
\usepackage{mathtools}
\usepackage{bigdelim}
\usepackage{tikz}
\usepackage{bbm}
\usepackage{mathtools}
\usepackage{centernot}

\usepackage{enumerate}
\usepackage{subcaption}
\usepackage{pgfplots}
\usepackage[caption=false,font=footnotesize]{subfig}
\usepackage{undertilde}
\usepackage{verbatim}
\usetikzlibrary{arrows,calc}
\usepackage{relsize}

\begin{document}
\title{A Micro-foundation of Social Capital in Evolving Social Networks}

\author{Ahmed~M.~Alaa,~\IEEEmembership{Member,~IEEE}, Kartik Ahuja, and Mihaela van der Schaar,~\IEEEmembership{Fellow,~IEEE}
\thanks{The authors are with the Department of Electrical Engineering, University of California Los Angeles (UCLA), Los Angeles, CA, 90095, USA (e-mail: ahmedmalaa@ucla.edu, ahujak@ucla.edu, mihaela@ee.ucla.edu). This work was funded by the Office of Naval Research (ONR).} 
}
\markboth{XXXXXXXXXXXXXXXXX, ~Vol.~XX, No.~X, XXXX~2015}%
{Alaa \MakeLowercase{\textit{et al.}}: A Micro-foundation of Social Capital in Evolving Social Networks}
\maketitle
\begin{abstract}
A social network confers benefits and advantages on individuals (and on groups); the literature refers to these benefits and advantages as \textit{social capital}. An individual's social capital depends on its position in the network and on the shape of the network -- but positions in the network and the shape of the network are determined endogenously and change as the network forms and evolves. This paper presents a micro-founded mathematical model of the evolution of a social network and of the social capital of individuals within the network. The evolution of the network and of social capital are driven by exogenous and endogenous processes -- birth, meeting, linking -- that have both random and deterministic components. These processes are influenced by the extent to which individuals are \textit{homophilic} (prefer others of their own type), \textit{structurally opportunistic} (prefer neighbors of neighbors to strangers), \textit{socially gregarious} (desire more or fewer connections) and by the \textit{distribution of types} in the society. In the analysis, we identify different kinds of social capital: \textit{bonding capital} refers to links to others; \textit{popularity capital} refers to links from others; \textit{bridging capital} refers to connections between others. We show that each form of capital plays a different role and is affected differently by the characteristics of the society. Bonding capital is created by forming a circle of connections; homophily increases bonding capital because it makes this circle of connections more homogeneous. Popularity capital leads to \textit{preferential attachment}: individuals who become popular tend to become more and more popular because others are more likely to link to them.  Homophily creates inequality in the popularity capital attained by different social groups; more gregarious types of agents are more likely to become popular. However, in homophilic societies, individuals who belong to less gregarious, less opportunistic, or major types are likely to be more \textit{central} in the network and thus acquire a bridging capital. And, while extreme homophily maximizes an individual's bonding capital, it also creates \textit{structural holes} in the network, which hinder the exchange of ideas and information across social groups. Such structural holes represent a potential source of bridging capital: non-homophilic (tolerant or open-minded) individuals can fill these holes and broker interactions at the interface between different groups. 
\end{abstract}

\IEEEpeerreviewmaketitle{}
\section{Introduction} 
\IEEEPARstart{S}{ocial} networks bestow benefits -- tangible benefits such as physical and monetary resources and intangible benefits such as social support, solidarity, influence, information, expertise, popularity, companionship and shared activity -- on the individuals and groups who belong to the network. Such resources allow individuals to do better in the network; they help individuals accomplish tasks, produce and spread information, broker interactions across social groups, display influence on other individuals, gain more knowledge, or enjoy more emotional and social support. The concept of \textit{social capital} has come to embody a set of different incarnations of the benefits attained by social groups via networked societal interactions \cite{refcap01}-\cite{refcap05}.\\ 
\\
Contemporary sociologists have established different definitions and conceptualizations for social capital. For instance, Coleman has defined the social capital as \textit{``a function of social structure producing advantage"} \cite{refcap01}, and he advanced social capital as a conceptual tool that puts economic rationality into a social context \cite{refcap01}\cite{refcap02}. Social capital for Bourdieu is related to the size of network and the volume of past accumulated social capital commanded by an individual \cite{ref370}. Bourdieu considers that clear profit is the main reason for an individual to engage in and maintain links in a network, and the individuals' potential for accruing social profit and control of capital are non-uniformly distributed. Both conceptualizations of Coleman and Bourdieu are related; they view social capital as existing in relationships and ties, and they postulate that \textit{density} and \textit{closure} are distinctive advantages of capital. While such vision assumes that \textit{strong ties} (the links between homogeneous and like-minded individuals) are the prominent sources of social capital, other sociologists such as Granovetter, Putnam, and Burt have argued that \textit{weak ties} (the links between diverse and weakly connected network components) are also a source of capital \cite{refcap05}-\cite{ref55sh3}. That is, individuals who can broker connections between otherwise disconnected social groups are more likely to connect non-redundant sources of information, thus promoting for innovation and new ideas \cite{ref55sh}. In \cite{ref55sh}, Burt provided a generalized framework for social capital, viewing bonding capital in connected communities as a source for bridging capital for individuals who connect these communities.\\ 
\\  
As it is for other forms of capital, inequality is displayed in the creation of social capital \cite{refcap00}; that is to say, social capital accrues over time as networks emerge and evolve, and since individuals gain different social positions in the emergent network, capital is not created uniformly across agents; ``better connected" agents possess more capital. While there are a number of somewhat different definitions of social capital in the literature, these definitions share the following set of features. First, social capital is a metaphor about advantage, and it can be thought of as the contextual complement of human capital; it is not depleted by use, but rather depleted by non-use. Second, social capital is a function of the collective social structure, and the social positions of individuals; well connected individuals possess more capital, and well connected networks possess a larger shared value. Finally, the creation of social capital exhibits inequality due to the heterogeneity of norms and behaviors of the different social groups, which reflects on their positions in the network.\\
\\
Motivated by this discussion, this paper aims at establishing the micro-foundations of emerging social capital in an evolving network. In particular, we present a comprehensive mathematical model for dynamic network formation, where agents belonging to heterogeneous social groups take link formation decisions (e.g. ``follow" a user on Twitter or ResearchGate \cite{ref01} \cite{ref02}, ``cite" a paper that is indexed by Google Scholar, etc) which on one hand gives rise to an endogenously formed network, and on the other hand creates social capital for individual agents and groups. We view social capital as: \textit{``any advantage or asset that is accrued by an individual or a social group in an evolving network due to the social position that they hold in the underlying network structure. An advantage can correspond to the extent of popularity, prestige, or centrality of an individual; or the density and quality of an individual's ego network."} In our model, we consider that \textit{homophily}, which is an individual's tendency to connect to similar individuals \cite{ref371}, contextualizes economic rationality, i.e. homophily is what creates the incentives for individuals to connect to each other. However, the way individuals meet, the number of links they form, and the way trust propagates among them is governed by norms and behaviors, which generally vary from one social group to another. We view the different forms of social capital as being emergent by virtue of an evolving network, where the evolution of the network is highly influenced by both the actions of individuals, as well as the norms and behaviors of social groups. Due to the heterogeneity of the norms and behaviors of different social groups, social capital inequality is exhibited, and some groups would collectively acquire more prominent positions in the network than others. In the following subsection, we briefly describe the basic elements of our model.

\subsection{A micro-foundational perspective of network evolution and social capital emergence} 

The central goal of the paper is to study the micro-foundations of different forms of emerging social capital via a mathematical model for network evolution. In our model, networks are formed over time by the actions of boundedly rational agents that join the network and meet other agents via a random process that is highly influenced by the dynamic network structure and the characteristics of the agents themselves. Thus, networks evolve over time as a stochastic process driven by the individual agents, where the formation of \textit{social ties} among agents are in part endogenously determined, as a function of the current network structure itself, and in part exogenously, as a function of the individual characteristics of the agents. Agents have bounded rationality, i.e. they only have information about other agents they meet over time, they are not able to observe the global network structure or reason about links formed by others, and they are \textit{myopic} in the sense that they take linking decisions without taking possible future meetings into account. We focus on the impact of various exogenous parameters that describe the norms and behaviors of heterogeneous social groups, on the endogenously evolving network structure, and consequently on the emerging social capital. Fig. 1 depicts all such exogenous and endogenous parameters. In the following, we provide definitions for the exogenous parameters under study. \\
\textbf{1- Type Distribution:} Agents are heterogeneous as they possess type attributes that designate the social groups to which they belong. A social group is a group of individuals with the same occupation, social class, age, gender, religion, race or ethnicity, and are assumed to follow the same norms and behavior. The experiences of the different interacting social groups in the network are generally not symmetric; thus, social capital is created non-uniformly across them. The type distribution corresponds to the relative population share of different social groups, and represents the fraction of agents of each type in the network. We say that an agent belongs to a \textit{type minority} to qualitatively describe a scenario where the fraction of agents of the corresponding type in the population is small, and we say that an agent belongs to a \textit{type majority} otherwise. \\
\textbf{2- Homophily:} Homophily refers to the tendency of  agents to connect to other similar-type agents; it is widely regarded as a pervasive feature of social networks \cite{ref381}\cite{ref93}\cite{ref90}. We capture the extent to which    an agent is \textit{homophilic} by an \textit{exogenous homophily index}, which we formally define in Section 2. The homophily index can be thought of as the amount of ``intolerance" that a certain type of agents have towards making contacts with other types. It can also represent the ``closed-mindedness" of a social group; low homophilic tendency means that agents are eager to connect and accept views of other social groups, whereas high homophilic tendency means that agents restrict their social ties to only like-minded individuals. \\
\textbf{3- Social Gregariousness:} Some types of agents can be more \textit{sociable} than others, and thus are willing to form more links. Social gregariousness is simply measured by the minimum number of links an agent is willing to make.\\
\textbf{4- Structural Opportunism:} Agents in the network are said to be \textit{opportunistic} if they exploit their contacts to find new contacts; thus, agents are more likely to link with the neighbors of their neighbors if they are opportunistic. Structural opportunism can also be interpreted as the \textit{flow of trust} among individuals; each agent trusts the connections of his neighbors more than he trusts others. Opportunism induces \textit{closure} in the network, i.e. connections of an individual are well connected, which on one hand may be thought of as a source of increasing social support for an individual, and on the other hand it can lead to information redundancy, i.e. all connections of an individual possess similar information since they are well connected among each other. Structural opportunism can also correspond to a property of a behavior-dependent meeting mechanism; for instance, users in Twitter are expected to \textit{retweet} the tweets posted by users they follow, which leads to the followers of followers of a certain user to follow him. Similarly, researchers find new papers to cite by looking at the references of papers that they have already cited. \\

We focus on three different incarnations of social capital that agents gain as the network evolves. These  forms of capital differ in terms of the type of advantage they offer to agents, the way they are created and distributed among agents and social groups, and their dependence on the underlying norms and behaviors of social groups, which are abstracted by the exogenous parameters. We focus on {\it directed networks}, i.e. networks in which ties are formed unilaterally such as Twitter and citation networks. In particular, we focus on the following forms of social capital that emerge in such networks.\\
\textbf{1- Bonding capital:} We define the bonding capital as the aggregate informational and social benefits that an individual draws from its direct neighbors in the network. The bonding capital depends only on an individual's ego network (direct connections), and is invariant to the global network structure as long as the local ego network is preserved. The bonding capital increases if the ego network is more homogeneous; individuals are better off when connecting to other similar individuals. This is because more similar individuals are more likely to provide more social support and more relevant information. Since in our model agents form links driven by homophilic incentives, we measure the bonding capital by the agents' utility functions. This form of capital is close to the definitions of Coleman and Bourdieu \cite{refcap01}-\cite{ref370}.\\
\textbf{2- Popularity capital:} In our model, we consider a directed social network, thus links are formed by an individual and others also form links towards that individual. Individuals gain bonding capital by forming links to others, and they also gain popularity capital by having other individuals form links to them. The popularity capital represents an individual's ability to influence others. That is, an individual's popularity capital allows it to better spread information and ideas in the network, and also to gain support and agreement on the individual's views and opinions. We measure the popularity capital of an individual by simply counting the number of individuals forming links with that individual. \\
\textbf{3- Bridging capital:} Individuals who connect different social groups are able to control the flow of information across those groups and obtain non-redundant information from diverse segregated communities, which allows them to come up with innovations and new ideas \cite{ref55sh3}. Thus, individuals can acquire a \textit{bridging capital} because of their centrality in the network rather than their popularity or the quality of their ego networks. We measure the bridging capital using a graph theoretic centrality measure, namely, the \textit{betweenness centrality}. \\

Examples of bonding capital include the knowledge acquired by citing research papers, information and news obtained from following users on Twitter, etc. Popularity capital includes the number of citations associated with a published paper, the impact factor of a journal, the number of followers of a user on Twitter \cite{reftwit}, etc. Examples of bridging capital include conducting interdisciplinary research, creating cross-cultural memes on Twitter, etc. Bonding capital helps individuals acquire knowledge, information and support, which allows them to accomplish tasks \cite{refgay}, whereas popularity capital can give financial returns (such as research funds for popular scholars), or intellectual influence (such as in the case of citation networks) \cite{ref02}. Finally, bridging capital leads to innovation \cite{ref55sh3}, i.e. innovative interdisciplinary research \cite{refcent7}; cross-cultural creative content generated by internet users \cite{refcent1}; or acquisition of non-redundant information about job opportunities in informal organizational networks \cite{ref55sh}. Fig. 1 depicts the framework of the paper; we focus on four different exogenous parameters, which abstract the norms and behaviors of social groups, and study their impact on the emergence of the three forms of social capital discussed above.    

\subsection{Preview of the results}
The central questions addressed in this paper are: how do bonding, popularity, and bridging forms of capital emerge simultaneously in an evolving network? Which social groups possess which forms of capital? How is the capital accrued by a social group affected by its norms and behavior? We classify our results based on the different forms of capital as follows. \\
\textbf{1- Bonding capital and the egocentric value of networking:} In Section 3, we study the emergence of bonding capital by characterizing the \textit{ego networks} of individual agents in terms of the time needed for an agent to form its ego network, and the types of agents in that network. We show that majority and opportunistic types are more likely to establish their ego networks in a short time period. Moreover, we show that extreme homophilic tendencies for all social groups is a necessary and sufficient condition for maximizing the aggregate bonding capital of the society -- so we show that polarization in a society maximizes bonding capital. However, we also show that polarization in a society leads to ``structural holes'' in the network that is formed, and these structural holes hinder the exchange of information and ideas. \\
\textbf{2- Popularity capital and preferential attachment:} In Section 4, we show that the acquisition of \textit{popularity capital} displays a \textit{preferential attachment} effect due to the individuals' structural opportunism. In other words, the popular individuals get more popular as structural opportunism promotes the propagation of trust and reputation across the network, which endows popular agents with ``reputational advantages" over time. Furthermore, we show that in tolerant (non-homophilic) societies, an individual's age and the collective gregariousness of social groups are the forces that determine an agent's popularity capital, whereas homophily can create asymmetries in the levels of popularity attained by different social groups; more gregarious types of agents have more chances to become popular, whereas the type distribution plays no significant role in the rate of accumulation of popularity capital. \\
\textbf{3- Bridging capital and the strength of weak ties:} In Section 5, we show via simulations that in a homophilic society, individuals who belong to less gregarious, less opportunistic, or major types are likely to be more \textit{central} in the network and thus acquire a bridging capital within their social groups. Moreover, we emphasize the strength of weak ties by showing that when a social group has a different attitude towards homophily compared to all other groups, it ends up being the most central in the network. In particular, we show that the \textit{structural holes} created in extremely homophilic networks represent a potential source of bridging capital for ``open-minded" social groups; non-homophilic individuals can fill these holes and broker interactions at the interface between different groups, which allows them to be the most central agents, even if they are neither the most popular nor  represent a majority type in the network. Furthermore, we show that in extremely non-homophilic societies, homophilic social groups are the most central; that is, despite the absence of cross-group structural holes, homophilic agents reside in the center of the network, acting as an \textit{information hub} or a \textit{dominant coalition}, through which information diffusion is controlled. \\

\subsection{Related works}

To the best of the authors' knowledge, none of the network formation models in literature have studied the emerging social capital associated with endogenously formed networks. Qualitative studies on social capital by contemporary sociologists such as Coleman, Bourdieu, Lin, Putnam, Portes and Granovetter can be found in \cite{refcap01}-\cite{refcap00}, \cite{reffish}-\cite{reftwit}. These studies give qualitative definitions for the social capital in general (not necessarily networked) societies along with some hypotheses about its emergence in different societies, and they support their hypotheses on the basis of historical and experimental evidence. Moreover, empirical and qualitative studies on the social capital in Online Social Networks (OSN) were carried out in \cite{reftwit}, \cite{refcent1} and \cite{refface}. These works have given qualitative insights into the emergence of social capital in OSNs mainly based on data, e.g. the number of followers and followees of a user on Twitter, the frequency of interaction and message exchange among users in Facebook, etc. All these works do not come up with mathematical models for the emerging social capital in evolving social networks, thus they neither offer a concrete understanding and explanation for the micro-foundations of social capital, nor offer a counterfactual analysis for different scenarios of network evolution.

While no mathematical model has studied emergent social capital in networks, there exists a voluminous literature focusing on network formation models. Previous works on network formation can be divided into three categories: networks formed based on \textit{random events} \cite{ref39}, \cite{ref24},\cite{ref9}, \cite{ref3}-\cite{ref13}, \cite{ref555}, networks formed based on \textit{strategic decisions} \cite{ref91}-\cite{ref16}, \cite{ref55sh2}, and empirical models distilled by mining networks' data \cite{ref01}-\cite{ref03}, \cite{ref19}, \cite{ref31}-\cite{ref361}, \cite{ref575}. While a fairly large    literature has been devoted to developing mathematical models for network formation, a much smaller literature attempts to interpret and understand how networks evolve over time, how  individual agents affect the characteristics of such networks, and  the ``value" of social networking conceptualized in terms of social capital. Probabilistic models based on random events are generative models that are concerned with constructing networks that mimic real-world social networks. In \cite{ref3}-\cite{reftnse2}, agents get connected in a pure probabilistic manner in order to realize some degree distribution \cite{ref3}, or according to a \textit{preferential attachment} rule \cite{ref4}\cite{ref5}. While such models can capture the basic structural properties of social networks, they fail to explain why and how such properties emerge over time. 

In contrast, strategic network formation models such as those in \cite{ref91}-\cite{ref17}, and our previous works in \cite{ref25}\cite{ref16}, can offer an explanation for why certain network topologies emerge as an equilibrium of a network formation game. However, these results are limited to studying network \textit{stability} and \textit{efficiency}, and provide only very limited insight into the dynamics and evolution of networks. Moreover, although mining empirical data can help in building algorithms for detecting communities \cite{ref20}-\cite{ref361}, predicting agents' popularity \cite{ref36}, or identifying agents in a network \cite{ref31}, it is of limited use in  understanding how networks form and evolve. 

\section{Model}
\subsection{Network model}
We construct a  model for a growing and evolving social network. Time is discrete. One \textit{agent} is born at each moment of time; we index agents by their birth dates $i\in\{1,2,.\,.\,.,t,.\,.\,.\}$. Agents who are alive at a given date $t$ have the opportunity to form (directed) links; we write $G^{t}$ for the network that has been formed (by birth and linking) at time $t$. As we will see, this is a random process $\{G^{t}\}_{t=0}^{\infty}.$ We write  $\mathcal{G}^{t}$ for the space of all possible networks that might emerge at time $t$ and $\Omega_{\mathcal{G}}$ for the space of all possible realizations of the network process. At date $t\in\mathbb{N}$, a snapshot of the network is captured by a \textit{step graph} $G^{t}=(\mathcal{V}^{t},\mathcal{E}^{t})$, where $\mathcal{V}^{t}$ is the set of nodes, $\mathcal{E}^{t}=\{e_{1}^{t},e_{2}^{t},.\,.\,.,e_{|\mathcal{E}^{t}|}^{t}\}$ is the set of edges between different nodes, with each edge $e_{k}^{t}$ being an ordered pair of nodes $e_{k}^{t}=(i,j)$ $\left( i\neq j,\,\mbox{and}\, i,j\in\mathcal{V}^{t}\right)$, and $|\mathcal{E}^{t}|$ is the number of distinct edges in the graph. We emphasize that $G^{t}$ is a directed graph. Nodes correspond to agents (social actors) and edges correspond to directed links (social ties) between the agents. The adjacency matrix of $G^{t}$ is denoted by ${\bf A}_{G}^{t}=[A^{t}(i,j)],A^{t}(i,j)\in\{0,1\},A^{t}(i,i)=0,\forall i,j\in\mathcal{V}^{t}$. An entry of the adjacency matrix $A^{t}(i,j)=1$ if $(i,j)\in\mathcal{E}_{k}^{t}$, and $A^{t}(i,j)=0$ otherwise. If $A^{t}(i,j)=1$, then agent $i$ has formed a link with agent $j$, and we say that $j$ is a ``followee" of $i$, and $i$ is a ``follower" of $j$. The directed nature of a link indicates the agent forming the link, and only this agent obtains the \textit{social benefit} of linking and pays the link cost. The \textit{indegree} of agent $i$ is the number of links that are initiated towards $i$, denoted by $\mbox{deg}_{i}^{-}(t)$, while the \textit{outdegree}, denoted by $\mbox{deg}_{i}^{+}(t)$, is the number of links initiated by agent $i$. Agents $i$ and $j$ are connected if there is a path of edges from $i$ to $j$ (ignoring directions); a \textit{component} is a maximal connected set of agents. A \textit{singleton component} is a component comprising one agent. The number of non-singleton components of a step graph $G^{t}$ is denoted by $\omega\left(G^{t}\right)$, where $1\leq\omega\left(G^{t}\right)\leq\left|\mathcal{V}^{t}\right|$.

Each agent $i$ is described by a type attribute $\theta_{i}$, which belongs to a finite set of types $\theta_{i}\in\Theta,\Theta=\{1,2,3,.\,.\,.,|\Theta|\}$, where $|\Theta|$ is the number of types. The type of an agent abstracts the social group to which it belongs; and all agents belonging to the same social group have the same characteristics and will follow the same behavior. The set of type-$k$ agents at time $t$ is denoted by $\mathcal{V}_{k}^{t}$, where $\mathcal{V}^{t}=\bigcup_{k=1}^{|\Theta|}\mathcal{V}_{k}^{t}$, and $\mathcal{V}_{k}^{t} \bigcap\mathcal{V}_{m}^{t}=\emptyset,\forall k,m\in\Theta,k\neq m$. We define the length-$L$ \textit{ego network} of agent $i$ at time $t$, $G_{i,L}^{t}$, as the subgraph of $G^{t}$ induced by node $i$, and any node $j$ that can be reached via a directed path of length less than or equal to $L$ starting from node $i$. In this paper, an ``ego network" generally refers to the length-1 ego network of an agent.

There are three aspects of network formation: agents are born; agents meet; agents form links. Birth is governed by a stationary random process; meeting is governed by a non-stationary random process; linking is governed by active choices. We describe each of these processes in the following subsections.
\subsection{The Birth Process}
At time $0$ the network is empty ($G^{0}=\emptyset$). Agents are born one at a time at each date $t$ according to a stationary stochastic process $\lambda(t)=\{\theta_{t}\}_{t\in\mathbb{N}}$, with a sample space $\Lambda=\Theta^{\mathbb{N}}$, i.e. $\Lambda=\left\{\left(\theta_{1},\theta_{2},.\,.\,.\right): \theta_{t}\in\Theta,\,\forall t\in\mathbb{N}\right\}$. We assume that the types of agents are independent and identically distributed ($\theta_{i}$ and $\theta_{j}$ are independent for all $i \neq j$), and that the agents' \textit{type distribution} is $\mathbb{P}(\theta_{i}=k)=p_{k}$, where $\sum_{k\in\Theta}p_{k}=1$, so, $\lambda(t)$ is a \textit{Bernoulli scheme}. At date $t$, the expected number of type-$k$ agents in the network is $p_{k}t$, the total number of agents is $t$, i.e. $|\mathcal{V}^{t}|=t,$ and $\lim_{t\rightarrow\infty}\frac{|\mathcal{V}_{k}^{t}|}{|\mathcal{V}^{t}|}=p_{k}$. Using \textit{Borel's law of large numbers}, we know that 
\[
\mathbb{P}\left(\lim_{t\rightarrow\infty}\frac{1}{t}\left|\mathcal{V}_{k}^{t}\right|=p_{k}\right) = 1.
\] 
In other words, for a sufficiently large network size (and age $t$), the actual fraction of agents of each type in the network converges almost surely to the prior type distribution of the Bernoulli scheme.

\subsection{The Meeting Process}

At each moment in time $t$, every agent $i$ who is alive at time $t$ (i.e. $i \leq t$) meets one other agent $m_i(t)$ (identified by its birth date). The meeting process is random (described in detail below); we write $M_{i}(t)=\{m_{i}(t)\}_{t=i}^{i+T_{i}-1}$ for the \textit{meeting process} of agent $i$. The meeting process may stop at some finite time $T_{i}$ (the stopping time) or continue indefinitely (in which case $T_{i} = \infty$). The sample space of the meeting process is given by $\mathcal{M}$. Agents meet other agents who belong to one of two {\it choice sets} \footnote{This terminology was first introduced by Bruch and Mare in \cite{ref579}.}, namely the set of {\it followees of followees} and the set of {\it strangers}. Unlike the birth process, which is stationary, the meeting process depends on the current network, which in turn depends on the past history: the probability that agent $i$ meets agent $j$ at time $t$ depends on their relative positions in the network at time $t$, which in turns depend on the sequence of meetings for both agents up to time $t-1$. Moreover, the probability that a certain sample path of the meeting process occurs depends on all the exogenous parameters shown in Fig. 1. 

Given a time $t$, an agent $i$ alive at time $t$, and the existing network $G^{t}$, write $\mathcal{N}_{i,t}^{+}$ for the set of followees of $i$ and $\mathcal{K}_{i,t} = \left(\bigcup_{j \in \mathcal{N}^{+}_{i,t-1}} \mathcal{N}^{+}_{j,t-1}/\left\{i\right\}\right)/\mathcal{N}^{+}_{i,t-1}$ for the set of \textit{followees of followees} of agent $i$. Everyone who is neither a followee nor a followee of followee is a {\it stranger}. (Note that the newly born agent $t$ is always a stranger.) At time $t$ agent $i$ meets either a followee of a followee or a stranger; the probabilitity of meeting a followee of a followee (if one exists) is an exogenous parameter $\gamma_{k}\in[0,1]$ (where $k$ is the type of $i$), which we think of as structural opportunism (taking advantage of opportunities \footnote{The parameter $\gamma_{k}$ can also be thought of as a realization of the \textit{triadic closure}; the flow of ``trust" among connected individuals \cite{ref555}, or as an exploration-exploitation behavior; an agent either explores the network or exploits his current connections with different probabilities.}), where $\gamma_{k}=1$ for fully opportunistic agents, and $\gamma_{k}=0$ for fully non-opportunistic agents. 

Denote the set of type-$k$ followees of agent $i\in\mathcal{V}^{t}$ by $\mathcal{N}_{i,t}^{+,k}$, and the set of all followees of $i$ as $\mathcal{N}_{i,t}^{+}=\bigcup_{k=1}^{|\Theta|}\mathcal{N}_{i,t}^{+,k}$, where $|\mathcal{N}_{i,t}^{+}|=\mbox{deg}_{i}^{+}(t)$. Similarly, we denote the followers of agent $i$ by $\mathcal{N}_{i,t}^{-}$, where $|\mathcal{N}_{i,t}^{-}|=\mbox{deg}_{i}^{-}(t)$. Define the set $\mathcal{K}_{i,t}=\left(\bigcup_{j\in\mathcal{N}_{i,t-1}^{+}}\mathcal{N}_{j,t-1}^{+}/\left\{ i\right\} \right)/\mathcal{N}_{i,t-1}^{+}$ as the set of \textit{followees of followees} of agent $i$ at time $t$, and the set $\bar{\mathcal{K}}_{i,t}=\mathcal{V}^{t}/\left\{\mathcal{K}_{i,t}\bigcup\mathcal{N}_{i,t-1}^{+}\bigcup i\right\}$ as the set of \textit{strangers} to agent $i$ at time $t$. The set of same type followees of followees is denoted as $\mathcal{K}^{\theta_{i}}_{i,t}$. Let $N_{i}^{s}(t)=|\mathcal{N}_{i,t}^{+,\theta_{i}}|, N_{i}^{d}(t)=\mbox{deg}_{i}^{+}(t)-N_{i}^{s}(t)$, $K_{i}(t)=|\mathcal{K}_{i,t}|$, $K_{i}^{s}(t)=|\mathcal{K}_{i,t}^{\theta_{i}}|$, and  $K_{i}^{d}(t)=K_{i}(t)-K_{i}^{s}(t)$.    

For $t\geq i$, if there are no followees of followees, then $i$ meets a stranger with uniform probability. If there are followees of followees, then $i$ meets a followee of followee with probability $\gamma_k$ (and uniform over this choice set) and meets an agent picked uniformly at random from the network with probability $1-\gamma_k$, i.e. 
\[
\mathbb{P}\left(m_{i}(t)\in\mathcal{K}_{i,t}\left|\mathcal{K}_{i,t}\neq\emptyset\right.\right)=\gamma_{\theta_{i}}+(1-\gamma_{\theta_{i}})\frac{K_{i}(t)}{t-1},
\]
and 
\[
\mathbb{P}\left(m_{i}(t)\in\bar{\mathcal{K}}_{i,t}\left|\mathcal{K}_{i,t}\neq\emptyset\right.\right)=(1-\gamma_{\theta_{i}})\frac{t-1-K_{i}(t)}{t-1},
\]
whereas $\mathbb{P}\left(m_{i}(t)\in\bar{\mathcal{K}}_{i,t}\left|\mathcal{K}_{i,t}= \emptyset\right.\right)=1$. Note that since a new agent is born at each time step, and such an agent is a stranger to all other agents, then we have $\mathbb{P}\left(\bar{\mathcal{K}}_{i,t} \neq \emptyset\right) = 1$ for any time step $t$. At each time $t$, $i$ meets a new agent; $i$ may or may not form a link to this agent. In addition some agents may meet agent $i$, but $i$ does not form links to those agents. The meeting process realizes the \textit{limited-observability} of agents over time, i.e. agent $i$ reasons about forming social ties with only the agent it meets at time $t$, and cannot observe the global network structure or the types of all agents it does not meet. This is different from the complete information and complete observability network formation games in \cite{ref4}, or the preferential attachment models in \cite{ref91} which assumes that the linking behavior of a newly born agent relies on its knowledge of all the degrees of other agents. 


\subsection{The Linking Process}
When agent $i$ meets agent $m_{i}(t)$ at time $t$, it observes the type of $m_{i}(t)$ and decides whether or not to form a link with $m_{i}(t)$ (Thus true types of agents who meet are revealed). Agents draw benefits by linking to others but link formation is costly. Agents optimize so they form a new link if the marginal benefit of that link exceeds marginal cost. The marginal benefit depends on existing links and on types; we assume that linking to agents of the same type is (weakly) better than linking to agents of a different type -- this is homophily. For simplicity we assume marginal cost of linking is a constant $c$. 

We assume \textit{local externalities}, i.e. linking benefits do not flow to indirect contacts, so $i$ derives benefits only from its (direct) neighbors. For simplicity we assume that the utility depends only on the number of followees of the same type $N_{i}^{s}(t)$ and the number of followees of different types $N_{i}^{d}(t)$, and has the form  
\begin{equation}
u_{i}^{t}\left(G_{i,1}^{t}\right)=v_{\theta_{i}}\left(\alpha_{\theta_{i}}^{s}N_{i}^{s}(t)+\alpha_{\theta_{i}}^{d}N_{i}^{d}(t)\right)-c \sum_{j=i}^{t-1}a_{i}^{j},
\label{eq2}
\end{equation}
where $a_{i}^{t} \in \{0,1\}$ is the action of agent $i$ at time $t$; $a_{i}^{t}=1$ means that $i$ links to $m_{i}(t)$, and $a_{i}^{t}=0$ means that $i$ decides not to link to $m_{i}(t)$, and $\sum_{j=i}^{t-1}a_{i}^{j} = \left(N_{i}^{s}(t)+N_{i}^{d}(t)\right)$, $\alpha_{\theta_{i}}^{s} \geq \alpha_{\theta_{i}}^{d}, \forall \theta_{i} \in \Theta$ are the (type-specific) linking benefits, $v_{\theta_{i}}(x):x\rightarrow\mathbb{R}^{+}$ is the (type-specific) \textit{social benefit aggregation} function. For convenience, we assume that $v_{\theta_{i}}(x)$ is strictly concave \footnote{While we assume concavity of the utility function, our analysis applies to any saturating function, e.g. the \textit{sigmoid} function.}, twice continuously differentiable, monotonically increasing in $x$, and $v_{\theta_{i}}(0)=0$. That is, the marginal benefit of forming links diminishes as the number of links increases. This corresponds to the fact that agents do not form an infinite number of links in the network, but rather form a ``satisfactory" number of links \footnote{For instance, in citation networks, the number of references cited in a paper is finite and corresponds to the number of papers the authors need to acquire knowledge, yet the number of citations on a specific paper can be arbitrarily large.}. As shown in (\ref{eq3}), $i$ decides to link to $m_i(t)$ only if the marginal utility is positive. Note that $i$'s link formation decisions depend not only on the types of agents it meets, but also on the order with which it meets these agents. 

\begin{figure*}[!t]
\setcounter{mytempeqncnt}{\value{equation}} \setcounter{equation}{1}
\[a_{i}^{t}= \mathbb{I}_{\left\{\Delta u_{i}^{t}\left(m_{i}^{t}\right) > 0\right\}},\]
\begin{equation}
\Delta u_{i}^{t}\left(m_{i}^{t}\right) = v_{\theta_{i}}\left(\alpha_{\theta_{i}}^{s}\left(N_{i}^{s}(t)+\mathbb{I}_{\left\{\theta_{m_{i}(t)} = \theta_{i}\right\}}\right)+\alpha_{\theta_{i}}^{d}\left(N_{i}^{d}(t)+\mathbb{I}_{\left\{\theta_{m_{i}(t)} \neq \theta_{i}\right\}}\right)\right) - v_{\theta_{i}}\left(\alpha_{\theta_{i}}^{s}N_{i}^{s}(t)+\alpha_{\theta_{i}}^{d}N_{i}^{d}(t)\right) - c.
\label{eq3}
\end{equation}
\setcounter{equation}{\value{mytempeqncnt}+1} \hrulefill{}\vspace*{4pt}
 \end{figure*}

Agent $i$ will form a link to $m_i(t)$ exactly when doing so creates a network that yields higher utility for him. Agents are \textit{myopic} and form links without taking the future  into account. This seems to us to be a realistic description of behavior in social networks. 

\subsection{The Exogenous Homophily Index and Social Gregariousness}
We propose a novel definition of an \textit{exogenous homophily index} for type-$k$ agents, which is a variant of the well known \textit{Coleman homophily index} \cite{ref13}. For an agent $i$ of type $k$, let $\mathbb{N}_{i,t}^{+}$ be the space of all possible sets of followees of $i$ at time $t$. The exogenous homophily index of type-$k$ agents is the minimum fraction of same-type followees that type-$k$ agents desire. Thus, $h_{k}$ satisfies 
\begin{align}
\mathbb{P}\left(\left. \lim_{t\rightarrow\infty}\inf_{\mathcal{N}_{i,t}^{+} \in \mathbb{N}_{i,t}^{+}}\frac{N_{i}^{s}(t)}{\mbox{deg}_{i}^{+}(t)} = h_{k}\right|\theta_{i}=k, \gamma_{k} < 1\right) = 1,
\label{eq221}
\end{align}
where $0\leq h_{k}\leq1$. Note that this index is exogenous because it only depends on the agent's utility function and not the meeting process, thus it is independent from the network evolution path. When type-$k$ agents are indifferent to the types of agents they connect to, i.e. type-$k$ agents are extremely non-homophilic, then we have $\alpha_{k}^{s}=\alpha_{k}^{d}$, which means that $\lim_{t\rightarrow\infty}\inf_{\mathcal{N}_{i,t}^{+} \in \mathbb{N}_{i,t}^{+}}\frac{N_{i}^{s}(t)}{\mbox{deg}_{i}^{+}(t)}=0,\forall\theta_{i}=k$, i.e. agent $i$ can get satisfied by connecting to a set of followees that does not contain any same type followee. On the other hand, if agents restrict their links to same-type agents only, then we have $\alpha_{k}^{d} = 0$, and $\lim_{t\rightarrow\infty}\inf_{\mathcal{N}_{i,t}^{+} \in \mathbb{N}_{i,t}^{+}}\frac{N_{i}^{s}(t)}{\mbox{deg}_{i}^{+}(t)}=1,\forall\theta_{i}=k$. We can provide a closed form computation for  the exogenous homophily index by connecting it to social gregariousness. Define 
\begin{equation}
L_{\theta_{i}}^{*}(\alpha)=\arg\max_{x\in\mathbb{Z}}v_{\theta_{i}}\left(x\alpha_{\theta_{i}}^{s}+\alpha\right)-xc,\label{eq301}
\end{equation}
and 
\begin{equation}
\bar{L}_{\theta_{i}}^{*}(\alpha)=\arg\max_{x\in\mathbb{Z}}v_{\theta_{i}}\left(x\alpha_{\theta_{i}}^{d}+\alpha\right)-xc.\label{eq302}
\end{equation}
It follows from the concavity of $v_{\theta_{i}}(.)$ that $L_{\theta_{i}}^{*}(\alpha) < \infty$ and $\bar{L}_{\theta_{i}}^{*}(\alpha) < \infty, \forall \alpha \in \mathbb{R}$. It can be easily shown that 
\[
L_{\theta_{i}}^{*}\left(0\right)=\lim_{t\rightarrow\infty} \inf_{M_{i}(t)\in\mathcal{M}} \mbox{deg}_{i}^{+}(t).
\]
Thus, the parameter $L_{\theta_{i}}^{*}\left(0\right)$ represents the minimum number of links an agent will form with probability 1 in any (infinite) realization of the formation process; this captures social gregariousness. It can be shown that the exogenous homophily index of agent $i$ is given by \footnote{A detailed proof can be found in Appendix A.} 
\[
h_{\theta_{i}}=\frac{L_{\theta_{i}}^{*}\left(\alpha_{\theta_{i}}^{d}\bar{L}_{\theta_{i}}^{*}\left(0\right)\right)}{L_{\theta_{i}}^{*}\left(\alpha_{\theta_{i}}^{d}\bar{L}_{\theta_{i}}^{*}\left(0\right)\right)+\bar{L}_{\theta_{i}}^{*}\left(0\right)}.
\]
Thus, gregariousness and homophily are coupled. While each type of agents has an exogenous homophily index, which is network-independent, the actual fraction of same-type links an agent will realize depends on the meeting process and the individual agent's experience in the network. 

\subsection{Summary: Exogenous Parameters}
In summary, our model involves four exogenous parameters: 
\begin{itemize}
\item \textit{Homophily}: the homophily of type-$k$ agents is captured by the exogenous homophily index $h_{k}$. 
\item \textit{Social gregariousness}: the gregariousness of type-$k$ agents is captured by $L_{k}^{*}(0)$. 
\item \textit{Structural opportunism}: the parameter $\gamma_{k}$ reflects the extent of structural opportunism for type-$k$ agents. 
\item \textit{Type distribution}: the fraction of type-$k$ agents in a large network (relative population share) is given by $p_{k}$. 
\end{itemize}
Throughout this paper, we will use the notion of {\it first-order stochastic dominance} (FOSD). We say that a pdf (or pmf) $f(x)$ first-order stochastically dominates a pdf $g(x)$ if and only if $G(x) \geq F(x), \forall x,$ with strict inequality for some values of $x$, where $F(x)$ and $G(x)$ are the cumulative density functions. We write $X \succeq Y$ for the two random variables $X$ and $Y$ when $X$ first-order stochastically dominates $Y$. 

\section{Bonding capital and the egocentric value of networking}	
Our model captures several different forms of capital that an agent $i$ might acquire over time: 
\begin{itemize}
\item \textit{bonding capital} reflects agent $i$'s direct utility; 
\item \textit{popularity capital} reflects how other agents feel toward agent $i$; 
\item \textit{bridging capital} reflects agent $i$'s ability to connect other agents.
\end{itemize}
In this section we focus on bonding capital; we discuss popularity capital and bridging capital in following sections.   

\subsection{Ego network formation time}
Unlike previous works where link formation is a one-shot process (which is the case in \cite{ref9}, \cite{ref5}, \cite{ref18}-\cite{ref92}, \cite{ref15}, and \cite{ref17}), links (and consequently the bonding capital) are created over time in our model; individuals meet others and decide to establish connections until they forms a ``satisfactory" ego network/network of followees. The time needed for an agent to form its ultimate ego network/network of followees is obviously an important aspect of network formation. In this section, we characterize the bonding capital in terms of the time needed for the emergence of an ego network, as well as the utility resulting from bonding to that ego network.

Based on the definition of the utility function in (\ref{eq2}) and (\ref{eq3}), we know that there exists a finite number of connections after which an agent stops forming links. The time horizon over which the agent forms its ego network is random and depends on all the exogenous parameters. For an agent $i$, the ego network formation time (EFT) $T_{i}$ is a random function of the exogenous parameters, defined as 
\[
T_{i}\triangleq
\]
\begin{equation}
\inf \left\{t \in \mathbb{N}: u^{\tau}_{i}\left(G_{i,1}^{t}\right) \geq u^{\tau}_{i}\left(G_{i,1}^{t} \cup j\right), \forall \theta_{j} \in \Theta, \tau > t\right\} - i + 1.
\label{eq8}
\end{equation} 
We emphasize that $T_{i}$ is random: it depends on the network formation process. We characterize the time spent by an agent in the process of forming his ego network/network of followees in terms of the probability mass function (pmf) of $T_{i}$. \footnote{Note that $T_{i}$ can be thought of as the \textit{stopping time} of the linking process. This can be easily proven by showing that the event $T_{i}=T$ only depends on the history of meetings and link formation decisions up to time $T$.} We denote the pmf of $T_{i}$ as $f_{T_{i}}(T_{i}):\mathbb{N}\rightarrow[0,1].$ The expected ego network Formation Time (EEFT) $\overline{T}_{i}$ conditioned on agent $i$'s type is given by 
\begin{equation}
\overline{T}_{i}=\mathbb{E}_{\Omega_{\mathcal{G}}}\left[T_{i}\left|\theta_{i}\right.\right],
\label{eq9}
\end{equation}
where $\mathbb{E}_{\Omega_{\mathcal{G}}}\left[.\right]$ is the expectation operator, and the expectation is taken over all realizations of the graph process (we drop the subscript $\Omega_{\mathcal{G}}$ in the rest of our analysis). We say that agent $i$ is \textit{socially unsatisfied} if $T_{i}=\infty$; a socially unsatisfied agent is an agent that never satisfies its gregariousness requirements, i.e. agent $i$ is socially unsatisfied if $\mbox{deg}_{i}^{+}(t)<L_{\theta_{i}}^{*}(0),\forall t\geq i$. Such an agent keeps searching for followees forever. In the following Lemma, we specify the necessary and sufficient conditions under which a newly born agent has a positive probability of becoming socially unsatisfied. 
\begin{lem}
In order that agent $i$ becomes socially unsatisfied with positive probability, it is necessary and sufficient that $\gamma_{\theta_{i}}=1$ and $0<h_{\theta_{i}}<1$. 
\end{lem}
\begin{proof} 
See Appendix B. \, \, \IEEEQEDhere 
\end{proof}
This Lemma says that an agent gets unsatisfied if and only if it is not extremely homophilic and at the same time does not explore the strangers' choice set in its meeting process. In such a scenario, an agent's meeting process is governed by the actions taken previously by his neighborhood, which may not allow that agent to meet with other agents of diverse types. Unless otherwise stated, we assume that $\gamma_{k}<1,\forall k\in\Theta,$ thus agents never get trapped and all agents have a finite EFT. In the rest of this subsection, we characterize the EFT. We start by characterizing the EFT for extreme cases of agents' homophily in the following Theorem. 
\begin{thm}
\begin{enumerate}
\item If $h_{k}=0,\forall k\in\Theta$, then the EFT for agent $i$ is equal to $T_{i}=L_{\theta_{i}}^{*}(0)$ almost surely. 
\item If $h_{k}=1, \forall k \in \Theta$, then the distribution of the EFT for every agent $i$ conditioned on its type converges to a {\it steady-state distribution}, i.e. $\lim_{i\rightarrow \infty} f_{T_{i}}\left(T_{i}\left|\theta_{i}=k\right.\right) \rightarrow f^{k}_{T}(T),$ and the EEFT for an agent $i$ conditioned on its type $\overline{T}_{i} = \mathbb{E}\left\{T_{i}\left|\theta_{i}=k\right.\right\}$ converges as follows  
\[\lim_{i \rightarrow \infty} \overline{T}_{i} = \frac{1}{p_{k}} + \frac{L_{k}^{*}(\alpha^{s}_{k})}{(1-\gamma_{k}) \, p_{k}+\gamma_{k}}.\] 
\end{enumerate}
\end{thm}
\begin{proof} 
See Appendix C. \, \, \IEEEQEDhere 
\end{proof}
\begin{figure*}[t!]
    \centering
		    \begin{subfigure}[b]{0.3\textwidth}
        \centering
        \includegraphics[width=2.5 in]{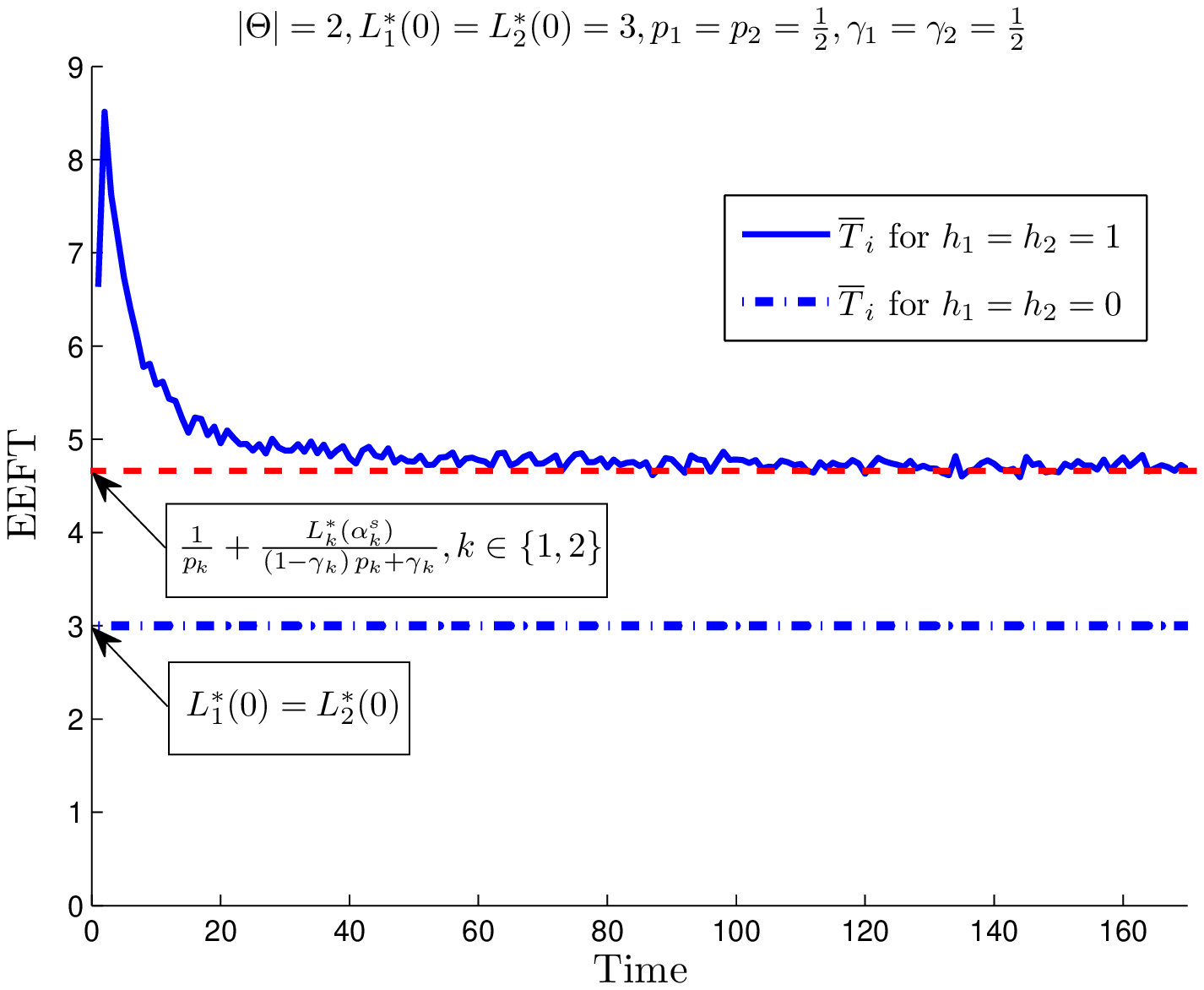}
        \caption{Impact of homophily on the EEFT.}
    \end{subfigure}
    ~ 
    \begin{subfigure}[b]{0.3\textwidth}
        \centering
        \includegraphics[width=2.5 in]{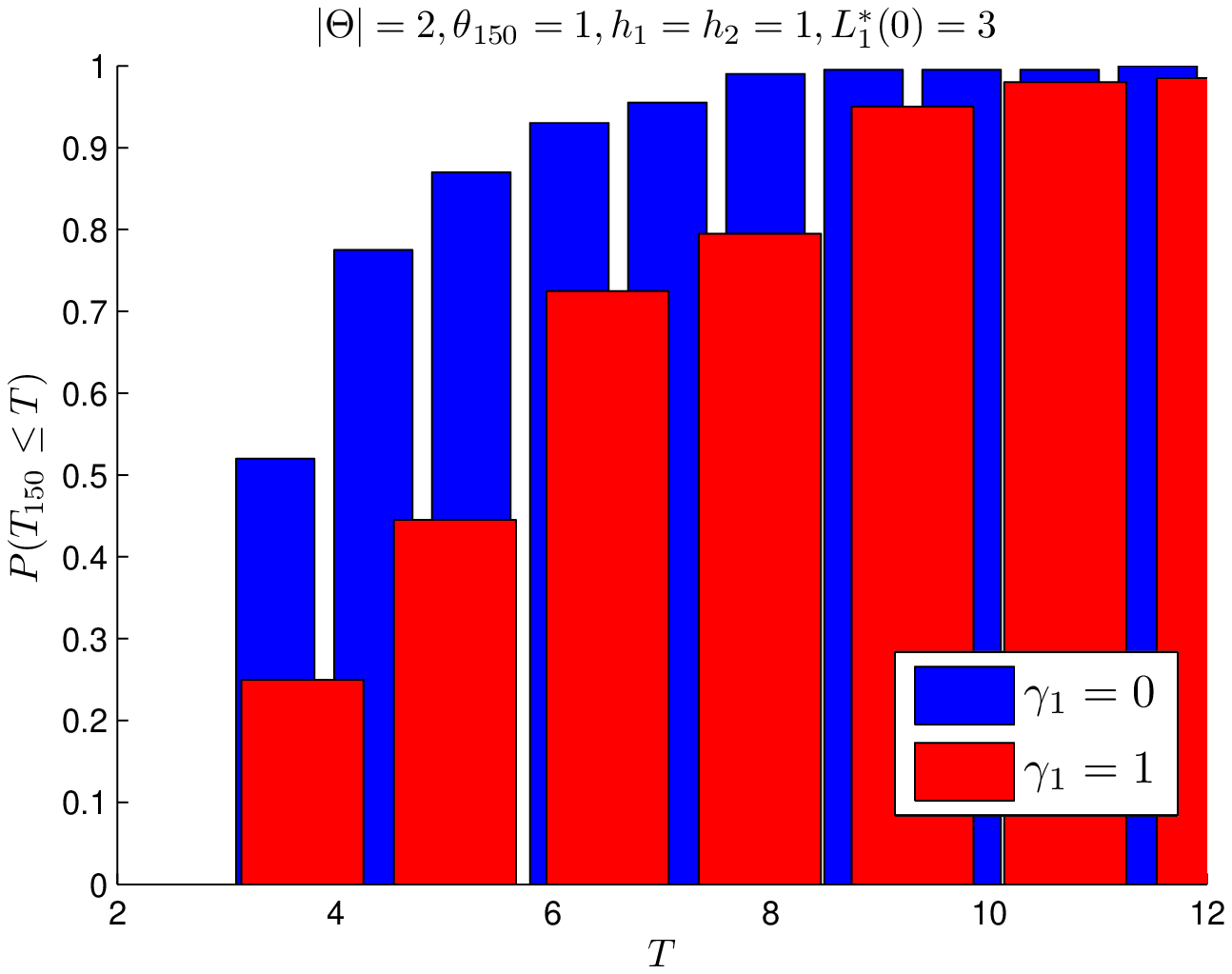}
        \caption{The EFT cdf for $\gamma_{1} = 0$ and $\gamma_{1} = 1$.}
    \end{subfigure}
    ~ 
     \begin{subfigure}[b]{0.3\textwidth}
        \centering
        \includegraphics[width=2.5 in]{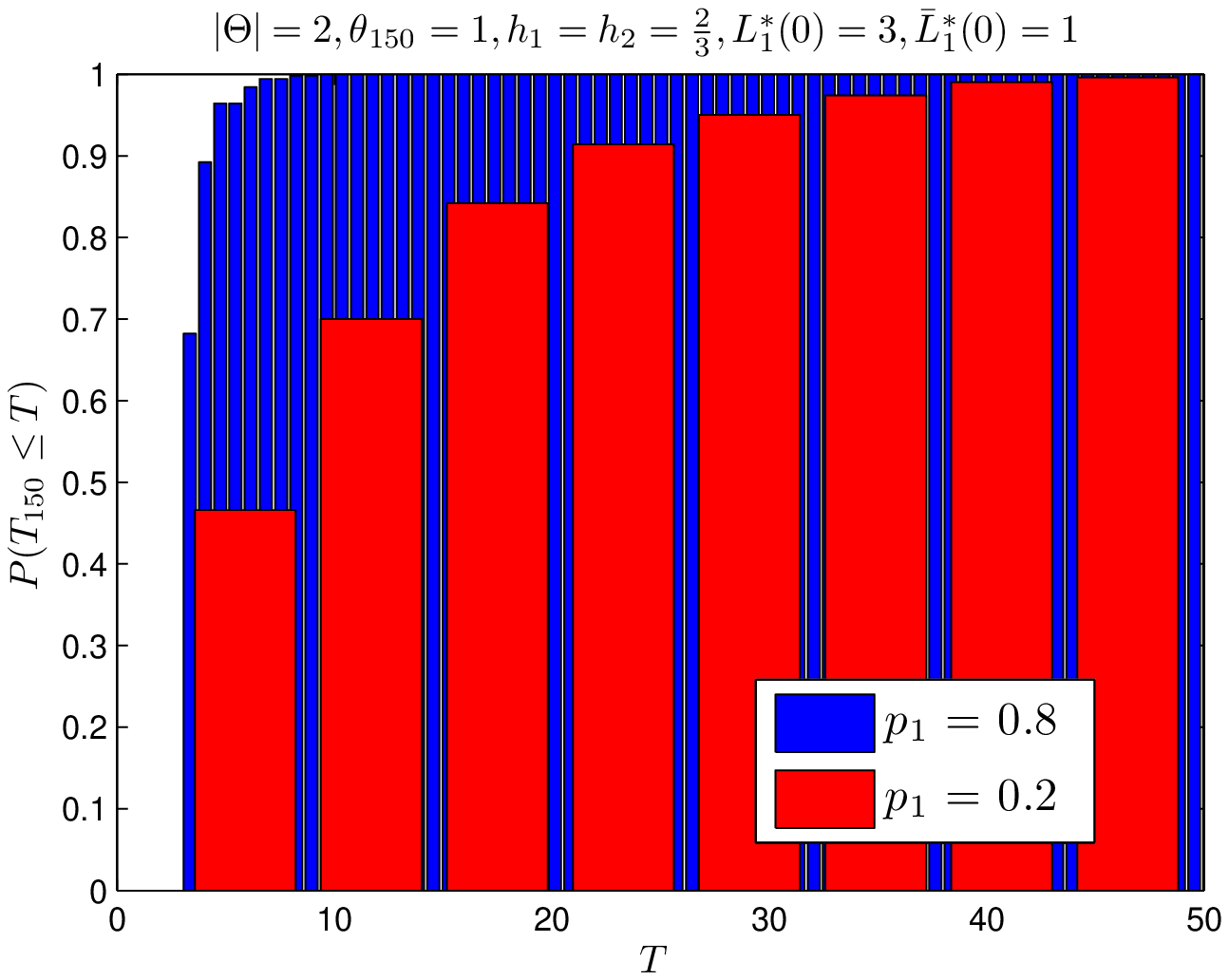}
        \caption{The EFT cdf for $p_{1} = 0.2$ and $p_{1} = 0.8$.}
    \end{subfigure}%
    \caption{Stochastic ordering of the EFT with respect to the exogenous parameters.}
\end{figure*}
Thus, the EFT for agents joining a large network only depends on their types. Theorem 1 says that when the agents are \textit{not} homophilic, there is no uncertainty in the ego network formation process, then both the number of links and the EFT are equal to $L_{\theta_{i}}^{*}(0)$ almost surely. This ``deterministic" EFT is independent of the network, and only depends on the agent's gregariousness. That is, if $h_{k}=0,\forall k\in\Theta$, then an agent's journey in the network is determined by how it values linking, and not by the network structure or the actions of others. If agents are more \textit{sociable}, i.e. are more gregarious, then they will spend more time searching for followees, yet this time is deterministic and only depends on parameters that are determined by the agent and not the network. On the other hand, if agents are extremely homophilic, then the agent's journey in the network will depend randomly on meetings with other agents with whom they do not form any links. It is clear from Theorem 1 that the EEFT of extremely homophilic agents depends on the type distribution and opportunism, in addition to gregariousness. We emphasize these dependencies in the following corollary.
\begin{crlry} \textit{(Gregarious agents and minorities search for followees longer, opportunistic agents search shorter)} 
If $h_{k}=1,\forall k\in\Theta$, $\tilde{L}_{\theta_{i}}^{*}(0)\geq L_{\theta_{i}}^{*}(0)$, $\tilde{p}_{\theta_{i}}\geq p_{\theta_{i}}$, and $\tilde{\gamma}_{\theta_{i}}\geq\gamma_{\theta_{i}}$, then for an agent $i$ born in an asymptotically large network we have that 
\[
T_{i}\left(p_{\theta_{i}},\gamma_{\theta_{i}},L_{\theta_{i}}^{*}(0)\right)\preceq T_{i}\left(p_{\theta_{i}},\gamma_{\theta_{i}},\tilde{L}_{\theta_{i}}^{*}(0)\right),
\]
\[
T_{i}\left(p_{\theta_{i}},\gamma_{\theta_{i}},L_{\theta_{i}}^{*}(0)\right)\succeq T_{i}\left(\tilde{p}_{\theta_{i}},\gamma_{\theta_{i}},L_{\theta_{i}}^{*}(0)\right),
\]
\[
T_{i}\left(p_{\theta_{i}},\gamma_{\theta_{i}},L_{\theta_{i}}^{*}(0)\right)\succeq T_{i}\left(p_{\theta_{i}},\tilde{\gamma}_{\theta_{i}},L_{\theta_{i}}^{*}(0)\right),
\]
where $T_{i}\left(p_{\theta_{i}},\gamma_{\theta_{i}},L_{\theta_{i}}^{*}(0)\right)$ is the EFT associated with the exogenous parameter tuple $\left(p_{\theta_{i}},\gamma_{\theta_{i}},L_{\theta_{i}}^{*}(0)\right).$
\end{crlry} 
\begin{proof} 
See Appendix D. \, \, \IEEEQEDhere
\end{proof} 
Note that stochastic dominance implies domination in mean. That is, if $T_{i}\left(p_{\theta_{i}},\gamma_{\theta_{i}},L_{\theta_{i}}^{*}(0)\right)\preceq T_{i}\left(p_{\theta_{i}},\gamma_{\theta_{i}},\tilde{L}_{\theta_{i}}^{*}(0)\right),$ then $\overline{T}_{i}\left(p_{\theta_{i}},\gamma_{\theta_{i}},L_{ \theta_{i}}^{*}(0)\right)\leq\overline{T}_{i}\left(p_{\theta_{i}},\gamma_{\theta_{i}},\tilde{L}_{\theta_{i}}^{*}(0)\right).$ Moreover, stochastic dominance implies domination of the expectation of any increasing function of the EFT; if the bonding capital is a decreasing function of the EFT, then one can infer the impact of the exogenous parameters on the expected bonding capital directly from the results of Corollary 1.

Corollary 1 says that in homophilic societies, the EFT of a social group \textit{increases (in the sense of FOSD)} as the gregariousness of that group increases. This is intuitive since the more followees an agent is willing to follow, the longer it takes to find those followees. Moreover, agents belonging to minorities are expected to spend more time in the link formation process. This is again intuitive since when the fraction of similar-type agents in the population is small, each agent would need to meet a longer sequence of agents in order to find similar-type followees. Finally, the EFT decreases in the sense of FOSD as structural opportunism increases. This is because once the agent becomes attached to a network component of similar-type agents, it is then better to be opportunistic and keep meeting the followees of followees who are guaranteed to be similar-type agents, rather than meeting strangers with uncertain types. In this context, structural opportunism captures what Mayhew calls ``structuralist" homophily effects in \cite{ref93}, and what Kossinets and Watts refer to as ``induced homophily" in \cite{ref90}, which corresponds to the fact that similar-type agents are more likely to ``meet" when agents are opportunistic. 

Note that the meeting process, encoded in the structural opportunism, plays a more crucial role for ``minor" types in homophilic societies as we show in the following Corollary.

\begin{crlry} 
If $h_{k}=1,\forall k\in\Theta$, then for an agent $i$ born in an asymptotically large network, the following is satisfied: 
\begin{enumerate}
\item If agent $i$ belongs to a minor type ($p_{\theta_{i}}\rightarrow 0$), then we have that $\lim_{\gamma_{\theta_{i}}\rightarrow 1}\overline{T}_{i} = \frac{1}{p_{\theta_{i}}}+L_{\theta_{i}}^{*}(0),$ and $\lim_{\gamma_{\theta_{i}}\rightarrow 0}\overline{T}_{i} = \frac{L_{\theta_{i}}^{*}(0)}{p_{\theta_{i}}}.$ 
\item If agent $i$ belongs to a major type ($p_{\theta_{i}}\rightarrow 1$), then for every $\gamma_{\theta_{i}}$ we have that $\lim_{\gamma_{\theta_{i}}\rightarrow 0}\overline{T}_{i} = L_{\theta_{i}}^{*}(0).$ \, \, \IEEEQEDhere
\end{enumerate}
\end{crlry}
Thus, if minor types \textit{exploit} their current connections to form new links, their EEFT becomes inversely proportional to their population size $p_{\theta_{i}}$ with an {\it additive} gregariousness parameter, whereas if the minor types \textit{explore} the network by meeting strangers, their EEFT becomes inversely proportional to their population size $p_{\theta_{i}}$ with a {\it multiplicative} gregariousness parameter. Therefore, minor types need to be more opportunistic for their EEFT to decrease. On the other hand, agents belonging to a ``major" type with $p_{\theta_{i}}\rightarrow1$ have an EEFT $\overline{T}_{i}\rightarrow L_{\theta_{i}}^{*}(0)$ regardless of their level of opportunism. Thus, the EFT of major types is less affected by the meeting process.

Fig. 3 reports simulations that illustrate the results of Theorem 1 and Corollary 1. In Fig. 3(a), we can see that the EEFT in an extremely homophilic society is greater than that of a non-homophilic society, and as the network grows, the EEFT for homophilic agents converges to the value specified by Theorem 1. In Fig. 3(b), we plot the cdf of the EFT for homophilic agents with different levels of opportunism, and we can see that the EFT of non opportunistic agents stochastically dominates that of opportunistic agents. Similarly, we demonstrate the impact of the type distribution in Fig. 3(c).
    
\subsection{Ego network characterization: homophily and structural holes}
In the previous subsection we have characterized the time needed for individuals to form their local ego networks, and thus realize a bonding capital. A common aspect in the definitions of bonding capital by Putnam \cite{refcap05}, Bourdieu \cite{ref370}, Coleman \cite{refcap01}, Fischer \cite{reffish}, and Cobb \cite{refcobb}, is that it corresponds to the \textit{social support} that individuals obtain through networking. Social support includes companionship, information exchange, emotional and instrumental support. In our model, agents derive social support from their followees; and such support is larger when the agent and its followees are of the same type, i.e. if an agent connects with same-type agents, they will acquire more relevant information \cite{refgay}. Thus, the types of agents in an agent's ego network determine its bonding capital. Based on this, we consider an agent's utility function, which represents the agent's net aggregate linking benefit, as an operational measure for the bonding capital accumulated by that agent. The bonding capital accumulated by an agent $i$ at time $t$ is simply measured by its utility $u_{i}^{t}$, whereas the bonding capital of type-$k$ agent is measured by their average utility $U_{k}^{t}=\frac{1}{\left|\mathcal{V}_{k}^{t}\right|}\sum_{j\in\mathcal{V}_{k}^{t}}u_{j}^{t},$ and the bonding capital of all agents in the network is $U^{t}=\frac{1}{t}\sum_{i\in\mathcal{V}^{t}}u_{i}^{t}$. 

We note that a larger ego network does not imply greater social or informational support. In fact, an agent might establish an ego network that comprises many different-type agents and will then have to pay the cost (time, effort, etc) to maintain the links with them while getting little social/informational support. For instance, a Twitter user who follows many accounts spreading information that is not relevant to the user's interests leads to low bonding capital: the user then spends time following such accounts but gets low informational benefits. The utility of each agent in a steady-state ego network is a measure for the support that an individual can obtain from other individuals in his local personal network. In the following Theorem, we show that maximum bonding capital is only achieved in societies with extreme homophily. 

\begin{thm}
(Homophily induces structural holes) Assume $h_{l}>0, \forall\ell\in\Theta$. In order that the total average utility $U^{t}$ converges to the optimal value $\overline{U}^{*}$ as the network grows without bound it is necessary and sufficient that $h_{l}=1,\forall\ell\in\Theta.$ If this is the case then the network at any time step will be disconnected almost surely and have at least $|\Theta|$ non-singleton components. 
\end{thm}
\begin{proof} 
See Appendix E. \, \, \IEEEQEDhere 
\end{proof}

If all agents are extremely homophilic, then a disconnected network that maximizes the achieved utility always emerges, and such a network is always disconnected even with the limited observability of the meeting process. A disconnected network obviously entails \textit{structural holes} as defined by Burt \cite{ref55sh}\cite{ref55sh3}: same-type agents form connected components that do not communicate with other types of agents, thus different types of agents do not exchange ideas and information. As shown in Fig. 4, the optimal total average utility is only achieved when agents are extremely intolerant towards different-type agents. Thus, maximizing the bonding capital in homophilic societies implies the presence of structural holes. For any network with non-extremely homophilic agents, the limited observability of agents dictated by the meeting process allows the agents to fill the network's structural holes. In other words, what makes the network connected is that not all similar-type agents observe each other at each time step, but they can potentially meet different-type agents with which they decide to connect. If the meeting process allows unlimited observability, i.e. $m_{i}(t)=\mathcal{V}^{t}/\{i\}$, then the agents will always converge to a disconnected network with $|\Theta|$ non-singleton components. 

The major conclusion drawn from this section is that homophily leads agents to reside in more homogeneous ego networks, but also leads the agents to wait longer in order to establish their ego networks, and induces structural holes in the global network structure. Thus, on one hand homophily  unifies the local structure of the network by gathering people with similar traits together, but on the other hand it divides the global network structure since dissimilar social groups become weakly connected. This creates another potential source of capital, namely a \textit{bridging capital}, which we discuss in Section 5.

\begin{figure}[t!]
    \centering
    \includegraphics[width=3.5 in]{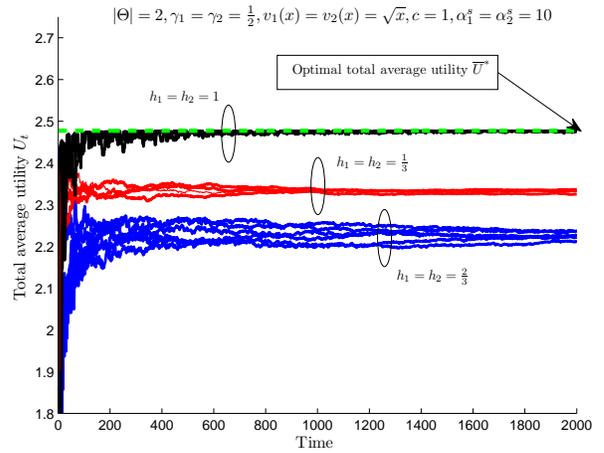}     
    \caption{The total average utility for various realizations of the evolving social network.}
\end{figure}

\section{Popularity capital and preferential attachment}	
Since in our model we consider directed networks, we distinguish between conventional bonding capital, which is realized by homogenous networks of like-minded people that provides social support for the individual, and \textit{popularity capital}, which corresponds to the individual's influence in the network that is gained by supporting others. Popularity is an important form of social capital that represents an individual's influence on a social group; an individual's ability to spread opinions, information, and ideas; and also an individual's acquisition for group support. For instance, users of Twitter acquire popularity measured by the number of followers, which allows them to express opinions, problems, and experiences, and acquire emotional support provided by their online support groups (followers). Followers retweet the tweets posted by users, which allows those users to spread their ideas and opinions \cite{reftwit}. Similarly, the popularity of researchers measured by the number of citations or the h-index allows those researchers to promote for new research ideas and directions. In this section we study popularity capital and connect it to preferential attachment, which is a central concept in network science.

The popularity of agent $i$ at time $t$ is simply given by $\mbox{deg}_{i}^{-}(t)$. We say that the \textit{popularity growth rate} of agent $i$ is $O\left(g(t)\right)$ if $\lim_{t\rightarrow\infty}\frac{\mathbb{E}\left\{\mbox{deg}_{i}^{-}(t)\right\}}{g(t)}=1,$ where the expectation is taken over all realizations of the graph process given that agent $i$ is born with a type $\theta_{i}$. (Note that the growth rate is only uniquely defined ``near infinity".) The \textit{popularity distribution} (sometimes called the degree distribution \cite{ref9}\cite{ref1811}\cite{ref12}) is denoted by $f_{d}^{t}(d),$ and corresponds to the fraction of agents with a popularity level of $d$ at time $t$, i.e. $f_{d}^{t}(d)=\frac{1}{t}\left|\left\{ i\left|\mbox{deg}_{i}^{-}(t)=d,i\in\mathcal{V}^{t}\right.\right\} \right|$. For a given type $k$, $f_{d}^{t,k}(d)$ denotes the popularity distribution of type-$k$ agents at time $t$: $f_{d}^{t,k}(d)=\frac{1}{\left|\mathcal{V}_{k}^{t}\right|}\left|\left\{ i\left|\mbox{deg}_{i}^{-}(t)=d,i\in\mathcal{V}_{k}^{t}\right.\right\} \right|$. Let $\Delta\mbox{deg}_{i}^{-}(t)$ be the number of followers gained by agent $i$ at time $t$, i.e. $\Delta\mbox{deg}_{i}^{-}(t)=\mbox{deg}_{i}^{-}(t)-\mbox{deg}_{i}^{-}(t-1)$.

\textit{Preferential attachment} has been used to explain the underlying mechanism of networks growth \cite{ref9}, \cite{ref4}-\cite{ref7}, \cite{ref18}-\cite{ref12}. All of these previous papers model agents as forming links only once; in our model, agents may form links many times. More importantly, all of these previous models {\it impose} preferential attachment as a behavioral rule (so network growth is viewed as a conventional {\it stochastic urn process}); in our model, preferential attachment {\it emerges endogenously}. 

To fix ideas, we first provide a general definition of preferential attachment that will be adopted in what follows. 
\begin{defntn} (\textit{Preferential attachment}) We say that preferential attachment emerges in the network
formation process if $\mbox{deg}_{i}^{-}(t) \geq \mbox{deg}_{j}^{-}(t)$ implies $\Delta\mbox{deg}_{i}^{-}(t)\succeq\Delta\mbox{deg}_{j}^{-}(t)$. \,\, \IEEEQEDhere 
\end{defntn} 
In words: preferential attachment means that agents who are more popular at a given time are likely to become even more popular in the future.    

\subsection{Popularity capital in tolerant societies}
We begin by studying popularity capital in societies with extreme exogenous homophily index for all types of agents given by $h_{k}=0,\forall k\in\Theta$. It seems natural to refer to such societies as tolerant (rather than totally non-homophilic). We study the factors that create inequality of popularity capital in  tolerant societies. In the following Theorem, we begin by studying the impact of the exogenous network parameters on the popularity growth rates. 
\begin{thm}
(\textit{Popularity growth in tolerant societies}) For a tolerant society popularity growth rates enjoy the following properties: 
\begin{itemize}
\item For $\gamma_{k}=0,\forall k\in\Theta$, the popularity of any agent $i$ grows logarithmically with time, i.e. $\mathbb{E}\left\{\mbox{deg}_{i}^{-}(t)\right\} $ is $O\left(\bar{L}\,\log(t)\right)$, where $\bar{L}=\sum_{m\in\Theta}p_{m}L_{m}^{*}(0)$. 
\item For $\gamma_{k}=1,\forall k\in\Theta$, the popularity of any agent $i$ grows at least sublinearly with time, i.e. $\mathbb{E}\left\{\mbox{deg}_{i}^{-}(t)\right\}$ is at least as fast as $O\left(t^{b}\right)$, where $b$ is given in Appendix F and is the same for all types of agents.
\end{itemize}
\end{thm}
\begin{proof} 
See Appendix F. \, \, \IEEEQEDhere 
\end{proof}

This Theorem demonstrates the impact of opportunism and gregariousness on popularity accumulation. On one hand, the popularity of agents in non-opportunistic societies grows logarithmically with time -- very slowly. On the other hand, the popularity of agents in opportunistic societies grows sublinearly with time -- again slowly, but much faster than for non-opporunistic agents. Thus, opportunism has an enormous influence on popularity. As we show below, this is a consequence of preferential attachment. 
\begin{crlry}
\textit{(Emergence of preferential attachment)} For a tolerant society, preferential attachment emerges if all agents are opportunistic, i.e. $\gamma_{k}=1,\forall k\in\Theta$. 
\end{crlry} 
\begin{proof}
See Appendix G. \, \, \IEEEQEDhere
\end{proof}  
In the following Corollary, we show that agents' ages in the network create inequality in the popularity capital. 
\begin{crlry} \textit{(Superiority of older agents in tolerant societies)} 
For a tolerant society, we have that $\mbox{deg}_{i}^{-}(t)\succeq\mbox{deg}_{j}^{-}(t)$ for all $i<j$.
\end{crlry} 
\begin{proof} 
See Appendix H. \, \, \IEEEQEDhere
\end{proof} 
Thus in the setting of Corollary 4, age is the only factor that creates inequality in popularity capital. In the following Corollary, we show that opportunism creates long term popularity advantages for agents forming the network. 
\begin{crlry}
{\it (Opportunism is good in the long-run)} 
If $d_{i}^{1}(t)$ is the popularity of agent $i$ at time $t$ in a tolerant society with $\gamma_{k}=0,\forall k \in \Theta,$ and $d_{i}^{2}(t)$ is the popularity of agent $i$ at time $t$ in a tolerant society with $\gamma_{k}=1,\forall k \in \Theta,$ then we have that $\mathbb{E}\left\{d_{i}^{2}(t)\right\} \leq \mathbb{E}\left\{d_{i}^{1}(t)\right\}$ for all $t\leq T^{*},$ and $\mathbb{E}\left\{d_{i}^{2}(t)\right\} \geq \mathbb{E}\left\{d_{i}^{1}(t)\right\}$ for all $t>T^{*},$ where $T^{*} \leq i\times\left(-\bar{L}\,\,\mathcal{W}_{-1}\left(\frac{-1}{\bar{L}}e^{\frac{-1}{\bar{L}}}\right)\right)^{\frac{1}{b}},$ $b = \sum_{m \in \Theta} p_{m} L^{*}_{m}(0)$, and $\mathcal{W}_{-1}(.)$ is the lower branch of the Lambert W function \cite{ref372}. 
\end{crlry}
\begin{proof} 
See Appendix I. \, \, \IEEEQEDhere 
\end{proof}
Thus, in societies where individuals are opportunistic, the long-term popularity capital is harvested after a certain time threshold as shown in Fig. 5. Such threshold is increasing in the agents' average gregariousness. Thus, younger agents or agents in a society with large average gregariousness need to wait longer to harvest the popularity gains attained by opportunism. To sum up, in tolerant societies, only age creates popularity capital inequality, and the growth of the popularity capital is governed by both the level of opportunism and the average gregariousness of the agents' types. However, there is no inequality in the acquisition of popularity capital across different social groups.   

\begin{figure}[t!]
    \centering
    \includegraphics[width=3.5 in]{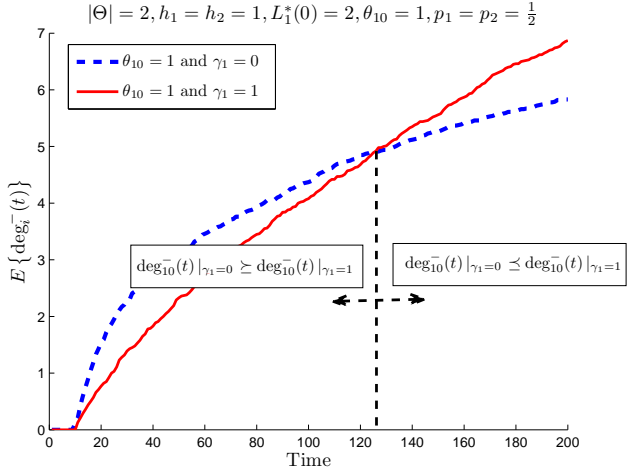}     
    \caption{Expected popularity over time for an agent born at $t=10$ for different levels of opportunism.}
\end{figure}

\begin{figure}[t!]
    \centering
    \includegraphics[width=3.5 in]{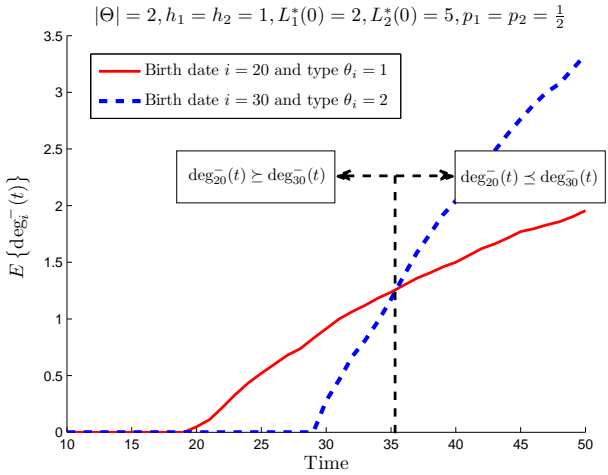}     
    \caption{Expected popularity over time for two agents belonging to 2 social groups with different levels of gregariousness.}
\end{figure}

\subsection{Popularity capital in intolerant societies}

We now study popularity capital in societies for which $h_k = 1$ for all $k$; it seems natural to refer to such societies as intolerant (rather than totally homophilic). In the following Theorem, we study the popularity growth rates for different types of agents in the network. 
\begin{thm}
(Popularity growth in intolerant societies) For an intolerant society, the popularity growth rates are given as follows: 
\begin{itemize}
\item For $\gamma_{k}=0,\forall k\in\Theta$, the mean-field approximation for the popularity of every agent $i$ grows logarithmically with time, i.e. $\mathbb{E}\left\{\mbox{deg}_{i}^{-}(t)\right\}$ is $O\left(L_{\theta_{i}}^{*}(0)\,\log(t)\right)$. 
\item For $\gamma_{k}=1,\forall k\in\Theta$, the mean-field approximation for the popularity of every agent $i$ grows at least sublinearly with time, i.e. $\mathbb{E}\left\{\mbox{deg}_{i}^{-}(t)\right\}$ is at least as fast as $O\left(t^{b_{\theta_{i}}}\right)$, where $b_{k} > b_{m}$ if $L_{k}^{*}(0) > L_{m}^{*}(0), \forall k, m \in \Theta$. 
\end{itemize}
\end{thm}
\begin{proof}
See Appendix J. \, \, \IEEEQEDhere 
\end{proof}
Thus, for tolerant and intolerant societies, the popularity growth rates are qualititatively similar -- but the sublinear growth obeys a different exponent. In the following Corollary, we show that gregariousness and opportunism create inequality in the popularity capital.
\begin{crlry} \textit{(Gregariousness and opportunism create inequality in the popularity capital)} For an intolerant society, and for the two agent types $k,m\in\Theta$ in the network with arbitrary $p_{k}$ and $p_{m}$, the following is satisfied: 
\begin{itemize}
\item If $\gamma_{k}=\gamma_{m}$, $\gamma_{k} \in \{0,1\}$, and $L_{k}^{*}(0)>L_{m}^{*}(0)$, then there exists a time $T^{*}<\infty$ where $\mathbb{E}\left\{\mbox{deg}_{i}^{-}(t)\right\}\geq \mathbb{E}\left\{\mbox{deg}_{j}^{-}(t)\right\},$ for all $t>T^{*}$, where $\theta_{i}=k$ and $\theta_{j}=m$. 
\item If $\gamma_{k}=1$ and $\gamma_{m}=0$, and $L_{k}^{*}(0)=L_{m}^{*}(0)$, then there exists a time $T^{*}<\infty$ where $\mathbb{E}\left\{\mbox{deg}_{i}^{-}(t)\right\}\geq\mathbb{E}\left\{\mbox{deg}_{j}^{-}(t)\right\},$ for all $t>T^{*}$, where $\theta_{i}=k$ and $\theta_{j}=m$. 
\end{itemize}
\end{crlry} 
\begin{proof} 
See Appendix K. \, \, \IEEEQEDhere
\end{proof} 
This agent-level characterization can be further generalized to the collective popularity of social groups in the following Theorem; we show that the popularity distribution of a more gregarious (or opportunistic) social group stochastically dominates that of a less gregarious (or opportunistic) group. 
\begin{thm}
(Popularity capital inequality across social groups) For an intolerant society, the following is satisfied: 
\begin{itemize}
\item If $\gamma_{k}=\gamma_{m}$, $\gamma_{k} \in \{0,1\}$, and $L_{k}^{*}(0)>L_{m}^{*}(0)$, then $f_{d}^{t,k}(d)$ first order stochastically dominates $f_{d}^{t,m}(d)$ assuming a mean-field approximation for the popularity acquisition process. 
\item If $\gamma_{k}=1$ and $\gamma_{m}=0$ and $L_{k}^{*}(0)=L_{m}^{*}(0)$, then $f_{d}^{t,k}(d)$ first order stochastically dominates $f_{d}^{t,m}(d)$ assuming a mean-field approximation for the popularity acquisition process. 
\end{itemize}
\end{thm}
\begin{proof} 
See Appendix L. \, \, \IEEEQEDhere 
\end{proof}
Thus, for intolerant socieites, popularity is influenced by gregariousness and structural opportunism rather than by population share. See Fig. 6. In contrast with tolerant societies, a younger agent in an intolerant society can become and remain more popular than an older agent if the younger agent belongs to a more gregarious or more opportunistic social group. 

Theorem 5 can also be understood in the context of citation networks \cite{ref373}. In the context of citation networks, intolerance means simply that papers only cite papers that are really related -- which is of course very common and not at all unusual. \cite{ref373} shows, in many scientific fields, that there is a strong positive correlation between the number of references per paper and the total number of citations. We quote the following conclusion from the report in \cite{ref577}, which is based on a statistical analysis of \textit{Thomson Reuters' Essential Science Indicators} database: ``\textit{One might think that the number of papers published or the population of researchers in a field are the predominant factors that influence the average rate of citation, but it is mostly the average number of references presented in papers of the field that determines the average citation rate.}" This conclusion is in perfect agreement with Theorem 5 (and Corollary 6), which predict that for the inherently intolerant citation networks, the popularity of researchers in different fields (total citation rate) is governed by their ``gregariousness" (number of references per paper), and not by the type distribution (number of papers/researchers). We know from \cite{ref577} that papers in mathematics typically list few references, while those in molecular biology typically list many. Thus, molecular biologists are more ``gregarious" than mathematicians -- and one would expect that younger molecular biologists can, on average, become more ``popular" -- have higher citation indices -- than mathematicians, solely because of the differences between disciplines and entirely unrelated to ``quality" or ``real impact". This would seem to provide caution for University review committees.

Of course, other dimensions in addition to popularity/citation counts, express the value of scholarly work. One of these dimensions is  \textit{interdisciplinarity} \cite{refcent7}, which is a form of bridging capital rather than a popularity capital as we show in the next Section.

\section{Bridging capital, connectedness and the strength of weak ties}

\begin{figure*}[t!]
    \centering
    \begin{subfigure}[b]{0.3\textwidth}
        \centering
        \includegraphics[width=2.25 in]{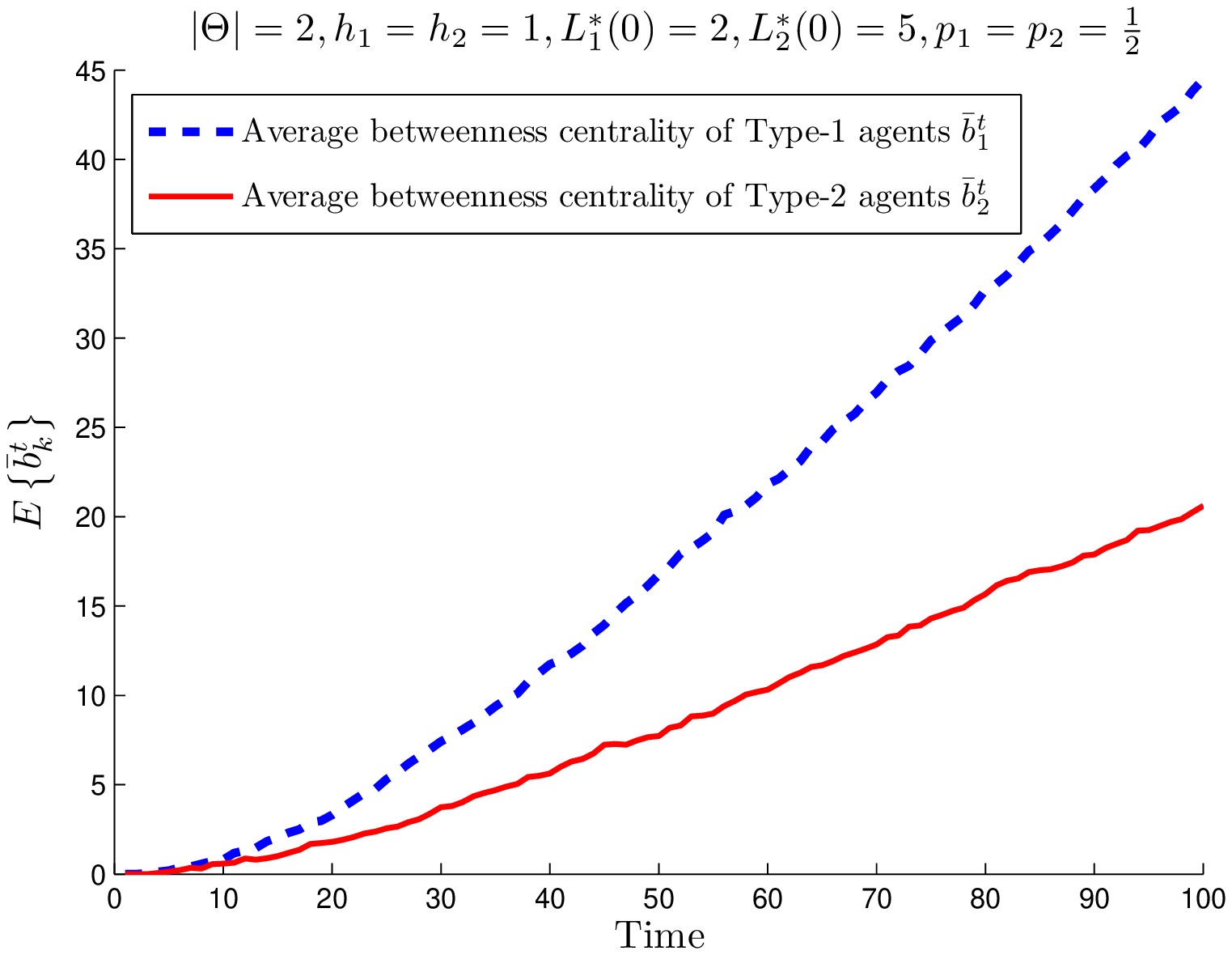}
        \caption{Impact of gergariousness on agents' centrality.}
    \end{subfigure}%
    ~ 
    \begin{subfigure}[b]{0.3\textwidth}
        \centering
        \includegraphics[width=2.25 in]{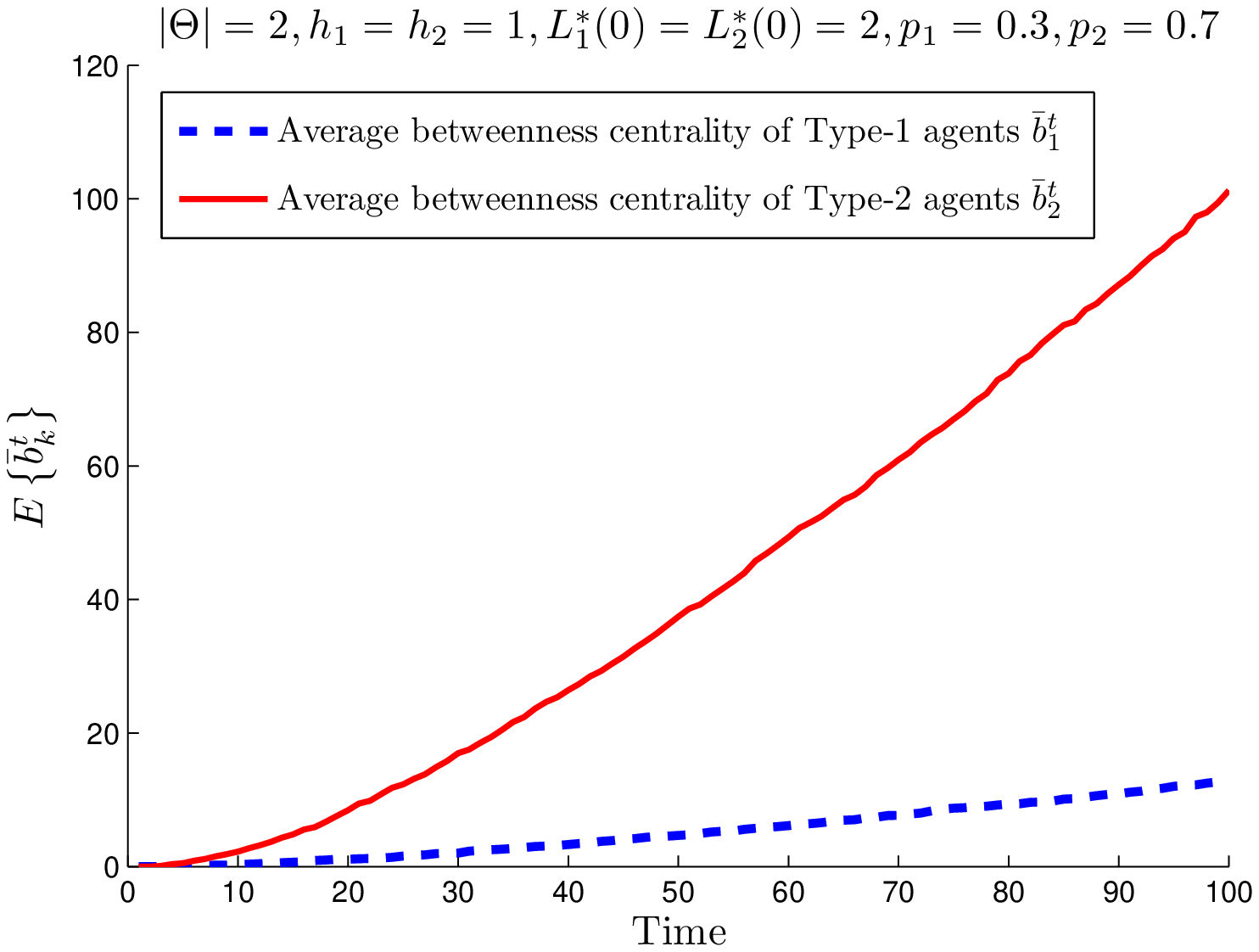}
        \caption{Impact of type distribution on agents' centrality.}
    \end{subfigure}
		~
		\begin{subfigure}[b]{0.3\textwidth}
        \centering
        \includegraphics[width=2.25 in]{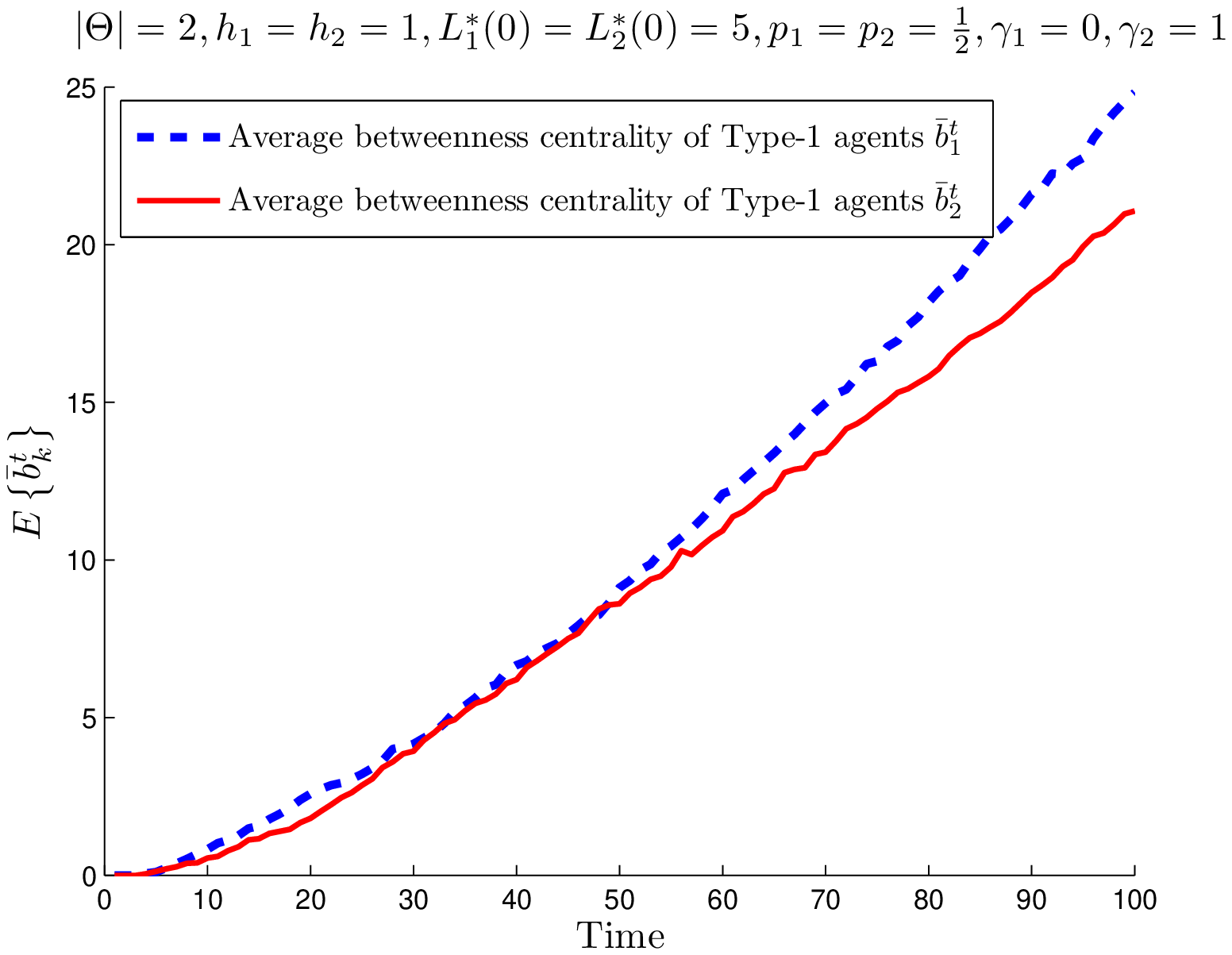}
        \caption{Impact of structural opportunism on agents' centrality.}
    \end{subfigure}
    \caption{Estimates for the average betweenness centrality of two types of agents in an extremely homophilic agents.}
\end{figure*}

\subsection{Betweenness centrality as a measure for bridging capital}

In Sections 3 and 4, we have studied two forms of capital that share two basic features: they are egocentric in the sense that they create value for individuals, and they are only affected by the agents' local network structures. Bonding occurs when an individual socializes with similar individuals driven by homophily, whereas bridging occurs when an individual links multiple segregated communities. While bonding creates egocentric values for individuals, bridging creates shared value for the network, e.g. allows diverse research communities to exchange ideas and innovations. As Burt points out in \cite{ref55sh}, individuals with bridging capital enjoy a \textit{central} position in the network as they act as a gateway for information exchange. Betweenness centrality, a graph-theoretic measure promoted by Freeman in \cite{refcap99}, is a conventional measure for centrality since for a given agent it counts the number of shortest paths between any two agents that involves that agents, and thus reflects the agents' ability to broker interactions at the interface between different groups \cite{refcent2}-\cite{refcent8}. The betweenness centrality of agent $i$ at time $t$, which is denoted by $b_{i}^{t}$, is an indicator of its centrality in the network \cite{refcap99}, and is given by
\begin{equation}
b_{i}^{t} = \sum_{k\neq j \neq i}\frac{\sigma_{kj}(i)}{\sigma_{kj}},
\label{betwenneq}
\end{equation}     
where $\sigma_{kj}$ is the total number of shortest paths between $k$ and $j$ in $G^{t}$ ignoring the edge directions, and $\sigma_{kj}(i)$ is the number of such paths that pass through $i$. In order to characterize the centrality of a certain social group, we define the average betweenness centrality of type-$k$ agents $\bar{b}_{k}^{t}$ as follows 
\[
\bar{b}_{k}^{t}=\frac{1}{\left|\mathcal{V}_{k}^{t}\right|}\sum_{i\in\mathcal{V}_{k}^{t}}b_{i}^{t}.
\]
Betweenness is a relational measure: an agent with a high betweenness centrality score does not belong to one of the dense groups, but relates them. While the evolving network is modelled as a directed graph, we capture the bridging capital by computing the betweenness centrality of agents in the {\it simplified} undirected version of the graph $G^{t}$. This is because bridging capital reflects the {\it structural centrality} of the agent, i.e. to what extent an agent is ``between" segregated groups, whereas the edge directions reflect the directions of information flow. As shown in Fig. 7, a central agent can either disseminate information to segregated groups, transfer information from one group to another, or gather information produced by different groups. In Fig. 7, the central agent has the same betweenness centrality score in the networks (a), (b), and (c), yet the role played by that agent in each network is different. In (a), the central agent gets non-redundant information from community 1 and community 2, which allows that agent to come up with {\it innovations and new ideas}. In (b), the central agent transfers information from community 1 to community 2, which allows that agent to {\it control} the flow of information across groups. In (c), the central agent displays {\it influence} on community 1 and community 2 by disseminating information to those communities. In the three networks, the bridging capital (i.e. extent of the agent's betweenness) is the same, yet the role of the central agent and the nature of its social advantage is different. We are interested in characterizing the extent of structural centrality of the agents in the network rather than the specific roles they play at the interface between groups.      

Characterizing the betweenness centrality for a general network is not mathematically tractable, and only empirical and simulation results are obtained in the literature \cite{refcent2}-\cite{refcent3}. We start by presenting simulation results for the betweenness centrality of agents in a network with 2 types, and show the impact of the exogenous parameters. In Fig. 8, we plot estimates for the expected average betweenness centrality of 2 types of extremely homophilic agents obtained via a Monte Carlo simulation, highlighting the impact of gregariousness, type distribution, and structural opportunism. In Fig. 8(a), we can see that increasing gregariousness decreases centrality, which is intuitive since when each agent forms many links, the number of shortest paths that involve far agents in terms of the geodesic distance will decrease, which leads to a decrease in the average centrality of the whole social group. That is, when all agents are sociable, then all agent are less central (on average). This is in striking contrast with the popularity capital, where gregariousness of agents in a social group was helping them acquiring popularity. Moreover, we can see in Fig. 8(b) that the type distribution plays a role in determining the agents' centrality; majorities are more central than minorities. Such result, which agrees with the qualitative study of Ibarra in \cite{refibarra}, is again in a striking contrast with the popularity capital acquisition where the type distribution had no significant impact on the agents' popularity growth rates. Finally, Fig. 8(c) shows that structural opportunism decreases centrality, which is again in contrast with the popularity acquisition experience where structural opportunism was allowing for the emergence of preferential attachment.
\begin{figure*}[t!]
    \centering
    \begin{subfigure}[b]{0.325\textwidth}
        \centering
        \includegraphics[width=3 in]{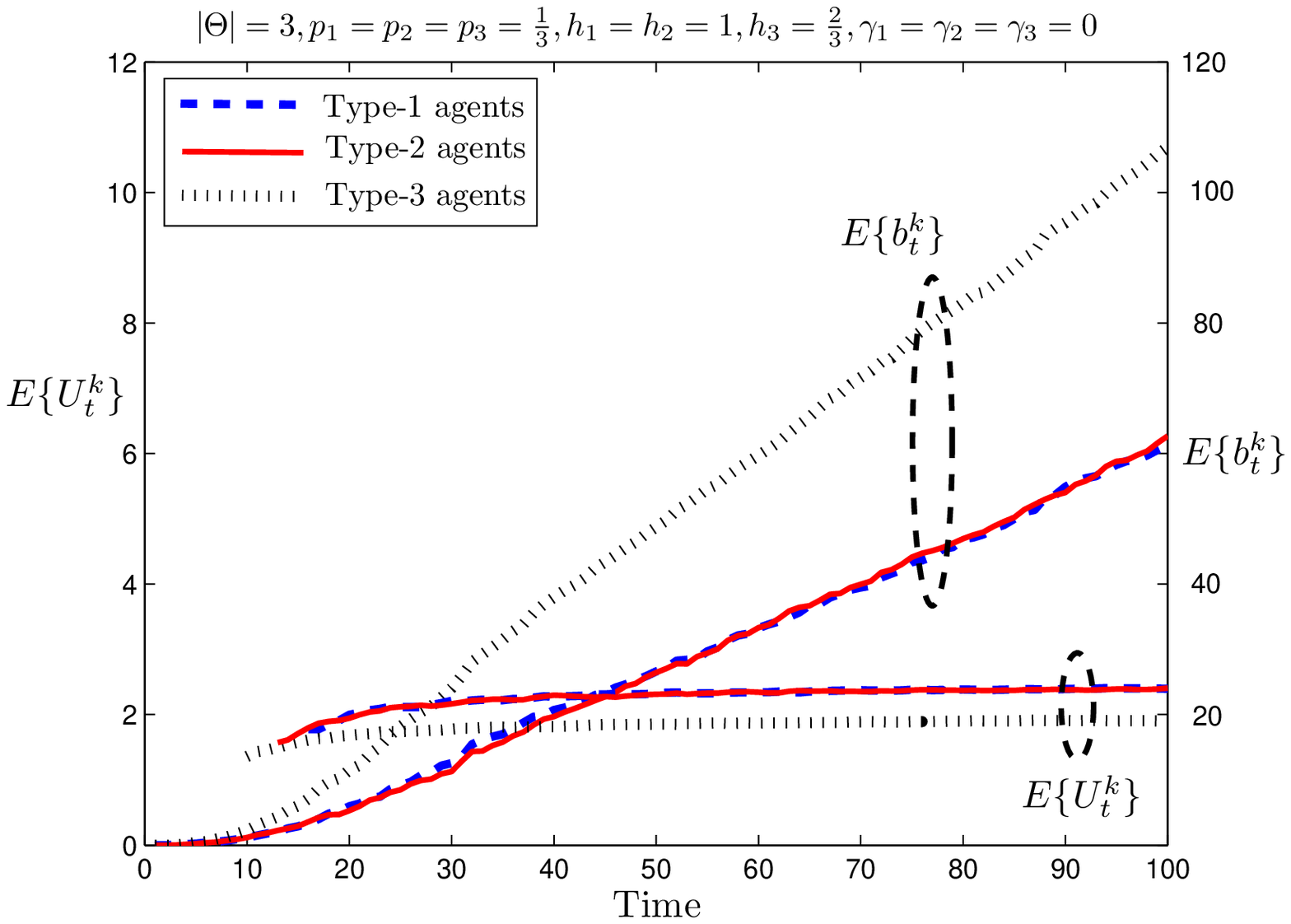}
        \caption{Centrality and utility of non-homophilic and exploring agents in a homophilic society.}
    \end{subfigure}%
    ~ 
    \begin{subfigure}[b]{0.525\textwidth}
        \centering
        \includegraphics[width=3 in]{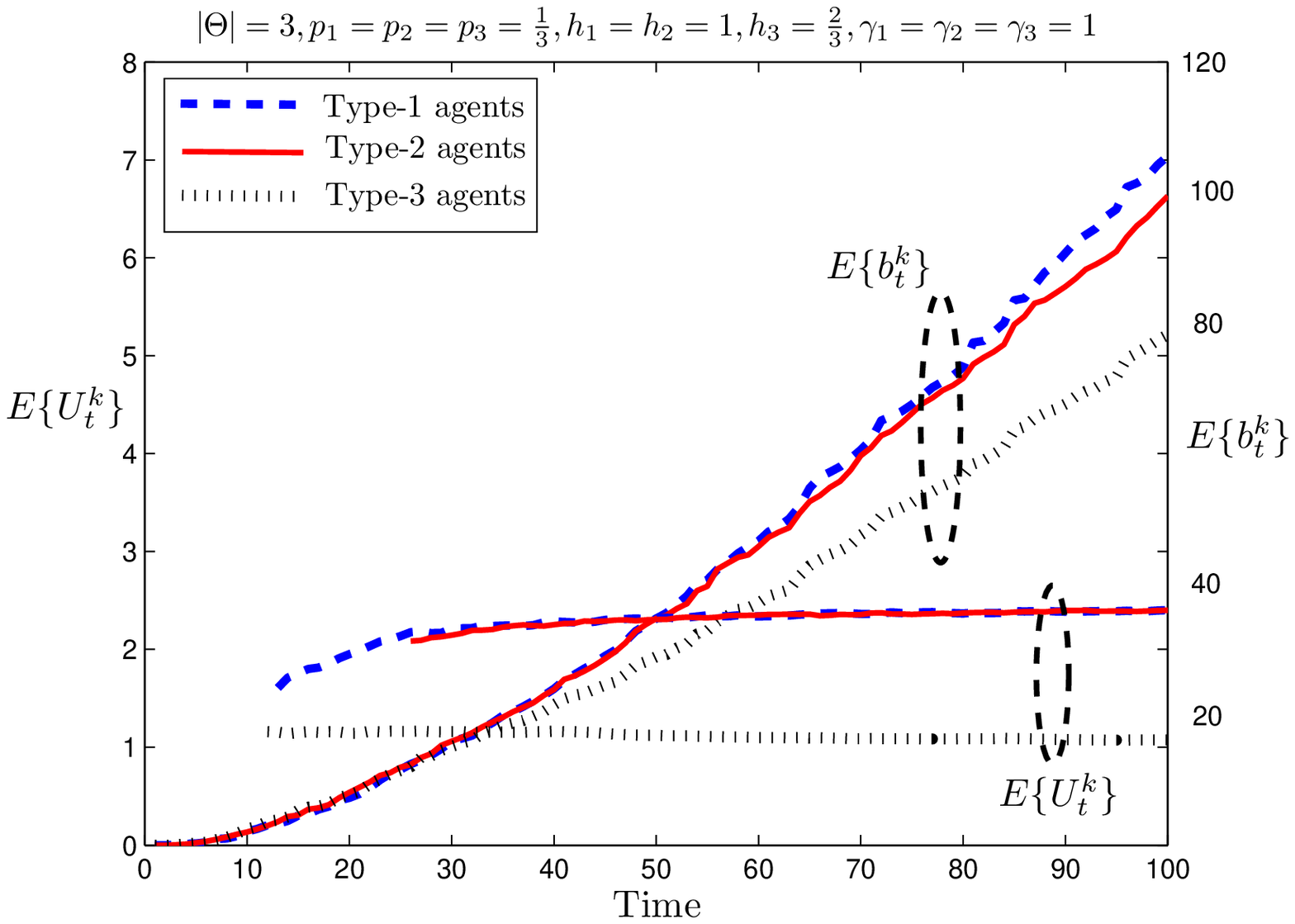}
        \caption{Centrality and utility of non-homophilic, non-exploring agents in a homophilic society.}
    \end{subfigure}
   \caption{Betweenness centrality and average utility of agents in a network that exhibits structural holes.}
\end{figure*}

From the simulation results in Fig. 8, we conclude that {\it homophily creates inequality in the acquisition of bridging capital}, and the different behaviors and norms of different social groups lead to the emergence of different forms of capital. The way that inequality is created in those forms of capital can have very different dependencies on the behaviors of the social groups. When agents in a homophilic group are very sociable, every agent is likely to be popular but not central. That is, socialization increases the bonding capital, but decreases the bridging capital. Moreover, minorities have the same chance as majorities to become popular, yet they have less chances to be central.

The results in Fig. 8 and the discussion above are concerned with the centrality of agents within their social groups. However, a more interesting form of centrality is the one that arises from bridging heterogeneous social groups. In fact, this is the form of social capital that Burt and Putnam have extensively studied in \cite{refcap05} and \cite{ref55sh}. In the following subsection, we introduce a new phenomenon that provides insights into the interplay between centrality and homophily. 

\subsection{Homophily and intergroup bridging}
In this subsection, we study a striking phenomenon that arises from the interplay between homophily and centrality. In particular, we show via simulations that when a social group possesses different homophilic tendencies compared to all other social groups, they end up being the most central group, and thus accrue the largest bridging capital. That is, in an extremely homophilic society, non-homophilic agents bridge segregated social groups, and thus become the most central and gain access to diverse sources of information. On the other hand, a homophilic social group in a non-homophilic society ends up being the most central as they form a highly connected core of the global network structure, which represent an information hub through which all individuals are bridged.

\subsubsection{Filling structural holes: the power of tolerance, open-mindedness, and interdisciplinarity}
As first pointed out by Granovetter in \cite{refcap11}, weak ties (the ties between individuals of different types) have strength as they bridge different segregated social groups. Opinions, beliefs, and ideas are more homogeneous within than between groups, so individuals connected across groups are more exposed to alternative ways of thinking and behaving. In other words, brokerage across the structural holes between homophilic groups provides a vision of new options that are otherwise unseen, which stimulates new ideas and innovation, and also allows agents to control information flow across different groups \cite{ref55sh3}. Thus brokerage creates a social capital, namely a bridging capital, and centrality in such case is gained by agents who link the segregated homophilic groups. 

\begin{figure*}[t!]
    \centering
    \begin{subfigure}[b]{0.3\textwidth}
        \centering
        \includegraphics[width=2.25 in]{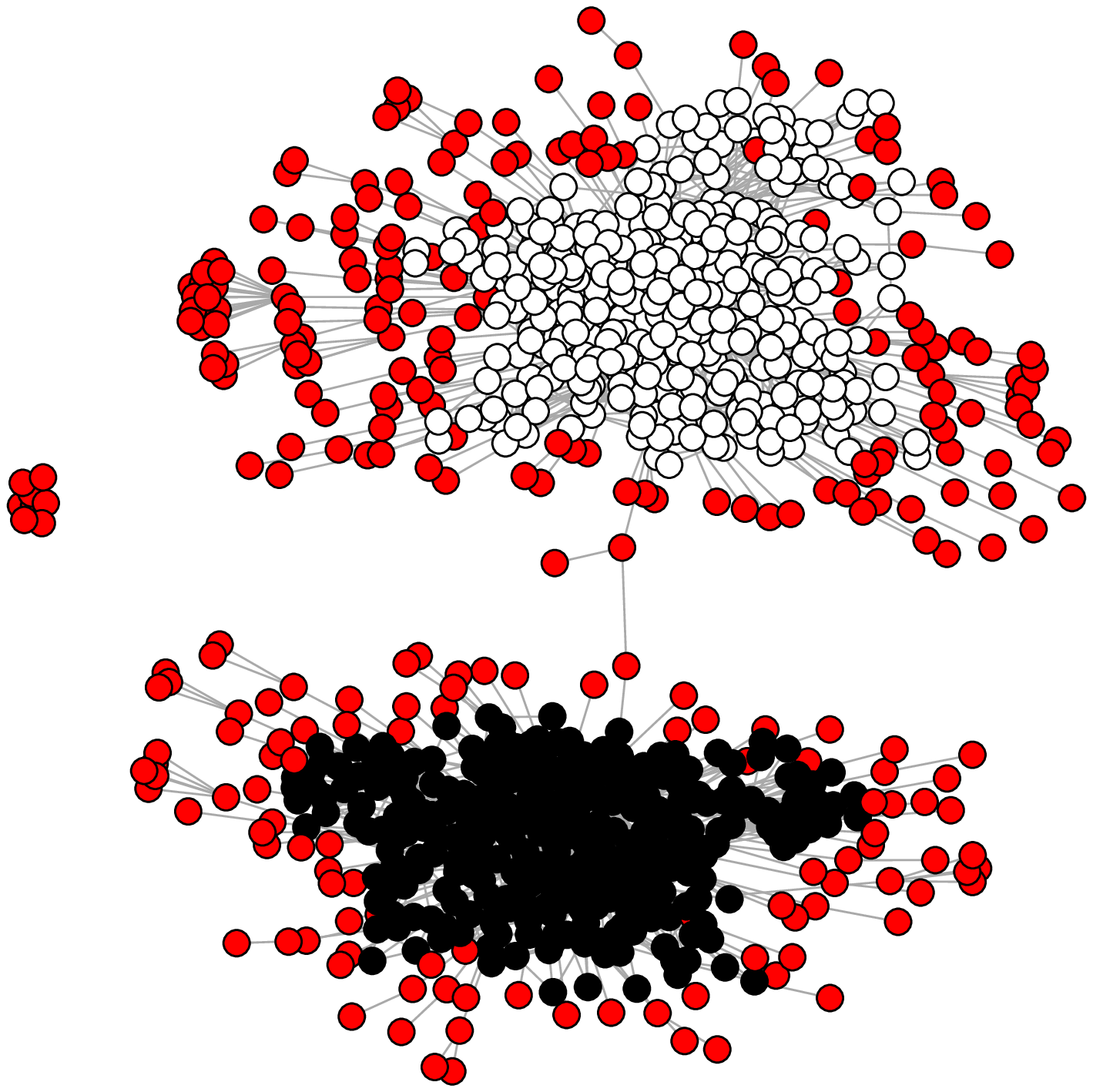}
        \caption{$\gamma_{3} = 1$.}
    \end{subfigure}%
    ~ 
    \begin{subfigure}[b]{0.3\textwidth}
        \centering
        \includegraphics[width=2.25 in]{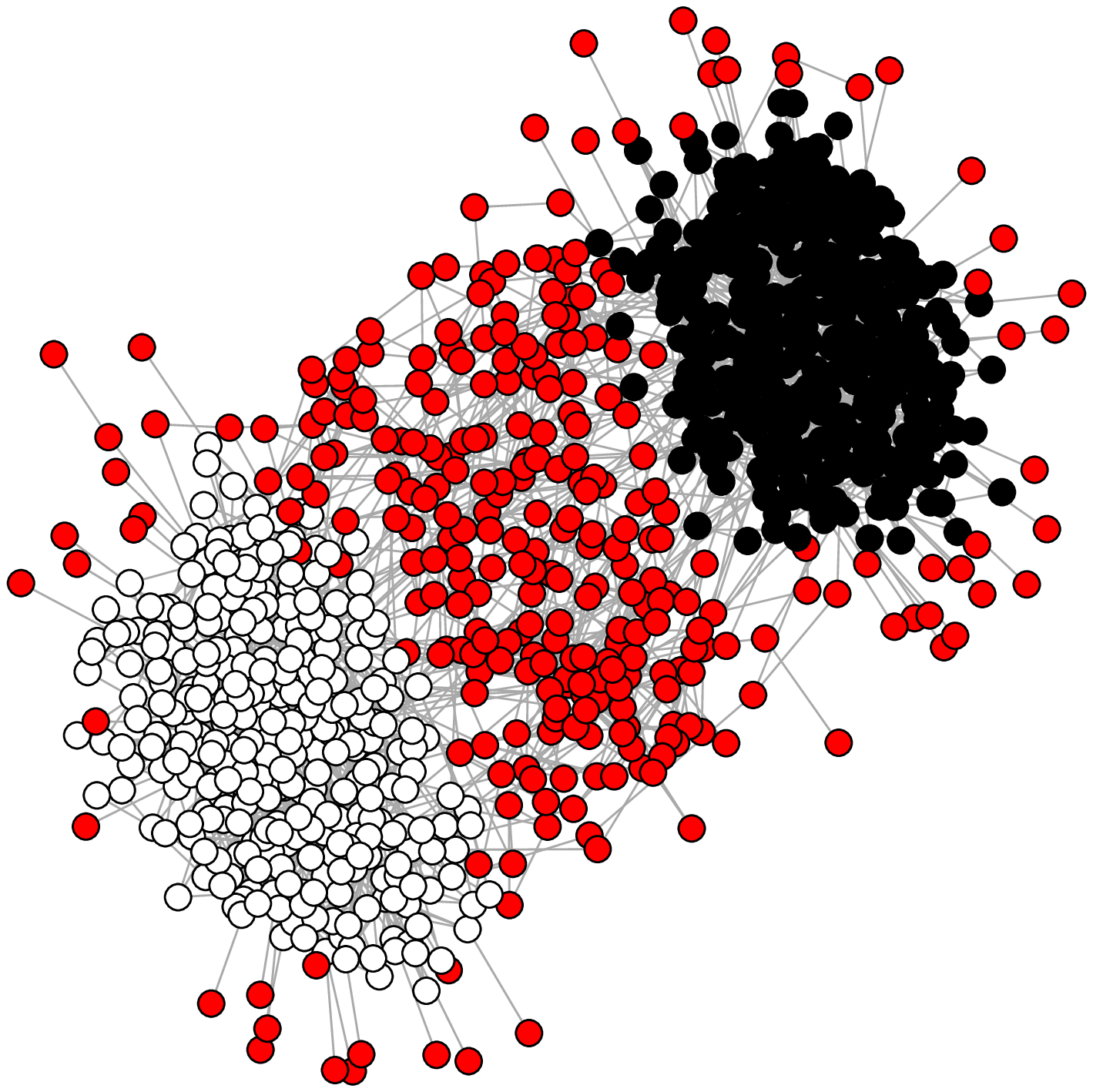}
        \caption{$\gamma_{3} = 0.1$.}
    \end{subfigure}
		~
		\begin{subfigure}[b]{0.3\textwidth}
        \centering
        \includegraphics[width=2.25 in]{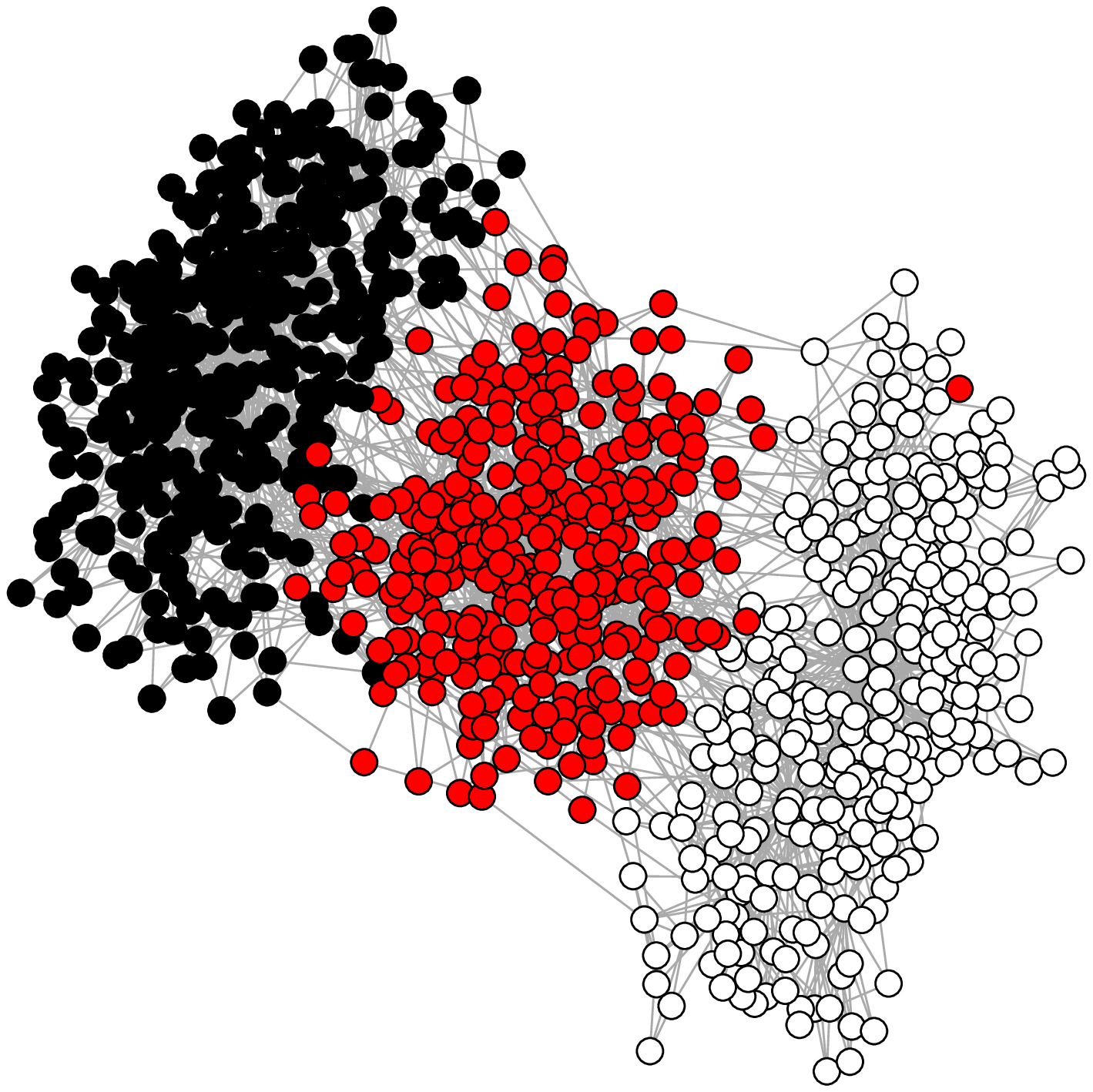}
        \caption{$\gamma_{3} = 0.$.}
    \end{subfigure}
    \caption{A Snapshot for a network with $h_{1}=h_{2}=1$, $\gamma_{1}=\gamma_{2}=1$, and $h_{3}=\frac{1}{3}$. Non-homophilic agents acquire a central position in the network when they are less opportunistic.}
\end{figure*}

In Fig. 9, we demonstrate the interplay between bonding and bridging capital in a network that exhibits structural holes. We carry out a Monte Carlo simulation by simulating 1000 instantiations of the network and plot the average utility and betweenness centrality of each type of agents in the network. We assume that the network has 3 types of agents, where type-1 and type-2 agents are extremely homophilic, whereas type-3 agents have a homophily index of $h_{3}=\frac{1}{3}$. Since type-1 and type-2 agents are extremely homophilic, their bonding capital is maximized, yet both types are disconnected, which creates an opportunity for harvesting bridging capital by type-3 agents since such a type can bridge the two disconnected communities. The ability of type-3 agents to acquire bridging capital depends on their meeting process, i.e. the extent to which they explore the network. If non-homophilic individuals are not exploring the network, then they will end up in a peripheral position in the network, and may not construct their ego networks in finite time (recall Lemma 1). Fig. 9(a) and 9(b) depict the impact of the meeting process on the bridging capital acquired by non-homophilic agents in a homophilic society. It is clear from both figures that there is a tension between the bonding capital (expressed in terms of the average utility), and the bridging capital (expressed in terms of the average centrality). That is to say, homophilic type-1 and type-2 agents acquire higher utility since they enjoy more homogeneous ego networks than type-3 agents. However, when $\gamma_{3}=0$, type-3 agents are more central in the network as they broker the interface between type-1 and type-2 social groups. Contrarily, when $\gamma_{3}=1$, type-3 agents acquire less bonding and bridging capital as they do not explore the network, thus they cannot bridge segregated groups, albeit being non-homophilic.   

Fig. 10 depicts the network structure at $t=1000$ for various meeting processes. In Fig. 10(a), we see that when type-3 agents (red colored) are fully opportunistic, they end up being either marginalized (acquire a peripherial position) or unsatisfied (never forms a satisfactory ego network). When the network exploration rate increases, we see in Fig. 10(b) that only a fraction of non-homophilic agents are prepherial at any time step, yet an intermediate community of such agents emerges and it bridges the otherwise segregated social groups. When $\gamma_{3}=1$, we see in Fig. 10(c) that all non-homophilic agents will reside in the central community and will acquire a central position. Such result provides the following interesting insight: it is not enough for individuals to be non-homophilic, tolerant, or open-minded in order to harvest the bridging capital, but it is essential for them to explore the network structure such that they meet diverse types of agents. Thus, in a society where the meeting process; reflected by policies, norms, regulations, geographical constraints or rules; hinders network exploration, then the existence of non-homophilic individuals does not guarantee that structural holes will be filled. In the following Theorem, we provide the necessary and sufficient conditions for any network to be connected.
\begin{thm}
(Network connectedness) An asymptotically large network is connected almost surely, i.e. $\mathbb{P}\left(\lim_{t\rightarrow\infty}\omega\left(G^{t}\right)=1\right)=1$, if and only if there exists at least one type of agents $k\in\Theta$ with $h_{k}<1$ and $\gamma_{k}<1$. 
\end{thm}
\begin{proof} 
See Appendix M. \, \, \IEEEQEDhere 
\end{proof}
 
Theorem 6 says that the existence of non-homophilic type of agents that explore the network with any non-zero rate will guarantee network connectedness. The condition of $\gamma_{k}<1$ follows from our assumption that agents have infinite lifetimes. If agents have finite lifetimes, then a threshold on $\gamma_{k}$ will decide the network connectedness. That is to say, open-minded individuals will have a threshold on the minimum rate of network exploration that is a function of their lifetime, beyond which they will not be able to fill the structural holes and acquire the largest bridging capital. Thus, non-homophilic agents, who can be thought of as being ``tolerant" or ``open-minded" individuals, can bridge segregated social groups and become the most central in the network when their meeting process involves exploring the network. 

The literature argues that the centrality of non-homophilic (or tolerant and open-minded) individuals play an important role in many networks. For instance, in the context of citation networks, Leydesdorff proposes betweenness centrality as a measure of a journal's ``interdisciplinarity". In addition to the \textit{impact factor} which is a measure of a journal's influence, centrality of a journal indicates the role it plays in promoting innovative and interdisciplinary research, which creates a social capital in the research citation and collaboration networks \cite{refcent7}\cite{refcent8}. Moreover, Burt emphasizes the role of centrality in the diffusion of information \cite{ref55sh}, and the creation of new ideas as a result to the exposure to non redundant sources of information \cite{ref55sh3}. It is worth noting that bridging capital not only leads to egocentric returns to individuals, but also creates a shared value for the network: it stimulates innovative and interdisciplinary research ideas, and allows for the diffusion of information along the global network structure.

\subsubsection{Emergence of information hubs: the power of the dominant coalition}

\begin{figure}[t!]
    \centering
    \includegraphics[width=3.5 in]{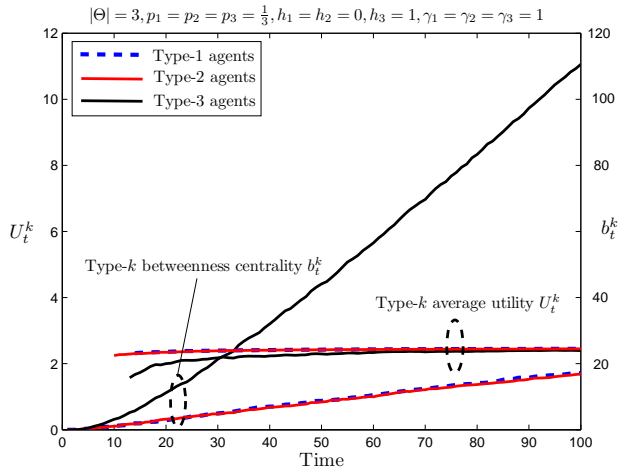}     
    \caption{Betweenness centrality and average utility of agents in a network with a dominant information hub.}
\end{figure}

In the previous subsection, we have shown that non-homophilic agents in a homophilic society acquire the most central network positions and thus attain the highest bridging capital. In this subsection, we show that in the reciprocal scenario where there is one homophilic type of agents in a non-homophilic society, homophilic agents end up being more central than others. In Fig. 11, we plot the average utility and betweenness centrality of 3 types of agents forming a network, where types 1 and 2 agents are extremely non-homophilic, whereas type 3 agents are extremely homophilic. It can be observed that the average centrality of type 3 agents dominates that of types 1 and 2. This is because type 3 agents tend to connect to each other, thus forming a {\it dominant coalition} or an {\it information hub} that resides in the core of the network. The term ``dominant coalition" was coined by Brass in \cite{refbrass} to describe same-gender highly connected influential agents in an organization's interaction network. Unlike the result of the previous subsection, homophilic central agents in a society dominated by non-homophilic types of agents do not bridge structural holes in the network, but rather form a densely connected sub-network through which information is disseminated over the entire network topology. In the context of citation networks, this result predicts that if types corresponds to journals, then a journal that is highly cited and at the same time maintains a self-citation rate that is significantly higher than other journals is likely to form an information hub in a network of papers. Fig. 12 illustrates the formation of an information hub by the extremely homophilic agents in a non-homophilic society, where it can be seen that the type-3 agents form a core sub-network that resides in the center of the global network topology. 

\begin{figure}[t!]
    \centering
    \includegraphics[width=3.5 in]{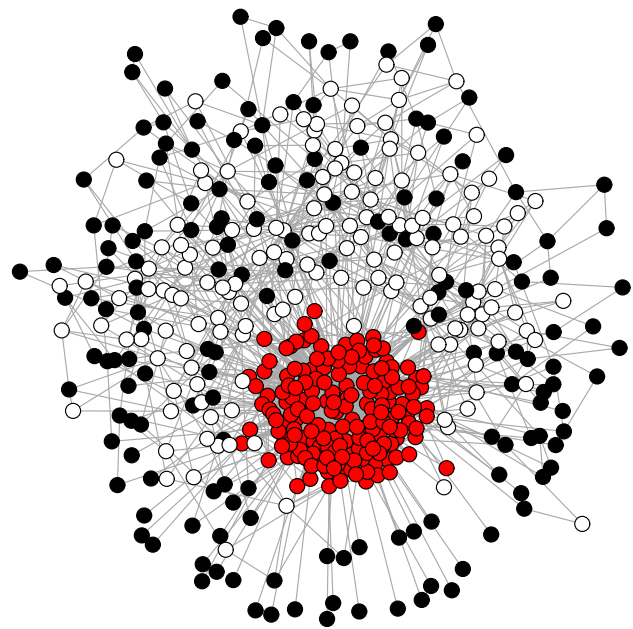}     
    \caption{The formation of an information hub in a network with $h_1 = h_2 = \frac{1}{3},$ and $h_{3} = 1$. The extremely homophilic type-3 agents form a dominant coalition that resides in the core of the network.}
\end{figure}

\section{Conclusions}   
In this paper, we presented a micro-founded mathematical model of the emerging social capital in evolving social networks. In our model, the evolution of the network and of social capital are driven by exogenous and endogenous processes, which are influenced by the extent to which individuals are homophilic, structurally opportunistic, socially gregarious and by the distribution of agents' types in the society. We focused on three different forms of endogenously emerging social capital: bonding, popularity, and bridging capital, and showed how these different forms of capital depend on the exogenous parameters. Bonding capital is maximized in extremely homophilic societies, yet extreme homophily creates structural holes that hinder communications across network components. Popularity capital leads to preferential attachment due to the agents' structural opportunism, which offers agents a cumulative advantage in popularity capital acquisition. Homophily creates inequality in the popularity capital; more gregarious types of agents are more likely to become popular. However, in homophilic societies, individuals who belong to less gregarious, less opportunistic, or major types are likely to be more central in the network and thus acquire a bridging capital. Finally, we studied a striking phenomenon that arises from the interplay between homophily and centrality. In particular, we showed that when a social group that possesses different homophilic tendencies compared to all other social groups, they end up being the most central group, and thus accrue the largest bridging capital. 

\appendices{}
\section{Derivation of the Exogenous Homophily Index}
\renewcommand{\theequation}{\thesection.\arabic{equation}}
\setcounter{mytempeqncnt}{\value{equation}} \setcounter{equation}{0}
From (\ref{eq221}), we know that the exogenous homophily index for type-$k$ agents is given by
\[h_{k} = \lim_{t \rightarrow \infty} \inf_{\mathcal{N}^{+}_{i,t} \in \mathbb{N}^{+}_{i,t}} \frac{N^{s}_{i}(t)}{\mbox{deg}_{i}^{+}(t)}, \forall \theta_{i}=k,\]
which can be rewritten as
\[h_{k} = \lim_{t \rightarrow \infty} \inf_{\mathcal{N}^{+}_{i,t} \in \mathbb{N}^{+}_{i,t}} \frac{N^{s}_{i}(t)}{N^{s}_{i}(t)+N^{d}_{i}(t)}, \forall \theta_{i}=k.\]
The exogenous homophily index can be further rearranged as
\begin{equation}
h_{k} = \lim_{t \rightarrow \infty} \inf_{\mathcal{N}^{+}_{i,t} \in \mathbb{N}^{+}_{i,t}} \frac{1}{1 + \frac{N^{d}_{i}(t)}{N^{s}_{i}(t)}}, \forall \theta_{i}=k.
\label{A1}
\end{equation}
Thus, $h_{k}$ is obtained by finding the set of followees $\mathcal{N}^{+}_{i,t}$ that maximizes $\lim_{t \rightarrow \infty} \frac{N^{d}_{i}(t)}{N^{s}_{i}(t)}$. Note that for an agent $i$, the achieved utility function is given by
\[u^{t}_{i} = v_{\theta_{i}}\left(N^{s}_{i}(t) \alpha^{s}_{\theta_{i}} + N^{d}_{i}(t) \alpha^{d}_{\theta_{i}} \right)-\left(N^{s}_{i}(t)+N^{d}_{i}(t)\right)\, c.\]
Due to the concavity of the utility function, it follows that for any two followee sets $\mathcal{N}^{+}_{i,t}$ and $\bar{\mathcal{N}}^{+}_{i,t} \in \mathbb{N}^{+}_{i,t}$, if $\bar{N}^{d}_{i}(t) > N^{d}_{i}(t),$ then $\bar{N}^{s}_{i}(t) \leq N^{s}_{i}(t)$. This can be easily shown by computing the number of possible same-type followees given a certain number of different-type followees. For instance, assume that $N^{s}_{i}(t)$ and $N^{d}_{i}(t)$ can take non-integer values, and let $g_{\theta_{i}}(x) = \frac{\partial v_{\theta_{i}}(x)}{\partial x}.$ Thus, for any valid value of $N^{d}_{i}(t)$, the value of $N^{s}_{i}(t)$ is given by   
\[N^{s}_{i}(t) = \frac{1}{\alpha^{s}_{\theta_{i}}}\left(g^{-1}_{\theta_{i}}\left(\frac{c}{\alpha^{s}_{\theta_{i}}}\right)-N^{d}_{i}(t)\alpha^{d}_{\theta_{i}}\right).\]
Thus, $N^{s}_{i}(t)$ is a (weakly) decreasing function of $N^{d}_{i}(t)$, and maximizing $\lim_{t \rightarrow \infty} \frac{N^{d}_{i}(t)}{N^{s}_{i}(t)}$ entails maximizing the number of different-type followees $N^{d}_{i}(t)$. Since the utility function is concave, and the marginal benefit of adding a different-type followee is always less than that from adding a similar-type followee (which follows from the assumption of $\alpha^{s}_{\theta_{i}} >  \alpha^{d}_{\theta_{i}}$), then the maximum number of different-type followees in agent $i$'s ego network is equal to the number of different-type followees agent $i$ can link with given that it is not linked to any similar-type followee. Thus, based on the definitions in (\ref{eq301}) and (\ref{eq302}), the maximum number of different-type followees is $\bar{L}_{\theta_{i}}^{*}(0)$, and the corresponding number of similar-type followees is $L_{\theta_{i}}^{*}(\alpha^{d}_{\theta_{i}}\bar{L}_{\theta_{i}}^{*}(0))$. Such followee set materializes for any realization of the meeting process $M_{i}(t) = \{\theta_{1},\theta_{2},.\,.\,.,\theta_{x}, \theta_{x+1},.\,.\,., \theta_{T_{i}}\},$ where $x \geq \bar{L}_{\theta_{i}}^{*}(0), \theta_{j} \neq \theta_{i}, \forall j \leq x$, i.e. agent $i$ meets a ``satisfactory" number of different-type followees first, and then meets the first similar-type followee. The exogenous homophily index for type-$k$ agents is then given by               
\[h_{k} = \frac{L_{k}^{*}(\alpha^{d}_{k}\bar{L}_{k}^{*}(0))}{\bar{L}_{k}^{*}(0)+L_{k}^{*}(\alpha^{d}_{k}\bar{L}_{k}^{*}(0))},\]
where $\bar{L}_{\theta_{i}}^{*}(\alpha)$ can be computed as follows 
\[L_{\theta_{i}}^{*}(\alpha) =   \left\{
     \begin{array}{lr} 
      \left\lfloor L \right\rfloor& : v_{\theta_{i}}\left(L \alpha_{\theta_{i}}^{s}\right)-v_{\theta_{i}}\left((L-1) \alpha_{\theta_{i}}^{s}\right) < c\\
      \left\lceil L \right\rceil& : v_{\theta_{i}}\left(L \alpha_{\theta_{i}}^{s}\right)-v_{\theta_{i}}\left((L-1) \alpha_{\theta_{i}}^{s}\right) > c
     \end{array}
   \right.\]
with $L =  \frac{1}{\alpha_{\theta_{i}}^{s}}\left(g_{\theta_{i}}^{-1}\left(\frac{c}{\alpha_{\theta_{i}}^{s}}\right)-\alpha\right),$ and $g_{\theta_{i}}(x) = \frac{\partial v_{\theta_{i}}(x)}{\partial x}.$ Similarly, $\bar{L}_{\theta_{i}}^{*}(\alpha)$ can be obtained as follows 
\[\bar{L}_{\theta_{i}}^{*}(\alpha) =   \left\{
     \begin{array}{lr} 
      \left\lfloor \bar{L} \right\rfloor& : v_{\theta_{i}}\left(\bar{L} \alpha_{\theta_{i}}^{d}\right)-v_{\theta_{i}}\left((\bar{L}-1) \alpha_{\theta_{i}}^{d}\right) < c\\
      \left\lceil \bar{L} \right\rceil& : v_{\theta_{i}}\left(\bar{L} \alpha_{\theta_{i}}^{d}\right)-v_{\theta_{i}}\left((\bar{L}-1) \alpha_{\theta_{i}}^{d}\right) > c
     \end{array}
   \right.\]
where $\bar{L} = \frac{1}{\alpha_{\theta_{i}}^{d}}\left(g_{\theta_{i}}^{-1}\left(\frac{c}{\alpha_{\theta_{i}}^{d}}\right)-\alpha\right).$


Fig. 11 shows an exemplary utility function with a corresponding exogenous homophily index of $\frac{2}{5}$, which is attained if the agent meets $\bar{L}_{\theta_{i}}^{*}(0) = 2$ different-type agents first, and then starts linking only with same-type agents.

\section{Proof of Lemma 1}
\renewcommand{\theequation}{\thesection.\arabic{equation}}
\setcounter{mytempeqncnt}{\value{equation}} \setcounter{equation}{0}
We start by showing that for any agent $i$, if $\gamma_{\theta_{i}} = 1$ and $0 < h_{\theta_{i}} < 1$, then $\mathbb{P}\left(T_{i}=\infty\left|\theta_{i}\right.\right)>0$. First, since $h_{\theta_{i}} < 1,$ then agent $i$ forms a link with the first agent it meets regardless of its type, i.e. $a_{i}^{i} = 1$ for any $\theta_{m_{i}(i)} \in \Theta$. After that, since $\gamma_{\theta_{i}} = 1$, then agent $i$ keeps meeting agents in the followees of followees choice set $\mathcal{K}_{i,t}$ after it forms the first link. Let $\mathcal{K}^{\theta_{i}}_{i,t} \subseteq \mathcal{K}_{i,t}$ be the set of type-$\theta_{i}$ followees of followees. Thus, one possible event that can lead to an agent getting socially unsatisfied is that $\mathcal{K}^{\theta_{i}}_{i,t}$ becomes an empty set in each time step. Therefore, the probability that agent $i$ gets socially unsatisfied can be lower bounded as follows            
\begin{equation}
\mathbb{P}\left(T_{i}=\infty\left|\theta_{i}\right.\right) \geq \mathbb{P}\left(\left.\bigcup_{t=i+1}^{\infty}\mathcal{K}^{\theta_{i}}_{i,t} = \emptyset \, \right| \, \theta_{i} \right).
\label{eqlowbound}
\end{equation}					
The inequality in (\ref{eqlowbound}) follows from the fact that there are other events that can lead to social unsatisfaction. However, it suffices to show that $\mathbb{P}\left(\left.\bigcup_{t=i+1}^{\infty}\mathcal{K}^{\theta_{i}}_{i,t} = \emptyset \right|\theta_{i}\right)>0$ in order to prove that $\mathbb{P}\left(T_{i}=\infty\left|\theta_{i}\right.\right)>0$. Note that $\mathbb{P}\left(\left.\bigcup_{t=i+1}^{\infty}\mathcal{K}^{\theta_{i}}_{i,t} = \emptyset \right|\theta_{i}\right)$ can be lower bounded as follows 
\[\mathbb{P}\left(\left.\bigcup_{t=i+1}^{\infty}\mathcal{K}^{\theta_{i}}_{i,t} = \emptyset \right|\theta_{i}\right)>\mathbb{P}\left(\left\{\theta_{m_{i}(j)}\neq\theta_{i}\right\}_{j=i}^{i+\bar{L}_{\theta_{i}}^{*}(0)-1}\right)\]
\[\mathbb{P}\left(\left.\bigcup_{t=i}^{i+\bar{L}_{\theta_{i}}^{*}(0)-1}\mathcal{N}^{+,\theta_{i}}_{m_{i}(t),t} = \emptyset\right|\theta_{i},\left\{\theta_{m_{i}(j)}\neq\theta_{i}\right\}_{j=i}^{i+\bar{L}_{\theta_{i}}^{*}(0)-1}\right).\]
Each agent of type other than $\theta_{i}$ have a non-zero probability of having no type-$\theta_{i}$ agents in their followees set. To see why this is true, consider any agent $j$ with type $\theta_{j} \neq \theta_{i}$. Such an agent can have a followee set that contains no type-$\theta_{i}$ agents with a non-zero probability, which can happen when agent $j$ meets $L_{\theta_{j}}^{*}(0)$ strangers in sequence, and all such agents turn out to be type-$\theta_{j}$ agents. Such event happens with a probability that is lower bounded as follows
\[\mathbb{P}\left(N^{s}_{j}(t) = L_{\theta_{j}}^{*}(0)\left|\theta_{j}\right.\right) > (1-\gamma_{\theta_{j}})^{L_{\theta_{j}}^{*}(0)-1}\left(\frac{1}{p_{\theta_{j}}}\right)^{L_{\theta_{j}}^{*}(0)},\] 
which is always positive for $\gamma_{\theta_{j}} < 1$. We further show that even for $\gamma_{\theta_{j}} = 1$, any agent $j$ has a positive probability for not connecting to any type-$\theta_{i}$ agent. For instance, agent $j$ can initially connect to a type-$\theta_{j}$ agent, say agent $k$, which happens with a probability of $\frac{1}{p_{\theta_{j}}}.$ Agent $k$ in turn may have connected initially to another type-$\theta_{j}$ agent, and such an agent may also have connected initially to another type-$\theta_{j}$ agent, and so on. Thus, if at each time step agent $j$ meets its type-$\theta_{j}$ followee of followee whom its followee has met initially in the network, then agent $j$ can end up being connected to a set of exclusively same-type agents. It can be easily shown that this happens with a probability that is lower bounded by             
\[\mathbb{P}\left(N^{s}_{j}(t) = L_{\theta_{j}}^{*}(0)\left|\theta_{j}\right.\right) > \frac{p_{\theta_{j}}^{L_{\theta_{j}}^{*}(0)}}{\prod_{m=0}^{L_{\theta_{j}}^{*}(0)-1}\left((m+1) \bar{L}_{\theta_{j}}^{*}(0)- m\right)}.\] 
Since agent $i$ can initially link to a different-type agent, and since any different-type agent can have a set of followees with no type-$\theta_{i}$ agents, it follows that the lower bound on $\mathbb{P}\left(\left.\bigcup_{t=i+1}^{\infty}\mathcal{K}^{\theta_{i}}_{i,t} = \emptyset \right|\theta_{i}\right)$ is greater than zero, thus $\mathbb{P}\left(\left.\bigcup_{t=i+1}^{\infty}\mathcal{K}^{\theta_{i}}_{i,t} = \emptyset \right|\theta_{i}\right) >0$. 

Now we prove the converse, and show that if $\mathbb{P}\left(T_{i}=\infty\left|\theta_{i}\right.\right)>0$, then $0 < h_{\theta_{i}} < 1$ and $\gamma_{\theta_{i}} = 1$. First, since $\mathbb{P}\left(T_{i}=\infty\left|\theta_{i}\right.\right)>0$, then $L_{k}^{*}(\alpha^{d}_{k}\bar{L}_{k}^{*}(0))>0$, i.e. agent $i$ must add at least one similar-type followee in order to get socially unsatisfied, and linking with $\bar{L}_{k}^{*}(0)$ different-type followees does not suffice to saturate the utility function and terminate the meeting process, thus $h_{\theta_{i}}>0$. Now assume that $h_{\theta_{i}} = 1$. In this case, agent $i$ forms its first link only when it meets a similar-type agent, and after that it meets a similar-type agent picked from $\mathcal{K}_{i,t}$, or meets an agent with an uncertain type picked uniformly at random from the network. At each time step after agent $i$ forms its first link, the probability that it meets a similar-type agent is given by
\[\mathbb{P}\left(\theta_{m_{i}(t)} = \theta_{i}\left|\mbox{deg}_{i}^{+}(t)>1\right.\right) = \gamma_{\theta_{i}} \mathbb{I}_{\left\{\mathcal{K}_{i,t} \neq \emptyset \right\}} + \frac{(1-\gamma_{\theta_{i}})}{t}\left|\mathcal{V}^{t}_{\theta_{i}}\right|,\]
which for a large $t$ converges to
\[\lim_{t \rightarrow \infty} \mathbb{P}\left(\theta_{m_{i}(t)} = \theta_{i}\left|\mbox{deg}_{i}^{+}(t)>1\right.\right) = \gamma_{\theta_{i}} + (1-\gamma_{\theta_{i}}) p_{\theta_{i}},\]      
which is always non-zero for any value of $\gamma_{\theta_{i}}$. Thus, an agent with $h_{\theta_{i}} = 1$ has a non-zero probability to meet a similar-type agent at each time step, which implies that $\lim_{t \rightarrow \infty} \mbox{deg}_{i}^{+}(t) = L_{\theta_{i}}^{*}(0),$ and $\mathbb{P}\left(T_{i}=\infty\left|\theta_{i},h_{\theta_{i}}=1\right.\right)=0$. Therefore, $\mathbb{P}\left(T_{i}=\infty\left|\theta_{i}\right.\right)>0$ implies that $0 < h_{\theta_{i}} < 1$. Moreover, we know that any agent with $\gamma_{\theta_{i}} < 1$ will experience the following meeting process in a large enough network
\[\mathbb{P}\left(\theta_{m_{i}(t)} = k\left|\mbox{deg}_{i}^{+}(t)>1\right.\right) = \gamma_{\theta_{i}} \frac{\left|\mathcal{K}^{k}_{i,t}\right|}{\left|\mathcal{K}_{i,t}\right|} + (1-\gamma_{\theta_{i}})p_{k},\]
which is lower bounded by        
\[\mathbb{P}\left(\theta_{m_{i}(t)} = k\left|\mbox{deg}_{i}^{+}(t)>1\right.\right) \geq (1-\gamma_{\theta_{i}})p_{k}.\]
Since $(1-\gamma_{\theta_{i}})p_{k} > 0, \forall k \in \Theta$, any agent with $\gamma_{\theta_{i}} < 1$ has a non-zero probability for meeting a similar-type agent at each time step, which means that such an agent is not socially unsatisfied in the almost sure sense. Thus, $\mathbb{P}\left(T_{i}=\infty\left|\theta_{i}\right.\right)>0$ implies that $\gamma_{\theta_{i}} = 1$.

\section{Proof of Theorem 1}
\renewcommand{\theequation}{\thesection.\arabic{equation}}
\setcounter{mytempeqncnt}{\value{equation}} \setcounter{equation}{0}
We divide this proof into two parts. First, we prove that for non-homophilic societies, the EFT of an agent $i$ is equal to $L_{\theta_{i}}^{*}(0)$. Next, we show that in homophilic societies, the distribution of EFTs for the agents in a large network converges to a steady-state distribution.\\
{\bf EFT for non-homophilic societies:}\\
Recall that the exogenous homophily index of an agent $i$ is given by
\begin{equation}
h_{\theta_{i}} = \frac{L_{\theta_{i}}^{*}\left(\alpha_{\theta_{i}}^{d}\bar{L}_{\theta_{i}}^{*}\left(0\right)\right)}{\bar{L}_{\theta_{i}}^{*}(0)+L_{\theta_{i}}^{*}\left(\alpha_{\theta_{i}}^{d}\bar{L}_{\theta_{i}}^{*}(0)\right)}.
\label{B1}
\end{equation}
We start by studying the case when $h_{\theta_{i}}=0$. From (\ref{B1}), we know that if $h_{\theta_{i}}=0$, then $L_{\theta_{i}}^{*}\left(\alpha_{\theta_{i}}^{d}\bar{L}_{\theta_{i}}^{*}\left(0\right)\right) = 0,$ which means that $L_{\theta_{i}}^{*}\left(0\right) = \bar{L}_{\theta_{i}}^{*}\left(0\right)$, i.e. agent $i$ forms a link with any agent it meets over time as long as its utility function is not yet saturated. Therefore, at any date $t \geq i$ we have
\[\mathbb{P}\left(a^{t}_{i} = 1 \left|\theta_{m_{i}(t)}\right.\right) = \left\{
     \begin{array}{lr} 
      1 &: t \leq L_{\theta_{i}}^{*}\left(0\right)\\ 
      0 &: t > L_{\theta_{i}}^{*}\left(0\right). 
     \end{array}
   \right.\]
Thus, agent $i$ forms a link with all the agents it meets until its utility function is saturated, which happens after $L_{\theta_{i}}^{*}(0)$ time steps almost surely, i.e. 
\[\mathbb{P}\left(T_{i} = L_{\theta_{i}}^{*}(0)\right) = 1.\]
The EFT of a non-homophilic agent $i$ is independent of the network structure and the types of agents it meets, i.e. $\mathbb{P}\left(a^{t}_{i} = 1 \left|\theta_{m_{i}(t)}, G^{t}\right.\right) = \mathbb{I}_{\left\{t \leq L_{\theta_{i}}^{*}\left(0\right)\right\}}.$ The EFT of agent $i$ is equal to $L_{\theta_{i}}^{*}(0)$ almost surely since agent $i$ meets other agents at a constant rate and forms links with them regardless of their types.\\ 
{\bf EFT for homophilic societies:}\\
Now we focus on the case when $h_{\theta_{i}}=1$. In this case, we have $\bar{L}_{\theta_{i}}^{*}(0) = 0, \forall \theta_{i} \in \Theta$, i.e. agents form links with same-type agents only. At any time step, agent $i$ forms a link with the agent it meets if and only if the agent it meets is a same-type agent and agent $i$'s utility function is not saturated, thus each agent forms exactly $L_{\theta_{i}}^{*}(0)$ links. The rest of the proof is organized as follows: we first derive the probability that an agent forms a link in a given time step, and then we show that for a large network, this probability becomes independent of the network topology, which implies that the ego network formation process converges to a stationary process.  

In the following, we derive the probability that an agent forms a link in a given time step. Note that the probability that an agent $i$ forms a link at a given time step conditioned on the current step graph (current network topology) and the agent it meets can be written as  
\begin{equation}
\mathbb{P}\left(a^{t}_{i} = 1 \left|\theta_{m_{i}(t)},G^{t}\right.\right) = \mathbb{I}_{\left\{\theta_{m_{i}(t)} = \theta_{i}, \, \mbox{deg}_{i}^{+}(t) \leq L_{\theta_{i}}^{*}(0)\right\}}.
\label{thm1eq1}
\end{equation}
From (\ref{thm1eq1}), and using Bayes rule, the probability that agent $i$ forms a link at time $t$ conditioned on the current step graph is given by
\begin{equation}
\mathbb{P}\left(a^{t}_{i} = 1 \left|G^{t}\right.\right) = \mathbb{P}\left(\theta_{m_{i}(t)} = \theta_{i}\left|G^{t}\right.\right) \, \mathbb{I}_{\left\{\mbox{deg}_{i}^{+}(t) \leq L_{\theta_{i}}^{*}(0)\right\}}.
\label{thm1eq2}
\end{equation} 
Note that the meeting process of agent $i$ goes throught two stages (see Section 2). Upon its arrival, and until it becomes attached to the network by forming its first link, agent $i$ meets agents picked uniformly at random from the network, and it forms its first link only if it meets a type-$\theta_{i}$ agent, therefore, for $\mbox{deg}_{i}^{+}(t) = 0$, (\ref{thm1eq2}) can be written as
\begin{equation}
\mathbb{P}\left(a^{t}_{i} = 1 \left|\mbox{deg}_{i}^{+}(t) = 0, G^{t}\right.\right) = \frac{\left|\mathcal{V}_{\theta_{i}}^{t}\right|-1}{\left|\mathcal{V}^{t}\right|-1}.
\label{thm1eq3}
\end{equation}
After agent $i$ becomes attached to the network (i.e. $\mbox{deg}_{i}^{+}(t) > 0$), it starts meeting other agents picked from two choice sets: the set of strangers $\bar{\mathcal{K}}_{i,t}$, and the set of followees of followees $\mathcal{K}_{i,t}$. We know from the definition of the meeting process in Section 2 that the probability that agent $i$ meets a type-$\theta_{i}$ agent is given by 
\[\mathbb{P}\left(\theta_{m_{i}(t)}=\theta_{i} \left| \mbox{deg}_{i}^{+}(t) > 0, G^{t}\right.\right) =\]
\[\left((1-\gamma_{\theta_{i}})(1-\mathbb{P}(\mathcal{K}_{i,t} = \emptyset))+\mathbb{P}(\mathcal{K}_{i,t} = \emptyset)\right) \times \]
\[\mathbb{P}\left(\theta_{m_{i}(t)}=\theta_{i}\left|m_{i}(t) \in \bar{\mathcal{K}}_{i,t} \cup \mathcal{K}_{i,t}, G^{t}\right.\right) +\]
\begin{equation}
\gamma_{\theta_{i}} (1-\mathbb{P}(\mathcal{K}_{i,t} = \emptyset)) \, \mathbb{P}\left(\theta_{m_{i}(t)}=\theta_{i}\left|m_{i}(t) \in \mathcal{K}_{i,t}, G^{t}\right.\right),
\label{E1}
\end{equation}
which can be simplified as follows
\[\mathbb{P}\left(\theta_{m_{i}(t)}=\theta_{i} \left| \mbox{deg}_{i}^{+}(t) > 0, G^{t}\right.\right) = \]
\[\frac{\gamma_{\theta_{i}}K^{s}_{i}(t)}{K_{i}(t)} \, (1-\mathbb{P}(K_{i}(t) = 0)) + \]
\begin{equation}
\left(1-\gamma_{\theta_{i}}+\gamma_{\theta_{i}}\mathbb{P}(K_{i}(t) = 0)\right) \frac{\left|\mathcal{V}_{\theta_{i}}^{t}\right|-N^{s}_{i}(t)-1}{\left|\mathcal{V}^{t}\right|-\mbox{deg}_{i}^{+}(t)-1}.
\label{thm1eq5}
\end{equation}
The probability that the choice set $\mathcal{K}_{i,t}$ becomes empty at any time step can be expressed as follows. First, note that if the event $K_{i}(t) = 0$ happens, then agent $i$ should have been connected initially to an agent that has not yet constructed its ego network, i.e. $N_{m_{i}(\tau)}(\tau) < L_{\theta_{i}}^{*}(0)$ if $i$ has formed its first link at time $\tau$. This is because otherwise we will have $K_{i}(t) \geq L_{\theta_{i}}^{*}(0), \forall t \geq \tau$, which implies that $\mathbb{P}(K_{i}(t) = 0) = 1$ at every time step. At any point of time $t>i$, the probability that the choice set $\mathcal{K}_{i,t}$ becomes empty is equal to the probability that $N_{m_{i}(\tau)}(\tau) < L_{\theta_{i}}^{*}(0)$ and the probability that the new linking actions of agents $i$ and $m_{i}(\tau)$ has not led to the emergence of new followees of followees. Therefore, we can bound $\mathbb{P}(K_{i}(t) = 0)$ as follows 
\[\mathbb{P}(K_{i}(t) = 0) \leq \frac{1}{\left|\mathcal{V}^{t}_{\theta_{i}}\right|}\sum_{m \in \mathcal{V}^{t}_{\theta_{i}}}\mathbb{I}_{\left\{t<T_{m}\right\}}.\]
That is, the probability that the followees of followees choice set of agent $i$ becomes empty is always less than the probability that agent $i$ initially links to an agent with an unsatisfied ego network. Therefore, the probability of forming a link at any time step conditioned on the current network topology is given by (\ref{eqlong}).  

\begin{figure*}[!t]
\setcounter{mytempeqncnt}{\value{equation}} \setcounter{equation}{6}
\begin{equation}
\mathbb{P}\left(a^{t}_{i} = 1 \left|G^{t}\right.\right) = \left\{
     \begin{array}{lr} 
      \frac{\left|\mathcal{V}_{\theta_{i}}^{t}\right|-N^{s}_{i}(t)-1}{\left|\mathcal{V}^{t}\right|-\mbox{deg}_{i}^{+}(t)-1} &: \mbox{deg}_{i}^{+}(t)=0\\ 
      \frac{\gamma_{\theta_{i}}K^{s}_{i}(t)}{K_{i}(t)} \, (1-\mathbb{P}(K_{i}(t) = 0)) +  \left(1-\gamma_{\theta_{i}}+\gamma_{\theta_{i}}\mathbb{P}(K_{i}(t) = 0)\right) \frac{\left|\mathcal{V}_{\theta_{i}}^{t}\right|-N^{s}_{i}(t)-1}{\left|\mathcal{V}^{t}\right|-\mbox{deg}_{i}^{+}(t)-1}&: 0<\mbox{deg}_{i}^{+}(t)\leq L_{\theta_{i}}^{*}(0)\\
			0 &: \mbox{deg}_{i}^{+}(t) > L_{\theta_{i}}^{*}(0)
     \end{array}
   \right.
\label{eqlong}
\end{equation}
\setcounter{equation}{\value{mytempeqncnt}+1} \hrulefill{}\vspace*{4pt}
\end{figure*}
\begin{figure*}[!t]
\setcounter{mytempeqncnt}{\value{equation}} \setcounter{equation}{7}
\begin{equation}
\lim_{t \rightarrow \infty} \mathbb{P}\left(a^{t}_{i} = 1 \left|G^{t}\right.\right) = \left\{
     \begin{array}{lr} 
      p_{\theta_{i}} &: \mbox{deg}_{i}^{+}(t)=0\\ 
      \gamma_{\theta_{i}}  +  \left(1-\gamma_{\theta_{i}}\right) p_{\theta_{i}}&: 0<\mbox{deg}_{i}^{+}(t)\leq L_{\theta_{i}}^{*}(0)\\
			0 &: \mbox{deg}_{i}^{+}(t) > L_{\theta_{i}}^{*}(0)
     \end{array}
   \right.
\label{eqlong3}
\end{equation}
\setcounter{equation}{\value{mytempeqncnt}+1} \hrulefill{}\vspace*{4pt}
\end{figure*}
Note that the expressions in (\ref{eqlong}) depend on the actual realization of the graph process at time $t$, i.e. the step graph $G^{t}$. In the following, we show that this dependency vanishes when the network is asymptotically large. First, since any type-$k$ agent has $h_{k}=1$, then all the followees of followees for an agent $i$ has a type $\theta_{i}$, i.e. $\mathbb{P}\left(\left. \frac{K^{s}_{i}(t)}{K_{i}(t)} = 1 \right|\mathcal{K}_{i,t} \neq \emptyset\right) = 1$. Moreover, in a large network we have $\mathbb{P}(K_{i}(t) = 0) \leq  \lim_{t \rightarrow \infty} \frac{1}{\left|\mathcal{V}^{t}_{\theta_{i}}\right|}\sum_{m \in \mathcal{V}^{t}_{\theta_{i}}}\mathbb{I}_{\left\{t<T_{m}\right\}}$, which is equivalent to $\mathbb{P}(K_{i}(t) = 0) \leq  \lim_{t \rightarrow \infty} \frac{1}{p_{\theta_{i}}t}\sum_{m \in \mathcal{V}^{t}_{\theta_{i}}}\mathbb{I}_{\left\{t<T_{m}\right\}}$. Since all agents are extremely homophilic, they have finite EFTs (recall Lemma 1), which means that $\lim_{t \rightarrow \infty} \frac{1}{p_{\theta_{i}}t}\sum_{m \in \mathcal{V}^{t}_{\theta_{i}}}\mathbb{I}_{\left\{t<T_{m}\right\}} = 0$. Thus, in a large network $\mathbb{P}(K_{i}(t) = 0) = 0,$ and $K_{i}(t)$ is bounded by $L_{\theta_{i}}^{*}(0) \leq K_{i}(t) \leq \left(L_{\theta_{i}}^{*}(0)\right)^{2}$. Furthermore, we have that
\begin{align}
\lim_{t \rightarrow \infty} \frac{\left|\mathcal{V}_{\theta_{i}}^{t}\right|-N^{s}_{i}(t)-1}{\left|\mathcal{V}^{t}\right|-\mbox{deg}_{i}^{+}(t)-1} &= \lim_{t \rightarrow \infty} \frac{p_{\theta_{i}}t-L_{\theta_{i}}^{*}(0)-1}{t-L_{\theta_{i}}^{*}(0)-1} \nonumber\\
&= p_{\theta_{i}}. 
\end{align}
This leads to the expressions in (\ref{eqlong3}). It is clear from (\ref{eqlong3}) that for a large network, the probability of taking a link formation decision at any time step depends only on the current number of followees of agent $i$. Thus, linking decisions depend only on agent $i$'s ego network, and are independent on the global network structure.               

Let $N^{j}_{i}$ for $j>1$, be the waiting time between forming link $j-1$ and link $j$ by agent $i$, and $N^{1}_{i}$ be the waiting time between forming the first link and agent $i$'s birth date. Thus, the EFT is given by $T_{i} = \sum_{j=1}^{L_{\theta_{i}}^{*}(0)} N^{j}_{i}$. Note that $N^{1}_{i}$ is the number of agents met by agent $i$ before it forms its first link. Therefore, when agent $i$ is singleton, every meeting will result in the formation of the first link with a probability of $p_{\theta_{i}}$ independent on the previous meetings, which means that $N^{1}_{i}$ is a {\it geometric random variable} with a success probability of $p_{\theta_{i}}$, and $\mathbb{E}\left\{N^{1}_{i}\right\} = \sum_{m=1}^{\infty} m p_{\theta_{i}}\left(1-p_{\theta_{i}}\right)^{m-1} = \frac{1}{p_{\theta_{i}}}$. Moreover, for a large network, the probability of forming a link at time step $t$ is 
\[
\mathbb{P}\left(\theta_{m_{i}(t)}=\theta_{i}\left|\mbox{deg}_{i}^{+}(t) > 0, G^{t}\right.\right) = (1-\gamma_{\theta_{i}}) \, p_{\theta_{i}} + \gamma_{\theta_{i}}. 
\]
Thus, after it forms its first link, agent $i$ needs to form $L_{\theta_{i}}^{*}(0)-1=L_{\theta_{i}}^{*}(\alpha^{s}_{\theta_{i}})$ links, and the probability of forming a link at any time step is $(1-\gamma_{\theta_{i}}) \, p_{\theta_{i}} + \gamma_{\theta_{i}}$, which is independent of the network topology and the history of actions of agent $i$. Thus, $N_{i}^{j}$ is a geometric random variable with a success probability of $(1-\gamma_{\theta_{i}}) \, p_{\theta_{i}} + \gamma_{\theta_{i}}$, where $N_{i}^{j}$ and $N_{i}^{m}$ are i.i.d (this means that for a large network and extremely homophilic agents, the actions of an agent do not affect the meeting process). Therefore, the distribution of EFT for any agent $i$ in an asymptotically large network follows a fixed distribution, which follows from obtaining the distribution of $\sum_{j=1}^{L_{\theta_{i}}^{*}(0)}N_{i}^{j}$, where $N_{i}^{1}$ is a geometric random variable with a success probability of $p_{\theta_{i}}$, and for $j>1$, $N_{i}^{j}$ is a geometric random variable with a success probability of $(1-\gamma_{\theta_{i}}) \, p_{\theta_{i}} + \gamma_{\theta_{i}}$. Given that the random variables $\left\{N_{i}^{2}, N_{i}^{3},.\,.\,.,N_{i}^{L_{\theta_{i}}^{*}(0)}\right\}$ are i.i.d, the pmf of the sum $\sum_{j=2}^{L_{\theta_{i}}^{*}(0)} N_{i}^{j}$ can be easily evaluated by taking the product of the Moment Generating Functions (MGF) of the individual random variables. The MGF of $N_{i}^{j}$ is given by $S_{N_{i}^{j}}(\Omega) = \mathbb{E}\left\{e^{\Omega N_{i}^{j}}\right\}, \Omega \in \mathbb{R},$ which can be obtained as follows
\[S_{N_{i}^{j}}(\Omega) = \frac{((1-\gamma_{\theta_{i}}) p_{\theta_{i}} + \gamma_{\theta_{i}}) e^{\Omega}}{1-(1-(1-\gamma_{\theta_{i}}) p_{\theta_{i}} - \gamma_{\theta_{i}}) e^{\Omega}},\]    
for $\Omega < -\log(1-((1-\gamma_{\theta_{i}}) p_{\theta_{i}} + \gamma_{\theta_{i}}))$. Thus, the MGF of $\sum_{j=2}^{L_{\theta_{i}}^{*}(0)} N_{i}^{j}$ is given by $\prod_{j=2}^{L_{\theta_{i}}^{*}(0)} S_{N_{i}^{j}}(\Omega),$ which can be written as $\left(\frac{((1-\gamma_{\theta_{i}}) p_{\theta_{i}} + \gamma_{\theta_{i}}) e^{\Omega}}{1-(1-(1-\gamma_{\theta_{i}}) p_{\theta_{i}} - \gamma_{\theta_{i}}) e^{\Omega}}\right)^{L_{\theta_{i}}^{*}(0)-1},$ which corresponds to the MGF of a {\it negative binomial random variable}. Let $\bar{N}^{1}_{i} = \sum_{j=2}^{L_{\theta_{i}}^{*}(0)} N_{i}^{j}$. The pmf of $\bar{N}^{1}_{i}$ is given by
\begin{equation}
f_{\bar{N}^{1}_{i}}\left(\bar{N}^{1}_{i}\right) = \binom{\bar{N}^{1}_{i}-1}{L_{\theta_{i}}^{*}(0)-2}\, p^{L_{\theta_{i}}^{*}(0)-1} \, \left(1-p\right)^{\bar{N}^{1}_{i}-L_{\theta_{i}}^{*}(0)+1},
\label{D3}
\end{equation}   
where $p=\left((1-\gamma_{\theta_{i}}) p_{\theta_{i}} + \gamma_{\theta_{i}}\right)$, and the pmf of $N^{1}_{i}$ is given by
\begin{equation}
f_{N^{1}_{i}}\left(N^{1}_{i}\right) = p \, \left(1-p\right)^{N^{1}_{i}-1}.
\label{D5}
\end{equation}
Thus, the pmf of $T_{i}$ is obtained by computing the convolution of $f_{\bar{N}^{1}_{i}}\left(\bar{N}^{1}_{i}\right)$ and $f_{N^{1}_{i}}\left(N^{1}_{i}\right)$ as follows    
\begin{equation}
f_{T_{i}}\left(T_{i}\right) = f_{N^{1}_{i}}\left(N^{1}_{i}\right) \star f_{\bar{N}^{1}_{i}}\left(\bar{N}^{1}_{i}\right),
\label{D7}
\end{equation}
where $\star$ is the convolution operator. Therefore, the distribution of the EFT for an agents of type $k$ converge to a steady-state distribution, i.e. $\lim_{t \rightarrow \infty} f_{T_{i}}\left(T_{i}\left|\theta_{i} = k\right.\right) = f^{k}_{T}\left(T\right),$ where $f^{k}_{T}\left(T\right) = f_{N^{1}_{i}}\left(N^{1}_{i}\left|\theta_{i} = k\right.\right) \star f_{\bar{N}^{1}_{i}}\left(\bar{N}^{1}_{i}\left|\theta_{i} = k\right.\right)$. Note that from Scheffe's lemma, convergence of the probability mass functions implies convergence in distribution, thus the sequence of EFTs converges in distribution for all types of agents.   

Now we compute the EEFT, which is simply given by $\overline{T}_{i} = \mathbb{E}\left[\sum_{j=1}^{L_{\theta_{i}}^{*}(0)}N_{i}^{j}\right]$. Thus, we have
\[\mathbb{E}\left[N^{j}_{i}\right] = \left\{
     \begin{array}{lr} 
      \frac{1}{p_{\theta_{i}}} &: j = 1,\\ 
       \frac{1}{(1-\gamma_{\theta_{i}}) \, p_{\theta_{i}} + \gamma_{\theta_{i}}} &: 2 \leq j \leq L_{\theta_{i}}^{*}\left(0\right). 
     \end{array}
   \right.\]
 Therefore, the EEFT is given by
\begin{align}
\overline{T}_{i} &= \mathbb{E}\left[N^{1}_{i}\right]+\mathbb{E}\left[\sum_{j=2}^{L_{\theta_{i}}^{*}(0)}N^{j}_{i}\right] \\ \nonumber
&= \mathbb{E}\left[N^{1}_{i}\right]+\sum_{j=2}^{L_{\theta_{i}}^{*}(0)} \mathbb{E}\left[N^{j}_{i}\right]\\ \nonumber
&= \frac{1}{p_{\theta_{i}}} + \frac{L_{\theta_{i}}^{*}(\alpha^{s}_{\theta_{i}})}{(1-\gamma_{\theta_{i}})\, p_{\theta_{i}} + \gamma_{\theta_{i}}},\\ \nonumber
\end{align} 
and the result of the Theorem follows.

\section{Proof of Corollary 1}
\renewcommand{\theequation}{\thesection.\arabic{equation}}
\setcounter{mytempeqncnt}{\value{equation}} \setcounter{equation}{0}
We first define the notion of {\it first-order stochastic dominance} as follows. A pdf (or pmf) $f(x)$ first-order stochastically dominates a pdf $g(x)$  if and only if $G(x) \geq F(x), \forall x,$ with strict inequality for some values of $x$, where $F(x)$ and $G(x)$ are the cumulative density functions. In this proof, we will use {\it first-order stochastic dominance} and {\it stochastic dominance} interchangeably. For the two random variables $x$ and $y$, if $f(x)$ stochastically dominates $f(y)$, then we say $y \succeq x$. In the following, we prove a Lemma that will be utilized in proving this Theorem. \\
{\bf Lemma D.1.} Let $X_{1}, X_{2}, Y_{1},$ and $Y_{2}$ be independent random variables, and let $Z_{1} = X_{1} + Y_{1}$ and $Z_{2} = X_{2} + Y_{2}$. If $X_{1} \succeq X_{2}$ and $Y_{1} \succeq Y_{2}$, then $Z_{1} \succeq Z_{2}$.   
\begin{proof}
We prove the Lemma for continuous random variables, and the result can be straightforwardly generalized to discrete random variables. Since $X_{1} \succeq X_{2}$ and $Y_{1} \succeq Y_{2}$, then we have $F_{X_{1}}(x_{1}) \leq F_{X_{2}}(x_{2})$, $F_{Y_{1}}(Y_{1}) \leq F_{Y_{2}}(Y_{2})$, $\int u(x_{1}) f(x_{1}) dx_{1} \geq \int u(x_{2}) f(x_{2}) dx_{2}$, and $\int u(y_{1}) f(y_{1}) dy_{1} \geq \int u(y_{2}) f(y_{2}) dy_{2}$, for any increasing function $u(.)$. Note that since $Z_{1} = X_{1} + Y_{1}$ and $Z_{2} = X_{2} + Y_{2}$, then we have that $F_{Z_{1}}(z_{1}) = \int F_{Y_{1}}(z_{1}-x_{1}) f(x_{1}) dx_{1}$, and $F_{Z_{2}}(z_{2}) = \int F_{Y_{2}}(z_{2}-x_{2}) f(x_{2}) dx_{2}$. Since $F_{Y_{1}}(Y_{1}) \leq F_{Y_{2}}(Y_{2})$ and $X_{1} \succeq X_{2}$, then $F_{Z_{1}}(z_{1}) \leq F_{Z_{2}}(z_{2})$ and it follows that $Z_{1} \succeq Z_{2}$. \, \IEEEQEDhere 
\end{proof} 
{\bf Lemma D.2.} If $Z_{1} = \sum_{i=1}^{N}X_{i}$ and $Z_{2} = \sum_{i=1}^{M}X_{i}$, where $N>M$, and the variables $X_{i}, \forall i \leq N$ are i.i.d non-negative random variables, then $Z_{1} \succeq Z_{2}$.    
\begin{proof}
Let $\tilde{X}_{1} = \sum_{i=1}^{M}X_{i}$, and $\tilde{X}_{2} = \sum_{i=1}^{N-M}X_{i}$. We can write $Z_{1}$ as $Z_{1} = \tilde{X}_{1} + \tilde{X}_{2}$. The cdf of $Z_{1}$ is then given by $F_{Z_{1}}(z_{1}) = \int F_{X_{1}}(z_{1}-\tilde{x}_{2}) f_{\tilde{x}_{2}}(\tilde{x}_{2}) d\tilde{x}_{2}$. Since $\int F_{X_{1}}(z_{1}-\tilde{x}_{2}) f_{\tilde{x}_{2}}(\tilde{x}_{2}) d\tilde{x}_{2} \leq F_{X_{1}}(z_{1}), \forall z_{1}$, and since $F_{X_{1}}(z_{1}) = F_{Z_{2}}(z_{1})$, then $F_{Z_{1}}(z) \leq F_{Z_{2}}(z), \forall z,$ and it follows that $Z_{1} \succeq Z_{2}$. \, \IEEEQEDhere 
\end{proof}

The pmf under study in this Theorem is $f_{T_{i}}(T_{i})$, which is the pmf of the EFT given a birth date and type of an agent, i.e. $f_{T_{i}}(T_{i}) = \sum_{G^{i-1} \in \mathcal{G}^{i-1}} f_{T_{i}}\left(T_{i}\left|G^{i-1}\right.\right) \mathbb{P}\left(G^{i-1}\right)$, which we have shown that it converges to a steady-state distribution in Appendix C. In the following, we apply a comparative statics analysis for the different exogenous parameters assuming a large enough network, and we start by the type distribution. For $\tilde{p}_{\theta_{i}} > p_{\theta_{i}}$, we compare $T_{i}\left(p_{\theta_{i}},h_{\theta_{i}},\gamma_{\theta_{i}}, L_{\theta_{i}}^{*}(0)\right)$ and $T_{i}\left(\tilde{p}_{\theta_{i}},h_{\theta_{i}},\gamma_{\theta_{i}}, L_{\theta_{i}}^{*}(0)\right)$. We first start by showing that for extremely homophilic agents, we have $T_{i}\left(p_{\theta_{i}},h_{\theta_{i}},\gamma_{\theta_{i}}, L_{\theta_{i}}^{*}(0)\right) \succeq T_{i}\left(\tilde{p}_{\theta_{i}},h_{\theta_{i}},\gamma_{\theta_{i}}, L_{\theta_{i}}^{*}(0)\right)$. Recall from Appendix C that for a large network and extremely homophilic agents, the EFT is simply given by $T_{i} = N^{1}_{i} + \bar{N}^{1}_{i}$, where the cdf of the two random variables $N^{1}_{i}$ and $\bar{N}^{1}_{i}$ are given by     
\[F(N^{1}_{i}) = 1-(1-p)^{N^{1}_{i}},\]
and
\[F(\bar{N}^{1}_{i}) = 1-I_{1-p}\left(L_{\theta_{i}}^{*}(0), \bar{N}^{1}_{i}-L_{\theta_{i}}^{*}(0)+1\right),\]   
where $p=\left((1-\gamma_{\theta_{i}}) p_{\theta_{i}} + \gamma_{\theta_{i}}\right)$, $I_{1-p}(x,y)$ is the {\it regularized incomplete beta function,} which is defined in terms of the incomplete beta function $B(1-p;x,y) = \int_{0}^{1-p} z^{x-1}(1-z)^{y-1} dz$ as $I_{1-p}(x,y) = \frac{B(1-p;x,y)}{B(x,y)}$. The first derivative of $I_{1-p}(x,y)$ with respect to $p$ is given by 
\[\frac{\partial I_{1-p}(x,y)}{\partial p} = \frac{-(1-p)^{x-1} p^{y-1}}{B(x,y)} < 0,\]
thus, $I_{1-p}(x,y)$ is monotonically decreasing in $p$.

Now let $p$ and $\tilde{p}$ be defined as $p=\left((1-\gamma_{\theta_{i}}) p_{\theta_{i}} + \gamma_{\theta_{i}}\right)$ and $\tilde{p}=\left((1-\gamma_{\theta_{i}}) \tilde{p}_{\theta_{i}} + \gamma_{\theta_{i}}\right)$. If $\tilde{p}_{\theta_{i}} > p_{\theta_{i}},$ then $\tilde{p} > p,$ and it follows that both $1-(1-p)^{N^{1}_{i}} < 1-(1-\tilde{p})^{N^{1}_{i}}$, and $1-I_{1-\tilde{p}}\left(L_{\theta_{i}}^{*}(0), \bar{N}^{1}_{i}-L_{\theta_{i}}^{*}(0)+1\right) > 1-I_{1-p}\left(L_{\theta_{i}}^{*}(0), \bar{N}^{1}_{i}-L_{\theta_{i}}^{*}(0)+1\right),$ which from Lemma D.1 implies that the cdf of $T_{i}$ for a type distribution $\tilde{p}_{\theta_{i}}$ is greater than or equal to the the cdf of $T_{i}$ for a type distribution $p_{\theta_{i}}$ for all values of $T_{i}$. Therefore, we have that $T_{i}\left(p_{\theta_{i}},h_{\theta_{i}},\gamma_{\theta_{i}}, L_{\theta_{i}}^{*}(0)\right) \succeq T_{i}\left(\tilde{p}_{\theta_{i}},h_{\theta_{i}},\gamma_{\theta_{i}}, L_{\theta_{i}}^{*}(0)\right)$. The same applies for the structural opportunism parameter $\gamma_{\theta_{i}}$. Let $p$ and $\tilde{p}$ be defined as $p=\left((1-\gamma_{\theta_{i}}) p_{\theta_{i}} + \gamma_{\theta_{i}}\right)$ and $\tilde{p}=\left((1-\tilde{\gamma}_{\theta_{i}}) p_{\theta_{i}} + \tilde{\gamma}_{\theta_{i}}\right)$. If $\tilde{\gamma}_{\theta_{i}} > \gamma_{\theta_{i}},$ then $\tilde{p} > p,$ and it follows that $T_{i}\left(p_{\theta_{i}},h_{\theta_{i}},\gamma_{\theta_{i}}, L_{\theta_{i}}^{*}(0)\right) \succeq T_{i}\left(p_{\theta_{i}},h_{\theta_{i}},\tilde{\gamma}_{\theta_{i}}, L_{\theta_{i}}^{*}(0)\right)$. Finally, since $T_{i} = N^{1}_{i} + \sum_{j=2}^{L_{\theta_{i}}^{*}(0)-1} N^{j}_{i}$, then it follows from Lemma D.2 that if $\tilde{L}_{\theta_{i}}^{*}(0) > L_{\theta_{i}}^{*}(0)$, then $T_{i}\left(p_{\theta_{i}},h_{\theta_{i}},\gamma_{\theta_{i}}, \tilde{L}_{\theta_{i}}^{*}(0)\right) \succeq T_{i}\left(p_{\theta_{i}},h_{\theta_{i}},\gamma_{\theta_{i}}, L_{\theta_{i}}^{*}(0)\right)$.

\section{Proof of Theorem 2}
\renewcommand{\theequation}{\thesection.\arabic{equation}}
\setcounter{mytempeqncnt}{\value{equation}} \setcounter{equation}{0}
We prove the Theorem through the following steps: we show that a steady-state utility function exists for every agent, and then we compute an upper-bound on the achieved utility of an agent. We prove the Theorem statement by showing that such an upper-bound is achieved if and only if agents are extremely homophilic.

First, we show that a steady-state utility function exists for every agent. Since $\gamma_{k}<1, \forall k \in \Theta$, we know from Lemma 1 that the meeting process of each agent has a finite stopping time. That is, for a large network each agent $i$ meets a same-type agent at each time step with a positive probability since $\mathbb{P}\left(\theta_{m_{i}(t)} = \theta_{i}\right) > \left(1-\gamma_{\theta_{i}}\right)p_{\theta_{i}}$. Thus each agent $i$ will eventually saturate its utility function and converge to a steady-state ego network that remains fixed for all $t > T_{i}$. 

Now we upper bound the average utility function of type-$k$ agents. Recall that $U_{k}^{t} = \frac{1}{\left|\mathcal{V}_{k}^{t}\right|}\sum_{j\in \mathcal{V}_{k}^{t}} u_{j}^{t}$. Since for each agent $i$ we have $\alpha^{s}_{\theta_{i}} \geq \alpha^{d}_{\theta_{i}}$, then each individual agent maximizes its utility when linked to similar type agents only, which corresponds to a utility of $u_{i}^{t} = v_{\theta_{i}}\left(\alpha^{s}_{\theta_{i}}L_{\theta_{i}}^{*}(0)\right) - c \, L_{\theta_{i}}^{*}(0)$. Thus, we have
\begin{align}
\overline{U}_{k}^{*} &= \frac{1}{\left|\mathcal{V}_{k}^{t}\right|}\sum_{j\in \mathcal{V}_{k}^{t}} v_{k}\left(\alpha^{s}_{k}L_{k}^{*}(0)\right) - c \, L_{k}^{*}(0) \nonumber \\
&= v_{k}\left(\alpha^{s}_{k}L_{k}^{*}(0)\right) - c \, L_{k}^{*}(0). \nonumber
\end{align}
The average utility of all agents in the network $U^{t}$ is upper bounded by $\overline{U}^{*}_{t} = \frac{1}{t}\sum_{j\in \mathcal{V}^{t}} u_{j}^{t} = \frac{1}{t} \sum_{k \in \Theta} \sum_{j\in \mathcal{V}_{k}^{t}} v_{k}\left(\alpha^{s}_{k}L_{k}^{*}(0)\right) - c \, L_{k}^{*}(0),$ which for a large network converges to $\overline{U}^{*}$, where $\overline{U}^{*} = \lim_{t \rightarrow \infty} \frac{1}{t}\sum_{j\in \mathcal{V}^{t}} u_{j}^{t}$, and
\begin{align}
\overline{U}^{*} &= \lim_{t \rightarrow \infty} \sum_{k \in \Theta} \frac{\left|\mathcal{V}_{k}^{t}\right|}{t} \left(v_{k}\left(\alpha^{s}_{k}L_{k}^{*}(0)\right) - c \, L_{k}^{*}(0)\right) \nonumber \\
&= \sum_{k \in \Theta} p_{k} \left(v_{k}\left(\alpha^{s}_{k}L_{k}^{*}(0)\right) - c \, L_{k}^{*}(0)\right). \nonumber
\end{align}

Now we prove that this upper-bound is achieved if and only if all types of agents are extremely homophilic. We start by showing that if $h_{l} = 1, \forall l \in \Theta,$ then $\lim_{t \rightarrow \infty} U^{t} = \overline{U}^{*}.$ If agents are extremely homophilic, then each agent connects only to similar type agents, i.e. $\mathbb{P}\left(a_{i}^{t}\left|\theta_{m_{i}(t)} = \theta_{i}\right.\right) = \mathbb{I}_{\left\{\theta_{m_{i}(t)} = \theta_{i}, N_{i}(t) < L_{\theta_{i}}^{*}(0)\right\}}$. Since each agent meets a same-type agent with a non-zero probability in every time step, and will always form its ego network in a finite time (recall Lemma 1), the utility achieved by each agent $i$ is then given by $v_{\theta_{i}}\left(\alpha^{s}_{\theta_{i}}L_{\theta_{i}}^{*}(0)\right) - c \, L_{\theta_{i}}^{*}(0)$, and $\lim_{t \rightarrow \infty} U^{t} = \overline{U}^{*}.$   

Now we prove the converse, and show that if $\lim_{t \rightarrow \infty} U^{t} = \overline{U}^{*},$ then $h_{l} = 1, \forall l \in \Theta$. If $h_{l} < 1$ for exactly one type of agents $l \in \Theta$, then there is a fraction of type-$l$ agents that form at least one link with a different-type agent, i.e. $\lim_{t \rightarrow \infty} \frac{1}{t}\left|\left\{j \in \mathcal{V}_{l}^{t}\left|N_{j}^{d}(t) > 0\right.\right\}\right|> (1-p_{l}),$ thus $\lim_{t \rightarrow \infty} U_{l}^{t} < \overline{U}_{l}^{*},$ and thus $\lim_{t \rightarrow \infty} U^{t} < \overline{U}^{*}.$ Therefore, all agents must be extremely homophilic for the optimal utility to be achieved.  

Finally, since when $h_{l} = 1, \forall l \in \Theta,$ agents restrict their links to same-type agents only, then there is no links between different groups and the network will be disconnected with the number of components being at least equal to the number of types, i.e. $\omega\left(G^{t}\right) \geq |\Theta|$. 

\section{Proof of Theorem 3}
\renewcommand{\theequation}{\thesection.\arabic{equation}}
\setcounter{mytempeqncnt}{\value{equation}} \setcounter{equation}{0}

We start by evaluating the popularity growth rate in a tolerant society with fully non-opportunistic agents, i.e. a society with $h_{k} = 0, \gamma_{k} = 0, \forall k \in \Theta$. Note that the expected popularity of any agent $i$ is given by 
\begin{align}  
\mathbb{E}\left\{\Delta \mbox{deg}_{i}^{-}(t)\right\} &= \mathbb{E}\left\{\sum_{j=i}^{t}\Delta \mbox{deg}_{i}^{-}(j)\right\} \nonumber \\
&= \sum_{j=i}^{t}\mathbb{E}\left\{\Delta \mbox{deg}_{i}^{-}(j)\right\},
\label{pop1}
\end{align}
where the expectation is taken over all the realizations of the graph process $\left\{G^{t}\right\}_{t=1}^{\infty}$, thus using the {\it the law of total expectation}, (\ref{pop1}) can be written as
\begin{equation}  
\mathbb{E}\left\{\Delta \mbox{deg}_{i}^{-}(j)\right\} = \sum_{G^{j} \in \mathcal{G}^{j}}\mathbb{E}\left\{\left.\Delta \mbox{deg}_{i}^{-}(j)\right|G^{j}\right\} \mathbb{P}\left(G^{j}\right).
\label{pop2}
\end{equation}
In the following, we compute the term $\sum_{G^{j} \in \mathcal{G}^{j}}\mathbb{E}\left\{\left.\Delta \mbox{deg}_{i}^{-}(j)\right|G^{j}\right\} \mathbb{P}\left(G^{j}\right)$, and then compute the summation in (\ref{pop1}) in order to obtain the popularity growth rate. First, note that from Theorem 1, we know that in a tolerant society, each agent $j$ stays $L_{\theta_{j}}^{*}(0)$ time steps in the ego network formation process almost surely. Thus, the set of agents that can potentially link to agent $i$ at any time step $t$ (which we denote as $\Phi^{t}$) is given by
\[\Phi^{t} =  \left\{t-\max_{l \in \Theta} L_{l}^{*}(0)+1, t-\max_{l \in \Theta} L_{l}^{*}(0)+2,.\,.\,.,t\right\}.\]       
That is, an agent's popularity acquisition process in a tolerant and non-opportunisic society depends only on the types of the $\max_{l \in \Theta} L_{l}^{*}(0)$ most recently born agents, and their actions in the most recent $\max_{l \in \Theta} L_{l}^{*}(0)$ time steps. The types of such agents determine their levels of gregariousness, and thus the possibility of each of them linking to agent $i$ at time $t$. Since agents find each others just by random matching, we can then write (\ref{pop2}) as shown in (\ref{pop3}). Based on (\ref{pop3}), the expected number of links gained by agent $i$ at time $t$ is given by (\ref{pop5}).   
\begin{figure*}[!t]
\setcounter{mytempeqncnt}{\value{equation}} \setcounter{equation}{2}
\begin{equation}  
\mathbb{E}\left\{\Delta \mbox{deg}_{i}^{-}(t)\right\} = \sum_{G^{t} \in \mathcal{G}^{t}}\mathbb{E}\left\{\left.\Delta \mbox{deg}_{i}^{-}(t)\right|\left\{\theta_{v}\right\}_{v \in \Phi^{t}}, \left\{A^{t-1}(v,i)\right\}_{v \in \Phi^{t}/\{t\}}\right\} \mathbb{P}\left(\left\{\theta_{v}\right\}_{v \in \Phi^{t}}, \left\{A^{t-1}(v,i)\right\}_{v \in \Phi^{t}/\{t\}}\right).
\label{pop3}
\end{equation} 
\setcounter{equation}{\value{mytempeqncnt}+1} \hrulefill{}\vspace*{4pt}
\end{figure*}
\begin{figure*}[!t]
\setcounter{mytempeqncnt}{\value{equation}} \setcounter{equation}{3}
\begin{align}
\mathbb{E}\left\{\Delta \mbox{deg}_{i}^{-}(t)\right\} &\overset{\text{(a)}}= \sum_{G^{t} \in \mathcal{G}^{t}}\mathbb{E}\left\{\left.\Delta \mbox{deg}_{i}^{-}(t)\right|\left\{\theta_{v}\right\}_{v \in \Phi^{t}}, \left\{A^{t-1}(v,i)\right\}_{v \in \Phi^{t}/\{t\}}\right\} \mathbb{P}\left(\left\{\theta_{v}\right\}_{v \in \Phi^{t}}, \left\{A^{t-1}(v,i)\right\}_{v \in \Phi^{t}/\{t\}}\right) \nonumber \\
&\overset{\text{(b)}}= \sum_{k \in \Phi^{t}} \sum_{\theta_{k} \in \Theta} p_{\theta_{k}} \left(\mathbb{P}\left(m_{k}(t) = i \left|k \notin \mathcal{N}_{i,t-1}^{-}, L_{\theta_{k}}^{*}(0)\geq t - k + 1\right.\right) \mathbb{P}\left(k \notin \mathcal{N}_{i,t-1}^{-}, L_{\theta_{k}}^{*}(0)\geq t - k + 1\right)\right)  \nonumber \\
&\overset{\text{(c)}}= \sum_{k \in \Phi^{t}} \sum_{\theta_{k} \in \Theta} p_{\theta_{k}} \left(\mathbb{I}_{\left\{L_{\theta_{k}}^{*}(0)\geq t - k + 1\right\}}\mathbb{P}\left(m_{k}(t) = i \left|k \notin \mathcal{N}_{i,t-1}^{-}\right.\right) \mathbb{P}\left(k \notin \mathcal{N}_{i,t-1}^{-}\right) \right)  \nonumber \\ 
&\overset{\text{(d)}}= \sum_{k \in \Phi^{t}} \sum_{\theta_{k} \in \Theta} p_{\theta_{k}} \left(\mathbb{I}_{\left\{L_{\theta_{k}}^{*}(0)\geq t - k + 1\right\}} \left(\frac{1}{t-1-(t-k)}\right) \prod_{n=k}^{t-1}\left(1-\frac{1}{n-1-(n-k)}\right)\right) \nonumber \\
&\overset{\text{(e)}}= \sum_{k \in \Phi^{t}} \left(\sum_{\theta_{k} \in \Theta} p_{\theta_{k}} \mathbb{I}_{\left\{L_{\theta_{k}}^{*}(0)\geq t - k + 1\right\}}\right) \left(\frac{1}{k-1}\right) \left(1-\frac{1}{k-1}\right)^{t-k} \nonumber \\
&\overset{\text{(f)}}\approx \sum_{k = t-\max_{l \in \Theta} L_{l}^{*}(0)+1}^{t} \left(\sum_{\theta_{k} \in \Theta} p_{\theta_{k}} \mathbb{I}_{\left\{L_{\theta_{k}}^{*}(0)\geq t - k + 1\right\}}\right) \left(\frac{e^{-\frac{t-k}{k-1}}}{k-1}\right)  \nonumber \\
&\overset{\text{(g)}}= \frac{\left(\sum_{\theta_{k} \in \Theta} p_{\theta_{k}} \mathbb{I}_{\left\{L_{\theta_{k}}^{*}(0)\geq \max_{l \in \Theta} L_{l}^{*}(0)\right\}}\right) e^{-\frac{\max_{l \in \Theta} L_{l}^{*}(0)-1}{t-\max_{l \in \Theta} L_{l}^{*}(0)}}}{t-\max_{l \in \Theta} L_{l}^{*}(0)}+.\,.\,.+ \frac{\left(\sum_{\theta_{k} \in \Theta} p_{\theta_{k}} \mathbb{I}_{\left\{L_{\theta_{k}}^{*}(0)\geq 2\right\}}\right)e^{-\frac{2}{t-3}}} {t-2} + \frac{1}{t-1} \nonumber \\
&\overset{\text{(h)}}= \sum_{w = 1}^{\max_{l \in \Theta} L_{l}^{*}(0)} \frac{\left(\sum_{\theta_{k} \in \Theta} p_{\theta_{k}} \mathbb{I}_{\left\{L_{\theta_{k}}^{*}(0)\geq w\right\}}\right) e^{-\frac{w-1}{t-w}}}{t-w} \nonumber \\
&\overset{\text{(i)}}\asymp \frac{\bar{L}}{t}. 
\label{pop5}
\end{align}
\setcounter{equation}{\value{mytempeqncnt}+1} \hrulefill{}\vspace*{4pt}
\end{figure*}

In the following, we briefly explain the steps involved in (\ref{pop5}). First, (a) is the application of the law of total expectation to the expected in-degree of an agent $i$ by averaging over all the possible types of the most recently born agents, and the probability that such agents have linked to $i$ in the previous time steps. In (b) and (c), we rewrite the expression in (a) by computing the average number of recently born agents who meet agent $i$ and have not already been linked to $i$ untill the time step $t-1$. The expressions for the probability that one of the recently born agents meet with agent $i$ are plugged in (d) and further simplified in (e). A Taylor series approximation $(1-\frac{1}{x})^{a} \approx e^{-\frac{a}{x}}$ is used in (f)-(h), and the asymptotic value of $\mathbb{E}\left\{\Delta \mbox{deg}_{i}^{-}(t)\right\}$ is provided in (i).        

From (\ref{pop5}), the expected number of links gained by any agent $i$ at time $t$ in a large network boils down to the simple expression $\mathbb{E}\left\{\Delta \mbox{deg}_{i}^{-}(t)\right\} = \frac{\bar{L}}{t},$ where $\bar{L} = \sum_{m \in \Theta} p_{m} L_{m}^{*}(0)$. Consequently, the expected popularity of agent $i$ at time $t$ is given by  
\begin{align} 
\mathbb{E}\left\{\mbox{deg}_{i}^{-}(t)\right\} &= \mathbb{E}\left\{\sum_{j=i}^{t}\Delta\mbox{deg}_{i}^{-}(j)\right\} \nonumber \\
&= \sum_{j=i}^{t}\mathbb{E}\left\{\Delta\mbox{deg}_{i}^{-}(j)\right\} \nonumber \\
&= \sum_{j=i}^{t} \frac{\bar{L}}{j} \nonumber \\
&= \bar{L}\left(H_{t}-H_{i-1}\right) \nonumber \\
&\asymp \bar{L}\left((\log(t)-\psi)-(\log(i-1)-\psi)\right) \nonumber \\
&= \bar{L} \log\left(\frac{t}{i-1}\right),
\label{pop9}
\end{align}
where $H_{N}$ is the $N^{th}$ harmonic number, and $\psi$ is the {\it Euler-mascheroni} constant. Thus, $\mathbb{E}\left\{\mbox{deg}_{i}^{-}(t)\right\}$ is $O\left(\bar{L}\log(t)\right)$, and the first part of the Theorem follows. 

Next, we evaluate the popularity growth rate in a tolerant society with fully opportunistic agents, i.e. a society with $h_{k} = 0, \gamma_{k} = 1, \forall k \in \Theta$. Similar to (\ref{pop1}), we start by evaluating $\sum_{G^{t} \in \mathcal{G}^{t}}\mathbb{E}\left\{\left.\Delta \mbox{deg}_{i}^{-}(t)\right|G^{t}\right\} \mathbb{P}\left(G^{t}\right)$. Note that the maximum number of agents forming links at any time step is given by $\max_{l \in \Theta} L_{l}^{*}(0)$ (the birth dates of all such agents belong to the set $\Phi^{t}$). Since in an opportunistic society the meeting process depends on the network structure, we start by evaluating the term $\mathbb{E}\left\{\left.\Delta \mbox{deg}_{i}^{-}(t)\right|G^{t}\right\}$ in (\ref{pop11}).    
\begin{figure*}[!t]
\setcounter{mytempeqncnt}{\value{equation}} \setcounter{equation}{5}
\[\mathbb{E}\left\{\left.\Delta \mbox{deg}_{i}^{-}(t)\right|\mbox{deg}_{i}^{-}(t-1), \mbox{deg}_{i}^{-}(t-2), .\,.\,., \mbox{deg}_{i}^{-}(i)\right\}\]
\begin{align}
&\overset{\text{(a)}}= \sum_{k \in \Phi^{t}} \sum_{\theta_{k} \in \Theta} p_{\theta_{k}} \left(\mathbb{P}\left(m_{k}(t) = i \left|k \notin \mathcal{N}_{i,t-1}^{-}, k \in \bigcup_{j \in \mathcal{N}_{i,t-1}^{-}} \mathcal{N}_{j,t-1}^{-}, K_{k}(t), L_{\theta_{k}}^{*}(0)\geq t - k + 1\right.\right)\right.   \nonumber \\
&\times \left. \mathbb{P}\left(k \notin \mathcal{N}_{i,t-1}^{-}, k \in \bigcup_{j \in \mathcal{N}_{i,t-1}^{-}} \mathcal{N}_{j,t-1}^{-}, K_{k}(t), L_{\theta_{k}}^{*}(0)\geq t - k + 1\right)\right)  \nonumber \\
&\overset{\text{(b)}}= \sum_{k \in \Phi^{t}} \sum_{\theta_{k} \in \Theta} p_{\theta_{k}} \left(\mathbb{I}_{\left\{L_{\theta_{k}}^{*}(0)\geq t - k + 1\right\}}\mathbb{P}\left(m_{k}(t) = i \left|k \notin \mathcal{N}_{i,t-1}^{-}, k \in \bigcup_{j \in \mathcal{N}_{i,t-1}^{-}} \mathcal{N}_{j,t-1}^{-}, K_{k}(t)\right.\right)\right. \nonumber \\
&\times \left. \mathbb{P}\left(k \notin \mathcal{N}_{i,t-1}^{-}, k \in \bigcup_{j \in \mathcal{N}_{i,t-1}^{-}} \mathcal{N}_{j,t-1}^{-},K_{k}(t)\right)\right) \nonumber \\
&\overset{\text{(c)}}= \sum_{k \in \Phi^{t}} \sum_{\theta_{k} \in \Theta} p_{\theta_{k}} \left(\, \frac{\mathbb{I}_{\left\{L_{\theta_{k}}^{*}(0)\geq t - k + 1\right\}}}{K_{k}(t)}\, \mathbb{P}\left(k \notin \mathcal{N}_{i,t-1}^{-}, k \in \bigcup_{j \in \mathcal{N}_{i,t-1}^{-}} \mathcal{N}_{j,t-1}^{-},K_{k}(t)\right)\right) \nonumber \\
&\overset{\text{(d)}}= \frac{1}{t-1} + \left(\frac{\mbox{deg}_{i}^{-}(t-1)}{t-2}\right) \left(\sum_{n \in \Theta} p_{n} \mathbb{I}_{\left\{L_{n}^{*}(0)\geq 2\right\}}\right) \sum_{m \in \Theta}\frac{p_{m}}{L_{m}^{*}(0)} + \sum_{k =t-\max_{l}L_{l}^{*}(0)+1}^{t-2} \sum_{\theta_{k} \in \Theta} p_{\theta_{k}} \left(\, \frac{\mathbb{I}_{\left\{L_{\theta_{k}}^{*}(0)\geq t - k + 1\right\}}}{K_{k}(t)}\times\right. \nonumber \\
&\, \left.\mathbb{P}\left(k \notin \mathcal{N}_{i,t-1}^{-}, k \in \bigcup_{j \in \mathcal{N}_{i,t-1}^{-}} \mathcal{N}_{j,t-1}^{-},K_{k}(t)\right)\right) \nonumber \\
&\overset{\text{(e)}}\geq \frac{1}{t-1}+ \left(\frac{\mbox{deg}_{i}^{-}(t-1)}{t-2}\right) \sum_{m \in \Theta}\frac{p_{m}}{L_{m}^{*}(0)} + \sum_{v=2}^{\max_{l} L_{l}^{*}(0)-1} \left(\frac{\mbox{deg}_{i}^{-}(t-v)}{t-v-1}\right) \left(\sum_{n \in \Theta} p_{n} \mathbb{I}_{\left\{L_{n}^{*}(0)\geq v+1\right\}}\right) \times  \nonumber \\
&\left(\sum_{z_{1} \in \Theta}\sum_{z_{2} \in \Theta}\cdots \sum_{z_{v} \in \Theta} \prod_{u=1}^{v}p_{z_{u}}\left(1-\frac{1}{\sum_{x=1}^{v-1} L_{z_{x}}^{*}(0)}\right) \frac{1}{\sum_{x=1}^{v} L_{z_{x}}^{*}(0)}\right)  \nonumber \\
&\overset{\text{(f)}}\geq \frac{1}{t-1}+\frac{\mbox{deg}_{i}^{-}(t-1)}{t-2} \sum_{w \in \Theta}\frac{p_{w}}{L_{w}^{*}(0)}  \nonumber \\
&\overset{\text{(g)}} \approx \frac{1+\left(\sum_{w \in \Theta}\frac{p_{w}}{L_{w}^{*}(0)}\right) \mbox{deg}_{i}^{-}(t-1)}{t}.  
\label{pop11}
\end{align}
\setcounter{equation}{\value{mytempeqncnt}+1} \hrulefill{}\vspace*{4pt}
\end{figure*}
\begin{figure*}[!t]
\setcounter{mytempeqncnt}{\value{equation}} \setcounter{equation}{6}
\begin{align}
\mathbb{E}\left\{\Delta \mbox{deg}_{i}^{-}(t)\right\} &\overset{\text{(a)}}= \sum_{G^{t} \in \mathcal{G}^{t}}\mathbb{E}\left\{\left.\Delta \mbox{deg}_{i}^{-}(t)\right|G^{t}\right\} \mathbb{P}\left(G^{t}\right)  \nonumber \\
&\overset{\text{(b)}}= \sum_{G^{t} \in \mathcal{G}^{t}}\frac{1+\left(\sum_{w \in \Theta}\frac{p_{w}}{L_{w}^{*}(0)}\right) \mbox{deg}_{i}^{-}(t-1)}{t} \mathbb{P}\left(G^{t}\right)  \nonumber \\
&\overset{\text{(c)}}= \frac{1+\left(\sum_{w \in \Theta}\frac{p_{w}}{L_{w}^{*}(0)}\right) \sum_{G^{t} \in \mathcal{G}^{t}}\mbox{deg}_{i}^{-}(t-1) \mathbb{P}\left(G^{t}\right)}{t} \nonumber \\
&\overset{\text{(d)}}= \frac{1+\left(\sum_{w \in \Theta}\frac{p_{w}}{L_{w}^{*}(0)}\right) \mathbb{E}\left\{\mbox{deg}_{i}^{-}(t-1)\right\}}{t}. 
\label{pop109}
\end{align}
\setcounter{equation}{\value{mytempeqncnt}+1} \hrulefill{}\vspace*{4pt}
\end{figure*}
We briefly explain the steps involved in (\ref{pop11}) in the following. In (a)-(b) we rewrite the law of total expectation by taking the expected number of agents that will meet agent $i$ given that these agents have not yet saturated their utilities, and have not linked with $i$ before, and have become a follower of follower for $i$ in a previous time step. In (c) we further simplify the expression by observing that the probability that agent $k$ meets agent $i$  given that $k$ is a follower of a follower of $i$ is simply given by $\frac{1}{K_{k}(t)}$. In (d) we get the exact probability that agent $t$ links to agent $i$, which is given by $\frac{1}{t-1}$, and the probability that agent $t-1$ links to agent $i$ at time $t$ is given by $\left(\frac{\mbox{deg}_{i}^{-}(t-1)}{t-2}\right)\sum_{m \in \Theta}\frac{p_{m}}{L_{m}^{*}(0)}$, i.e. $k$ links to a follower of $i$ at time $t-1$, and then links to $i$ at time $t$, which happens with a probability of $\frac{p_{m}}{L_{m}^{*}(0)}$ (i.e. the reciprocal of the gregariousness of agent $i$'s follower averaged over its random type distribution). The probability that the ``older" agents link to $i$ at time $t$ depend on the degrees of the followers of $i$, which makes the problem intractable. Thus, we lower bound the $\mathbb{E}\left\{\Delta \mbox{deg}_{i}^{-}(t)\right\}$ by assuming that agents find agent $i$ after more than two time steps only if they have linked to one of its followers at their birth date. This leads to the following approximation for the expected gained links by agent $i$ at time $t$
\[\mathbb{E}\left\{\Delta \mbox{deg}_{i}^{-}(t)\right\} =  \frac{1+b \, \mathbb{E}\left\{\mbox{deg}_{i}^{-}(t-1)\right\}}{t},\]    
where
\[b = \sum_{w \in \Theta}\frac{p_{w}}{L_{w}^{*}(0)}.\]
Now assume a continuous-time approximation for the popularity growth process. Bote that since 
\[\mathbb{E}\left\{\Delta \mbox{deg}_{i}^{-}(t)\right\} = \mathbb{E}\left\{\mbox{deg}_{i}^{-}(t)\right\}-\mathbb{E}\left\{\mbox{deg}_{i}^{-}(t-1)\right\},\] 
then we have that
\[\frac{\partial \mathbb{E}\left\{\mbox{deg}_{i}^{-}(t)\right\}}{\partial t} \approx \mathbb{E}\left\{\Delta \mbox{deg}_{i}^{-}(t)\right\}.\] 
Thus, the popularity of each agent $i$ is governed by the following differential equation 
\[\frac{\partial \mathbb{E}\left\{\mbox{deg}_{i}^{-}(t)\right\}}{\partial t} = \frac{1}{t}\left(1+ b \, \mathbb{E}\left\{\mbox{deg}_{i}^{-}(t)\right\}\right).\]
This differential equation can be solved by dividing both sides by $\left(1+ b \, \mathbb{E}\left\{\mbox{deg}_{i}^{-}(t)\right\}\right)$ and integrating both sides as follows 
\[\int \frac{1}{\left(1+ b \, \mathbb{E}\left\{\mbox{deg}_{i}^{-}(t)\right\}\right)} d\mathbb{E}\left\{\mbox{deg}_{i}^{-}(t)\right\} = \int \frac{1}{t} dt,\]
which reduces to 
\[\frac{1}{b}\log\left(1+b\,\mathbb{E}\left\{\mbox{deg}_{i}^{-}(t)\right\}\right) + c_{1} = \log(t)+c_{2},\]
and hence we have that
\[\mathbb{E}\left\{\mbox{deg}_{i}^{-}(t)\right\} = c_{3} t^{b}-\frac{1}{b},\]
where $c_{1},$ $c_{2},$ and $c_{3}$ are constants. The constant $c_{3}$ can be obtained from the initial conditions as follows. Note that at $t = i$ cannot receive any links since all agents are opportunistic and the set of followers of $i$ is empty at its birth date. Thus, $\mbox{deg}_{i}^{-}(i) = 0$, and consequently $\mathbb{E}\left\{\mbox{deg}_{i}^{-}(t)\right\} = 0,$ which means that $c_{3} = \frac{1}{i^{b}b}$. Since the differential equation above was solved by plugging in a lower bound on agent $i$'s expected popularity at any time step, it follows that
\[\mathbb{E}\left\{\mbox{deg}_{i}^{-}(t)\right\} \geq \frac{1}{b} \left(\frac{t}{i}\right)^{b}-\frac{1}{b}, \forall t \geq i,\]
which means that the popularity of every agent $i$ grows at least sublinearly in time, and the second part of the Theorem follows. 

\section{Proof of Corollary 3}
\renewcommand{\theequation}{\thesection.\arabic{equation}}
\setcounter{mytempeqncnt}{\value{equation}} \setcounter{equation}{0}
Following Definition 1, we say that preferential attachment emerges if $\mbox{deg}_{i}^{-}(t)\geq\mbox{deg}_{j}^{-}(t)$ implies that $\mbox{deg}_{i}^{-}(t)\succeq\mbox{deg}_{j}^{-}(t),\forall i, j \leq t, t \in \mathbb{N}$. We know that in a tolerant society, the maximum number of agents forming links at any time step is $\max_{l}L_{l}^{*}(0)$, and the birth dates of such agents are given by the set $\Phi^{t} = \{t,t-1,.\,.\,.,t-\max_{l}L_{l}^{*}(0)+1\}$. Following the analysis in Appendix F, the probability that agent $k \in \Phi^{t}$ links to agent $i$ at time $t$, which we denote as $p_{ik}^{t}$, is given by (\ref{pop15}).       
\begin{figure*}[!t]
\setcounter{mytempeqncnt}{\value{equation}} \setcounter{equation}{0}
\begin{align}
p_{ik}^{t} = \sum_{\theta_{k} \in \Theta} p_{\theta_{k}} \left(\mathbb{P}\left(m_{k}(t) = i \left|k \notin \mathcal{N}_{i,t-1}^{-}, k \in \bigcup_{j \in \mathcal{N}_{i,t-1}^{-}} \mathcal{N}_{j,t-1}^{-}, K_{k}(t), L_{\theta_{k}}^{*}(0)\geq t - k + 1\right.\right)\right)
\label{pop15}
\end{align}
\setcounter{equation}{\value{mytempeqncnt}+1} \hrulefill{}\vspace*{4pt}
\end{figure*}			
It follows that $\Delta \mbox{deg}_{i}^{-}(t)$ obeys a {\it poisson binomial distribution} with a support $n \in \{0,1,.\,.\,.,\max_{l}L_{l}^{*}(0)\}$. The pmf of $\Delta \mbox{deg}_{i}^{-}(t)$ is given by
\[\mathbb{P}\left(\Delta \mbox{deg}_{i}^{-}(t) = n\right) = \sum_{A \in S^{t}_{n}} \prod_{q \in A} p_{iq}^{t} \prod_{r \in A^{c}} (1-p_{ir}^{t}),\]  
where $S^{t}_{n}$ is the set of all size-$n$ subsets of $\Phi^{t}$. The CDF of $\Delta \mbox{deg}_{i}^{-}(t)$ is given by
\[\mathbb{P}\left(\Delta \mbox{deg}_{i}^{-}(t) \leq n\right) = \sum_{l=0}^{n} \sum_{A \in S^{t}_{l}} \prod_{q \in A} p_{iq}^{t} \prod_{r \in A^{c}} (1-p_{ir}^{t}).\]    
It can be easily shown that $\frac{\partial \mathbb{P}\left(\Delta \mbox{deg}_{i}^{-}(t) \leq n\right)}{\partial p_{iy}^{t}} < 0, \forall y \in \Phi^{t}$. We know from (\ref{pop11}) that if $\mbox{deg}_{i}^{-}(t) \geq \mbox{deg}_{j}^{-}(t),$ then $\mathbb{P}\left(\Delta \mbox{deg}_{i}^{-}(t) \leq n\right) \leq \mathbb{P}\left(\Delta \mbox{deg}_{j}^{-}(t) \leq n\right), \forall n$. Thus, $\mbox{deg}_{i}^{-}(t) \geq \mbox{deg}_{j}^{-}(t)$ implies that $\Delta \mbox{deg}_{i}^{-}(t) \succeq \Delta \mbox{deg}_{j}^{-}(t)$.      														
									
\section{Proof of Corollary 4}
\renewcommand{\theequation}{\thesection.\arabic{equation}}
\setcounter{mytempeqncnt}{\value{equation}} \setcounter{equation}{0}
Note that for the two agents $i$ and $j$ with $i < j$, we can write $\mbox{deg}_{i}^{-}(t)$ as
\[\mbox{deg}_{i}^{-}(t) = \sum_{m=i}^{j-1}\Delta \mbox{deg}_{i}^{-}(m)+\sum_{v=j}^{t}\Delta \mbox{deg}_{i}^{-}(v),\]  
whereas $\mbox{deg}_{j}^{-}(t)$ can be written as 
\[\mbox{deg}_{i}^{-}(t) = \sum_{l=j}^{t}\Delta \mbox{deg}_{j}^{-}(l).\]
Based on the derivations in Appendix F, it is easy to see that both $\sum_{l=j}^{t}\Delta \mbox{deg}_{j}^{-}(l)$ and $\sum_{v=j}^{t}\Delta \mbox{deg}_{i}^{-}(v)$ follow the same distribution, i.e. $\mathbb{P}\left(\sum_{l=j}^{t}\Delta \mbox{deg}_{j}^{-}(l) = n\right) = \mathbb{P}\left(\sum_{v=j}^{t}\Delta \mbox{deg}_{i}^{-}(v) = n\right)$. Now we show that $\mbox{deg}_{i}^{-}(t) \succeq \mbox{deg}_{j}^{-}(t)$. In order to prove that $\mbox{deg}_{i}^{-}(t) \succeq \mbox{deg}_{j}^{-}(t)$, we need to show that $\mathbb{P}\left(\sum_{l=j}^{t}\Delta \mbox{deg}_{j}^{-}(l) \leq n\right) \geq \mathbb{P}\left(\sum_{v=i}^{t}\Delta \mbox{deg}_{i}^{-}(v) \leq n\right), \forall n$. This can be shown by observing that
\[\mathbb{P}\left(\sum_{v=i}^{t}\Delta \mbox{deg}_{i}^{-}(v) \leq n\right) = \]
\[\mathbb{P}\left(\sum_{m=i}^{j-1}\Delta \mbox{deg}_{i}^{-}(m)+\sum_{v=j}^{t}\Delta \mbox{deg}_{i}^{-}(v) \leq n\right),\]
which can be re-arranged as
\[\mathbb{P}\left(\sum_{v=i}^{t}\Delta \mbox{deg}_{i}^{-}(v) \leq n\right) = \]
\[\mathbb{P}\left(\sum_{v=j}^{t}\Delta \mbox{deg}_{i}^{-}(v) \leq n-\sum_{m=i}^{j-1}\Delta \mbox{deg}_{i}^{-}(m)\right).\]    
Since $\sum_{m=i}^{j-1}\Delta \mbox{deg}_{i}^{-}(m) > 0$ almost surely, and since $\sum_{l=j}^{t}\Delta \mbox{deg}_{j}^{-}(l)$ and $\sum_{v=j}^{t}\Delta \mbox{deg}_{i}^{-}(v)$ follow the same distribution, then it follows from the monotonicity of CDFs that
\[\mathbb{P}\left(\sum_{v=j}^{t}\Delta \mbox{deg}_{i}^{-}(v) \leq n-\sum_{m=i}^{j-1}\Delta \mbox{deg}_{i}^{-}(m)\right) \leq\]
\[\mathbb{P}\left(\sum_{l=j}^{t}\Delta \mbox{deg}_{j}^{-}(l) \leq n\right), \forall n,\]       
thus $\mbox{deg}_{i}^{-}(t) \succeq \mbox{deg}_{j}^{-}(t)$. 			

\section{Proof of Corollary 5}
\renewcommand{\theequation}{\thesection.\arabic{equation}}
\setcounter{mytempeqncnt}{\value{equation}} \setcounter{equation}{0}
From (\ref{pop9}), we know that 
\[\mathbb{E}\left\{d_{i}^{1}(t)\right\} = \bar{L} \log\left(\frac{t}{i-1}\right),\]
whereas $\mathbb{E}\left\{d_{i}^{2}(t)\right\}$ is lower-bounded as follows
\[\mathbb{E}\left\{d_{i}^{2}(t)\right\} \geq \frac{1}{b} \left(\left(\frac{t}{i}\right)^{b}-1\right).\] 
Since $\mathbb{E}\left\{d_{i}^{2}(t)\right\}$ grows faster than $\mathbb{E}\left\{d_{i}^{1}(t)\right\},$ we know that $\mathbb{E}\left\{d_{i}^{2}(t)\right\}$ dominates $\mathbb{E}\left\{d_{i}^{1}(t)\right\}$ after a finite time $T^{*}$. We can obtain $T^{*}$ by solving the following transcendental equation for $t$
\[\mathbb{E}\left\{d_{i}^{1}(t)\right\} = \mathbb{E}\left\{d_{i}^{2}(t)\right\}.\]   
An upper-bound on the solution can be obtained by solving the following transcendental equation for $t$, in which we replace $\mathbb{E}\left\{d_{i}^{2}(t)\right\}$ by its lower-bound
\begin{equation}
\bar{L} \log\left(\frac{t}{i-1}\right) = \frac{1}{b} \left(\left(\frac{t}{i}\right)^{b}-1\right). 
\label{lamart1}
\end{equation} 
Note that (\ref{lamart1}) can be put in the following form
\begin{equation}
t^{b \bar{L}} = \frac{(i-1)^{b \bar{L}}}{e} \, e^{\left(\frac{t}{i}\right)^{b}}.
\label{lamart2}
\end{equation}			
A functional form for the solution to (\ref{lamart2}) can be obtained in terms of the Lambert W function $\mathcal{W}_{-1}(.)$ \cite{ref372} as follows  			
\begin{equation}
t^{*} = i \times \left(-\bar{L} \mathcal{W}_{-1}\left(\frac{-1}{\bar{L}} e^{\frac{-1}{\bar{L}}}\right)\right)^{\frac{1}{b}}.
\label{lamart3}
\end{equation}
Thus, $T^{*} \leq t^{*}$ and the Theorem is concluded. 
			
\section{Proof of Theorem 4}
\renewcommand{\theequation}{\thesection.\arabic{equation}}
\setcounter{mytempeqncnt}{\value{equation}} \setcounter{equation}{0}

Similar to the proof of Theorem 3, we start by evaluating the popularity growth rate in an intolerant society with fully non-opportunistic agents, i.e. a society with $h_{k} = 1, \gamma_{k} = 0, \forall k \in \Theta$. The expected popularity of any agent $i$ can be written as
\begin{equation}  
\mathbb{E}\left\{\mbox{deg}_{i}^{-}(t)\right\} = \sum_{j=i}^{t}\mathbb{E}\left\{\Delta \mbox{deg}_{i}^{-}(j)\right\},
\label{pop71}
\end{equation}
where the expectation is taken over all the realizations of the graph process $\left\{G^{t}\right\}_{t=1}^{\infty}$, thus using the {\it the law of total expectation}, $\mathbb{E}\left\{\Delta \mbox{deg}_{i}^{-}(t)\right\}$ in (\ref{pop71}) can be written as
\begin{equation}  
\mathbb{E}\left\{\Delta \mbox{deg}_{i}^{-}(t)\right\} = \sum_{G^{t} \in \mathcal{G}^{t}}\mathbb{E}\left\{\left.\Delta \mbox{deg}_{i}^{-}(t)\right|G^{t}\right\} \mathbb{P}\left(G^{t}\right).
\label{pop72}
\end{equation}
Unlike the case of tolerant societies, the number of agents forming a link at any time step can be arbitrarily large, i.e. since agents are homophilic, they may wait for an arbitrarily large period to form any link since they are constrained to linking to same-type agents only. Thus, (\ref{pop72}) can be written as shown in (\ref{pop73}). The derivation of $\mathbb{E}\left\{\Delta \mbox{deg}_{i}^{-}(t)\right\}$ is given in (\ref{pop75}). In the following, we briefly explain the steps involved in (\ref{pop75}). (a) is an application of the law of total expectation, and (b)-(c) are simplifications of (a) that is obtained by observing that an agent $k$ can link to $i$ only if it has not yet linked to it in a previous time step and its utility function is not yet satisfied. (d) is obtained by observing that $k$ can link to $i$ only if $\theta_{i} = \theta_{k}$. In (e) and (f) we compute the conditional probabilities explicitly, and (g) is obtained by using a first-order Taylor series approximation, which converges to (h) for an asymptotically large network. In the rest of the steps, we compute an approximation for the summation in (h). Note that the result in (l) can be obtained by a simple Mean-field approximation: the expected EFT of a type-$m$ agent is $\frac{L_{m}^{*}(0)}{p_{m}}$, and thus the expected number of type-$m$ agents is $L_{m}^{*}(0)$, thus the expected number of links gained by a type-$m$ agent at time $t$ is $\frac{L_{m}^{*}(0)}{t}$.             
\begin{figure*}[!t]
\setcounter{mytempeqncnt}{\value{equation}} \setcounter{equation}{2}
\begin{equation}  
\mathbb{E}\left\{\Delta \mbox{deg}_{i}^{-}(t)\right\} = \sum_{G^{t} \in \mathcal{G}^{t}}\mathbb{E}\left\{\left.\Delta \mbox{deg}_{i}^{-}(t)\right|\left\{\theta_{v}\right\}_{v = 1}^{t}, \left\{A^{t-1}(v,i)\right\}_{v = 1}^{t}, \left\{\mbox{deg}_{v}^{+}(t)\right\}_{v = 1}^{t}\right\} \mathbb{P}\left(\left\{\theta_{v}\right\}_{v = 1}^{t}, \left\{A^{t-1}(v,i)\right\}_{v = 1}^{t}, \left\{\mbox{deg}_{v}^{+}(t)\right\}_{v = 1}^{t}\right).
\label{pop73}
\end{equation} 
\setcounter{equation}{\value{mytempeqncnt}+1} \hrulefill{}\vspace*{4pt}
\end{figure*}
\begin{figure*}[!t]
\setcounter{mytempeqncnt}{\value{equation}} \setcounter{equation}{3}
\begin{align}
\mathbb{E}\left\{\Delta \mbox{deg}_{i}^{-}(t)\right\} &\overset{\text{(a)}}= \sum_{G^{t} \in \mathcal{G}^{t}}\mathbb{E}\left\{\left.\Delta \mbox{deg}_{i}^{-}(t)\right|\left\{\theta_{v}\right\}_{v = 1}^{t}, \left\{A^{t-1}(v,i)\right\}_{v = 1}^{t}, \left\{\mbox{deg}_{v}^{+}(t)\right\}_{v = 1}^{t}\right\} \mathbb{P}\left(\left\{\theta_{v}\right\}_{v = 1}^{t}, \left\{A^{t-1}(v,i)\right\}_{v = 1}^{t}, \left\{\mbox{deg}_{v}^{+}(t)\right\}_{v = 1}^{t}\right) \nonumber \\
&\overset{\text{(b)}}= \sum_{k = 1}^{t} \sum_{\theta_{k} \in \Theta} p_{\theta_{k}} \left(\mathbb{P}\left(A^{t}(k,i) = 1 \left|k \notin \mathcal{N}_{i,t-1}^{-}, \mbox{deg}_{k}^{+}(t) < L_{\theta_{k}}^{*}(0)\right.\right) \mathbb{P}\left(k \notin \mathcal{N}_{i,t-1}^{-}, \mbox{deg}_{k}^{+}(t) < L_{\theta_{k}}^{*}(0)\right)\right)  \nonumber \\
&\overset{\text{(c)}}= \sum_{k = 1}^{t} \left(\mathbb{P}\left(m_{k}(t) = i \left|k \notin \mathcal{N}_{i,t-1}^{-}, \mbox{deg}_{k}^{+}(t) < L_{\theta_{k}}^{*}(0), \theta_{k} = \theta_{i}\right.\right) \mathbb{P}\left(k \notin \mathcal{N}_{i,t-1}^{-}, \mbox{deg}_{k}^{+}(t) < L_{\theta_{k}}^{*}(0), \theta_{k} = \theta_{i}\right)\right)  \nonumber \\
&\overset{\text{(d)}}= \frac{p_{\theta_{i}}}{t-1} \, \sum_{k = 1}^{t} \, \mathbb{P}\left(k \notin \mathcal{N}_{i,t-1}^{-}, \mbox{deg}_{k}^{+}(t) < L_{\theta_{k}}^{*}(0)\right)  \nonumber \\
&\overset{\text{(e)}}= \frac{p_{\theta_{i}}}{t-1} \, \left(\sum_{k = 1}^{i-1} \, \mathbb{P}\left(k \notin \mathcal{N}_{i,t-1}^{-}, \mbox{deg}_{k}^{+}(t) < L_{\theta_{k}}^{*}(0)\right) + \sum_{k = i+1}^{t} \, \mathbb{P}\left(k \notin \mathcal{N}_{i,t-1}^{-}, \mbox{deg}_{k}^{+}(t) < L_{\theta_{k}}^{*}(0)\right)\right)  \nonumber \\
&\overset{\text{(f)}}= \frac{p_{\theta_{i}}}{t-1} \left(\sum_{k = 1}^{i-1} \left(1-I_{1-p_{\theta_{i}}}\left(L_{\theta_{i}}^{*}(0), t-k-L_{\theta_{i}}^{*}(0)+1\right)\right) \, \prod_{w = i}^{t} \left(1-\frac{1}{w-1}\right) + \nonumber \right. \\
& \left. \sum_{k = i+1}^{t} \left(1-I_{1-p_{\theta_{i}}}\left(L_{\theta_{i}}^{*}(0), t-k-L_{\theta_{i}}^{*}(0)+1\right)\right) \, \prod_{w = k}^{t} \left(1-\frac{1}{w-1}\right)\right) \nonumber \\
&\overset{\text{(g)}}\approx  \frac{p_{\theta_{i}}}{t-1} \left(\sum_{k = 1}^{i-1} \left(1-I_{1-p_{\theta_{i}}}\left(L_{\theta_{i}}^{*}(0), t-k-L_{\theta_{i}}^{*}(0)+1\right)\right) \, \prod_{w = i}^{t} \left(1-\frac{1}{w-1}\right) + \nonumber \right. \\
& \left. \sum_{k = i+1}^{t} \left(1-I_{1-p_{\theta_{i}}}\left(L_{\theta_{i}}^{*}(0), t-k-L_{\theta_{i}}^{*}(0)+1\right)\right) \, e^{-\sum_{w = k}^{t}\frac{1}{w-1}}\right) \nonumber \\
&\overset{\text{(h)}}\asymp \frac{p_{\theta_{i}}}{t-1} \sum_{k = i+1}^{t} \left(1-I_{1-p_{\theta_{i}}}\left(L_{\theta_{i}}^{*}(0), t-k-L_{\theta_{i}}^{*}(0)+1\right)\right) \, e^{-\sum_{w = k}^{t}\frac{1}{w-1}}  \nonumber \\
&\overset{\text{(i)}}= \frac{p_{\theta_{i}}}{t-1} \left(\sum_{k = t-L_{\theta_{i}}^{*}(0)+1}^{t} e^{-\sum_{w = k}^{t}\frac{1}{w-1}} + \sum_{k = i+1}^{t-L_{\theta_{i}}^{*}(0)} \left(1-I_{1-p_{\theta_{i}}}\left(L_{\theta_{i}}^{*}(0), t-k-L_{\theta_{i}}^{*}(0)+1\right)\right) \, e^{-\sum_{w = k}^{t}\frac{1}{w-1}}\right) \nonumber \\
&\overset{\text{(j)}}\approx \frac{p_{\theta_{i}}}{t-1} \left(L_{\theta_{i}}^{*}(0) + \sum_{k = i+1}^{t-L_{\theta_{i}}^{*}(0)} \left(1-I_{1-p_{\theta_{i}}}\left(L_{\theta_{i}}^{*}(0), t-k-L_{\theta_{i}}^{*}(0)+1\right)\right) \, e^{-\sum_{w = k}^{t}\frac{1}{w-1}}\right) \nonumber \\
&\overset{\text{(k)}}\approx \frac{p_{\theta_{i}}}{t-1} \left(L_{\theta_{i}}^{*}(0) + \sum_{k = i+1}^{t-L_{\theta_{i}}^{*}(0)} (1-p_{\theta_{i}})^{k}\right) \nonumber \\
&\overset{\text{(l)}}\approx \frac{L_{\theta_{i}}^{*}(0)}{t}.
\label{pop75}
\end{align}
\setcounter{equation}{\value{mytempeqncnt}+1} \hrulefill{}\vspace*{4pt}
\end{figure*}

Based on (\ref{pop75}), the popularity of an agent $i$ at time $t$ is given by
\begin{align}
\mathbb{E}\left\{\mbox{deg}_{i}^{-}(t)\right\} &= \sum_{j=i}^{t}\mathbb{E}\left\{\Delta \mbox{deg}_{i}^{-}(j)\right\} \nonumber \\
&= \sum_{j=i}^{t} \frac{L_{\theta_{i}}^{*}(0)}{j} \nonumber \\
&= L_{\theta_{i}}^{*}(0) \left(H_{t}-H_{i-1}\right) \nonumber \\
&\asymp L_{\theta_{i}}^{*}(0) \log\left(\frac{t}{i-1}\right). 
\label{pop789}
\end{align} 

\begin{figure*}[!t]
\setcounter{mytempeqncnt}{\value{equation}} \setcounter{equation}{5}
\[\mathbb{E}\left\{\left.\Delta \mbox{deg}_{i}^{-}(t)\right|\mbox{deg}_{i}^{-}(t-1), \mbox{deg}_{i}^{-}(t-2), .\,.\,., \mbox{deg}_{i}^{-}(i)\right\}\]
\begin{align}
&\overset{\text{(a)}}= \sum_{k=1}^{t} \sum_{\theta_{k} \in \Theta} p_{\theta_{k}} \left(\mathbb{P}\left(m_{k}(t) = i \left|k \notin \mathcal{N}_{i,t-1}^{-}, k \in \bigcup_{j \in \mathcal{N}_{i,t-1}^{-}} \mathcal{N}_{j,t-1}^{-}, K_{k}(t), 0<\mbox{deg}_{k}^{-}(t) < L_{\theta_{k}}^{*}(0), \theta_{k}=\theta_{i}\right.\right)\right.   \nonumber \\
&\times \left. \mathbb{P}\left(k \notin \mathcal{N}_{i,t-1}^{-}, k \in \bigcup_{j \in \mathcal{N}_{i,t-1}^{-}} \mathcal{N}_{j,t-1}^{-}, K_{k}(t), 0<\mbox{deg}_{k}^{-}(t) < L_{\theta_{k}}^{*}(0),\theta_{k}=\theta_{i}\right)+ \frac{1}{t-1}\,\mathbb{P}\left(\mbox{deg}_{k}^{-}(t) = 0,\theta_{k}=\theta_{i}\right)\right)  \nonumber \\
&\overset{\text{(b)}}= p_{\theta_{i}} \sum_{k =1}^{t} \left(\mathbb{P}\left(m_{k}(t) = i \left|k \notin \mathcal{N}_{i,t-1}^{-}, k \in \bigcup_{j \in \mathcal{N}_{i,t-1}^{-}} \mathcal{N}_{j,t-1}^{-}, K_{k}(t), 0<\mbox{deg}_{k}^{-}(t) < L_{\theta_{k}}^{*}(0)\right.\right)\right. \nonumber \\
&\times \left. \mathbb{P}\left(k \notin \mathcal{N}_{i,t-1}^{-}, k \in \bigcup_{j \in \mathcal{N}_{i,t-1}^{-}} \mathcal{N}_{j,t-1}^{-},K_{k}(t), 0<\mbox{deg}_{k}^{-}(t) < L_{\theta_{k}}^{*}(0)\right)+ \frac{1}{t-1}\,\mathbb{P}\left(\mbox{deg}_{k}^{-}(t) = 0\right)\right) \nonumber \\
&\overset{\text{(c)}}= p_{\theta_{i}} \sum_{k=1}^{t} \, \frac{1}{K_{k}(t)}\, \mathbb{P}\left(k \notin \mathcal{N}_{i,t-1}^{-}, k \in \bigcup_{j \in \mathcal{N}_{i,t-1}^{-}} \mathcal{N}_{j,t-1}^{-},K_{k}(t), 0 < \mbox{deg}_{k}^{-}(t) < L_{\theta_{k}}^{*}(0)\right) + \frac{(1-p_{\theta_{i}})^{t-k}}{t-1} \nonumber \\
&\overset{\text{(d)}}= \frac{1-(1-p_{\theta_{i}})^{t}}{t-1} + p_{\theta_{i}} \sum_{k=1}^{t} \, \frac{1}{K_{k}(t)}\, \mathbb{P}\left(k \notin \mathcal{N}_{i,t-1}^{-}, k \in \bigcup_{j \in \mathcal{N}_{i,t-1}^{-}} \mathcal{N}_{j,t-1}^{-},K_{k}(t), 0 < \mbox{deg}_{k}^{-}(t) < L_{\theta_{k}}^{*}(0)\right) \nonumber \\
&\overset{\text{(e)}}\geq \frac{1-(1-p_{\theta_{i}})^{t}}{t-1} + p_{\theta_{i}} \sum_{k=i}^{t}\, \frac{1}{K_{k}(t)}\, \frac{\mbox{deg}_{i}^{-}(k)}{k-1} \, \mathbb{P}\left(k \notin \mathcal{N}_{i,t-1}^{-}, K_{k}(t), 0 < \mbox{deg}_{k}^{-}(t) < L_{\theta_{k}}^{*}(0)\right) \nonumber \\
&\overset{\text{(f)}}\approx \frac{1-(1-p_{\theta_{i}})^{t}}{t-1} + p_{\theta_{i}} \sum_{k=t-\lfloor \frac{L_{\theta_{i}}^{*}(0)}{p_{\theta_{i}}} \rfloor + 1}^{t}\, \frac{1}{K_{k}(t)}\, \frac{\mbox{deg}_{i}^{-}(k)}{k-1} \, \mathbb{P}\left(k \notin \mathcal{N}_{i,t-1}^{-}, K_{k}(t), 0 < \mbox{deg}_{k}^{-}(t) < L_{\theta_{k}}^{*}(0)\right) \nonumber \\
&\overset{\text{(g)}}\geq \frac{1-(1-p_{\theta_{i}})^{t}}{t-1} + \frac{\mbox{deg}_{i}^{-}\left(t-\lfloor \frac{L_{\theta_{i}}^{*}(0)}{p_{\theta_{i}}} \rfloor + 1 \right)}{t-\lfloor \frac{L_{\theta_{i}}^{*}(0)}{p_{\theta_{i}}} \rfloor} \, \sum_{m=0}^{L_{\theta_{i}}^{*}(0)-1} \frac{L_{\theta_{i}}^{*}(0)}{(m+1)L_{\theta_{i}}^{*}(0)-m} \prod_{v=1}^{m}\left(1-\frac{1}{v L_{\theta_{i}}^{*}(0) - (v-1)}\right) \nonumber \\
&\overset{\text{(h)}}\geq \frac{1-(1-p_{\theta_{i}})^{t}}{t-1} + \frac{\mbox{deg}_{i}^{-}\left(t\right)}{t-1} \sum_{m=0}^{L_{\theta_{i}}^{*}(0)-1} \frac{L_{\theta_{i}}^{*}(0)}{(m+1)L_{\theta_{i}}^{*}(0)-m} \prod_{v=1}^{m}\left(1-\frac{1}{v L_{\theta_{i}}^{*}(0) - (v-1)}\right) \nonumber \\
&\overset{\text{(i)}}\asymp \frac{1}{t} + \frac{\mbox{deg}_{i}^{-}\left(t\right)}{t} \sum_{m=0}^{L_{\theta_{i}}^{*}(0)-1} \frac{L_{\theta_{i}}^{*}(0)}{(m+1)L_{\theta_{i}}^{*}(0)-m} \prod_{v=1}^{m}\left(1-\frac{1}{v L_{\theta_{i}}^{*}(0) - (v-1)}\right).
\label{pop1789}
\end{align}
\setcounter{equation}{\value{mytempeqncnt}+1} \hrulefill{}\vspace*{4pt}
\end{figure*}

Next, we evaluate the popularity growth rate in an intolerant society with fully opportunistic agents, i.e. a society with $h_{k} = 1, \gamma_{k} = 1, \forall k \in \Theta$. Similar to (\ref{pop75}), we start by evaluating $\sum_{G^{t} \in \mathcal{G}^{t}}\mathbb{E}\left\{\left.\Delta \mbox{deg}_{i}^{-}(t)\right|G^{t}\right\} \mathbb{P}\left(G^{t}\right)$. Since in an opportunistic society the meeting process depends on the network structure, we start by evaluating the term $\mathbb{E}\left\{\left.\Delta \mbox{deg}_{i}^{-}(t)\right|G^{t}\right\}$ in (\ref{pop1789}). The steps involved in (\ref{pop1789}) are similar to those in (\ref{pop75}).

Following the same steps in Appendix F, it can be easily shown that
\[\mathbb{E}\left\{\mbox{deg}_{i}^{-}(t)\right\} \geq \frac{1}{b_{\theta_{i}}} \left(\frac{t}{i}\right)^{b_{\theta_{i}}}-\frac{1}{b_{\theta_{i}}},\]
where $b_{k} > b_{m}$ if $L_{k}^{*}(0) > L_{m}^{*}(0)$. 			

\section{Proof of Corollary 6}
\renewcommand{\theequation}{\thesection.\arabic{equation}}
\setcounter{mytempeqncnt}{\value{equation}} \setcounter{equation}{0}
For $\gamma_{m} = \gamma_{k} = 0,$ we have that  			
\[\mathbb{E}\left\{\mbox{deg}_{i}^{-}(t)\right\} = \log\left(\frac{L_{k}^{*}(0)}{i-1}\right),\]
and			
\[\mathbb{E}\left\{\mbox{deg}_{j}^{-}(t)\right\} = \log\left(\frac{L_{m}^{*}(0)}{j-1}\right).\]			
Thus, if $i < j$, then $\mathbb{E}\left\{\mbox{deg}_{i}^{-}(t)\right\} \geq \mathbb{E}\left\{\mbox{deg}_{j}^{-}(t)\right\}, \forall t \geq i,$ whereas if $i > j$, then $\mathbb{E}\left\{\mbox{deg}_{i}^{-}(t)\right\} \geq \mathbb{E}\left\{\mbox{deg}_{j}^{-}(t)\right\}, \forall t \geq T^{*},$ where 
\[T^{*} = \frac{\left(i-1\right)^{\frac{L_{k}^{*}(0)}{L_{k}^{*}(0)-L_{m}^{*}(0)}}}{\left(j-1\right)^{\frac{L_{m}^{*}(0)}{L_{k}^{*}(0)-L_{m}^{*}(0)}}}.\]
The same conclusion can be reached for the case of $\gamma_{m} = \gamma_{k} = 1$. The second part of the Corollary follows from the fact that popularity grows logarithmically in non-opportunistic societies, whereas it grows at least sublinearly in time for opportunistic societies, thus there always exists a finite time after which the expected popularity of an agent in an opportunistic society exceeds that of an agent in a non-opportunistic society.       

\section{Proof of Theorem 5}
\renewcommand{\theequation}{\thesection.\arabic{equation}}
\setcounter{mytempeqncnt}{\value{equation}} \setcounter{equation}{0}
We assume a mean-field approximation for the popularity growth process and consider that an agent's indegree is deterministic and is given by the expected indegree of that agent. In this case, the cdf of the popularity of type-$k$ agents at time $t$, denoted by $F^{t,k}_{d}(d)$, can be computed as follows
\begin{align}
F^{t,k}_{d}(d) &= 1 - \frac{\left|\left\{i \leq t \left| \mathbb{E}\left\{\mbox{deg}_{i}^{-}(t)\right\} \geq d, \theta_{i} = k\right.\right\}\right|}{\left|\mathcal{V}^{t}_{k}\right|} \nonumber \\
&= 1 - \frac{i^{*}(d)}{t},
\end{align} 
where $i^{*}(d)$ corresponds to the agent's birth date that solves the equation $\mathbb{E}\left\{\mbox{deg}_{i}^{-}(t)\right\} = d$.  

Now we focus on the case where $\gamma_{m} = \gamma_{k} = 0$, and $L^{*}_{k}(0) > L^{*}_{m}(0)$. In this case, we have that
\[F^{t,k}_{d}(d) = 1-e^{\frac{-d}{L^{*}_{k}(0)}},\] 
whereas
\[F^{t,m}_{d}(d) = 1-e^{\frac{-d}{L^{*}_{m}(0)}}.\]
Thus, $F^{t,m}_{d}(d) \geq F^{t,k}_{d}(d), \forall d$. The same approach can be used to prove the second part of the Theorem.

\section{Proof of Theorem 6}
\renewcommand{\theequation}{\thesection.\arabic{equation}}
\setcounter{mytempeqncnt}{\value{equation}} \setcounter{equation}{0}									
We start by showing that if there exists one type of agents $k \in \Theta$ for which $h_{k} < 1$ and $\gamma_{k} < 1$, then the network is connected almost surely. Assume that at any point of time, say time step $\tau$, the network has 2 disconnected components $\mathcal{C}_{1}$ and $\mathcal{C}_{2}$, where $\mathcal{C}_{i}$ is the set of agents in component $i$. We show that these two components will connect almost surely as the network grows. Assume that $E_{12}^{t}$ is the event that an agent of type $k$ who is attached to component $\mathcal{C}_{1}$ meets a stranger who belongs to type $\mathcal{C}_{2}$ and links to it, and $E_{21}^{t}$ is the event that an agent of type $k$ who is attached to component $\mathcal{C}_{2}$ meets a stranger who belongs to type $\mathcal{C}_{1}$ and links to it. It is clear that $\mathbb{P}\left(E_{12}^{t} \vee E_{21}^{t}\right) > 0, \forall t > \tau$. For instance, it is easy to show that $\mathbb{P}\left(E_{12}^{t}\right) > p_{k}(1-\gamma_{k})\frac{\left|\mathcal{C}_{1}\right|}{\left|\mathcal{C}_{1}\right|+\left|\mathcal{C}_{2}\right|}$. Thus, $\mathbb{P}\left(\bigvee_{t=\tau}^{\infty} \left(E_{12}^{t} \vee E_{21}^{t}\right)\right) = 1$, and any two disconnected components in the network will eventually get connected through a type-$k$ agent.

The converse follows from the fact that if the network is connected, then there exists cross-type links and we cannot have $h_{k} = 1, \forall k \in \Theta$. Moreover, if there exists one type of agents $k \in \Theta$ for which $h_{k} < 1$ while all other types are extremely homophilic, then we must have $\gamma_{k} < 1$ or otherwise type-$k$ agents will never meet strangers and will restrict their links to followees of followees who belong to the same network component, and the network will display $|\Theta|-1$ disconnected components.

\end{document}